\documentclass[12pt,reqno]{amsart}

\usepackage{latexsym}
\usepackage{amsmath}
\usepackage{amscd}
\usepackage{amssymb}
\usepackage{dsfont}
\usepackage{braket}
\usepackage{bbold}

\usepackage{mathrsfs}
\usepackage{comment}
\usepackage{color}
\usepackage{enumerate}
\usepackage{soul}
\usepackage{framed}

\usepackage{graphicx}
\usepackage{subfigure}
\usepackage[colorlinks]{hyperref}
\usepackage{url}

\newtheorem{theorem}{Theorem}[section]

\newtheorem{definition}{Definition}[section]

\newtheorem{proposition}[theorem]{Proposition}
\newtheorem{remark}[theorem]{Remark}

\def\ad{\text{ad}}

\numberwithin{equation}{section}

\begin{document}

\title[Networks of coadjoint orbits]{ Networks of coadjoint orbits:\\ from geometric to statistical mechanics\\ 
\vspace{10mm} {\it T\MakeLowercase{o} D\MakeLowercase{arryl} H\MakeLowercase{olm, for his 70th birthday}}}

\author[A. Arnaudon]{Alexis Arnaudon}
\author[S. Takao]{So Takao}
\address{AA, ST: Department of Mathematics, Imperial College, London SW7 2AZ, UK}

\maketitle

\begin{abstract}
A class of network models with symmetry group $G$ that evolve as a Lie-Poisson system is derived from the framework of geometric mechanics, which generalises the classical Heisenberg model studied in statistical mechanics.
We considered two ways of coupling the spins: one via the momentum and the other via the position and studied in details the equilibrium solutions and their corresponding nonlinear stability properties using the energy-Casimir method.
We then took the example $G=SO(3)$ and saw that the momentum-coupled system reduces to the classical Heisenberg model with massive spins and the position-coupled case reduces to a new system that has a broken symmetry group $SO(3)/SO(2)$ similar to the heavy top. In the latter system, we numerically observed an interesting synchronisation-like phenomenon for a certain class of initial conditions.
Adding a type of noise and dissipation that preserves the coadjoint orbit of the network model, we found that the invariant measure is given by the Gibbs measure, from which the notion of temperature is defined. 
We then observed a surprising `triple-humped' phase transition in the heavy top-like lattice model, where the spins switched from one equilibrium position to another before losing magnetisation as we increased the temperature.
This work is only a first step towards connecting geometric mechanics with statistical mechanics and several interesting problems are open for further investigation.

\end{abstract}

\setcounter{tocdepth}{1}
\tableofcontents

\section{Introduction}
The main purpose of this paper is to provide a link between geometric mechanics and statistical mechanics following recent advances in stochastic geometric mechanics \cite{lazaro2008stochastic,arnaudon2014stochastic,holm2015variational,arnaudon2016noise,arnaudon2016stochastic}, where the inclusion of a structure-preserving noise in Lagrangian or Hamiltonian systems have been considered.
Geometric mechanics provides a mathematical framework for describing Hamiltonian mechanical systems, but with the addition of a structure-preserving noise and dissipation, a theory of statistical geometric mechanics becomes possible, as suggested in \cite{arnaudon2016noise}. 
Here, we take a step further to include lattices, or more generally, networks of these systems and study their stability and phase transitions, similar to the well-known Ising or Heisenberg models. 

The Ising model is one the most well-studied systems in statistical mechanics, which is a simplified model of ferromagnetism. Its significance in statistical mechanics comes from the fact that it is the simplest model that exhibits {\em phase transition}, which is ubiquitous in the theory of matter, such as in the transition from liquids to gases or in the magnetisation of materials. See for example Chandler \cite{chandler1987introduction} for a classic reference in statistical mechanics. Relaxing the discrete nature of the Ising model to a system of interacting unit vectors in $\mathbb R^3$ on a lattice, we obtain the classical Heisenberg model, which, again is a simple model of ferromagnetism but now with a continuous symmetry group $SO(3)$, instead of the discrete $C_2$-symmetry of the Ising model. This model exhibits a phase transition, albeit one with a different universality class due to the difference in the symmetry group.

On the other hand, geometric mechanics is a mathematical discipline that deals with the dynamics of a `few-body' mechanical system with symmetry and mainly concerned with issues such as symmetry reduction, integrability and stability of equilibria. For basic references in geometric mechanics, see for example Arnold \cite{arnold89mechanics}, Marsden and Ratiu \cite{marsden1999book} and Holm \cite{holm2008geometric}. Although the two fields are very different, recent progress in the theory of stochastic geometric mechanics has shed some light to bridge the gap between the two. Initiated in the work by Bismut \cite{bismut1982mecanique} for canonical systems and extended to general Hamiltonian systems in L\'azaro-Cam\'i and Ortega \cite{lazaro2008stochastic}, stochastic geometric mechanics have found applications in areas as vast as simple mechanical systems \cite{arnaudon2016noise2}, fluid dynamics \cite{holm2015variational}, geometric integrators \cite{bourabee2009stochastic}, stability theory \cite{arnaudon2017stability} and shape analysis \cite{arnaudon2016stochastic}.
In particular, Holm \cite{holm2015variational} introduced stochastic processes at the level of the variational principle for fluid systems in such a way that the momentum map is preserved, and this introduced the idea of {\em structure preserving noise}, where the added noise does not destroy the essential geometric feature of the system.
In a similar spirit, adding a dissipative mechanism in Hamiltonian systems that preserve the basic geometric structure has been considered by Bloch et al. \cite{bloch1996euler}, where a type of dissipative bracket that dissipates energy while preserving the coadjoint orbit is introduced. A similar type of bracket that dissipates a chosen conserved quantity while preserving another, called the {\em selective decay mechanism}, is considered in Gay-Balmaz and Holm \cite{gaybalmaz2013selective,gaybalmaz2014geometric}.

Our work is based on Arnaudon et al. \cite{arnaudon2016noise} where both noise and dissipation is introduced in a general finite-dimensional Hamiltonian system with configuration group $G$, phase space $T^*G$ and symmetry group $G$, in such a way that it preserves the coadjoint orbits, where the reduced dynamics take place. A remarkable consequence of adding noise and dissipation in this way is that, not only is the geometric structure of the equation preserved but also the invariant measure for the stochastic dissipative system is {\em exactly given by the Gibbs measure} restricted to the coadjoint orbit, provided that we choose an isotropic noise.  The emergence of the Gibbs measure in this system is key to connecting our theory with statistical mechanics, where it arises naturally as the invariant measure in a canonical ensemble, which is a system of particles in statistical equilibrium that is coupled to an external heat bath. This strongly suggests that we can approach this system from the point of view of statistical mechanics and vice versa.

 Our first goal of this paper is to construct a simple lattice model that has a general symmetry Lie group $G$, or more generally on a network, using symmetry reduction techniques in geometric mechanics and study their deterministic properties. By considering a $G$-invariant canonical Hamiltonian system on phase space $T^*G$ at each node of a given graph, we regard the momentum map $J: T^*G \rightarrow \mathfrak g^*$ as the {\em spins} in our model, which interact with other spins according to the structure of the graph, where the interaction takes place if they are connected by an edge. The coupling between the spins can be taken in two ways: (1) at the reduced space $\mathfrak g^*$ or (2) directly on the group $G$. In the former, which is the easier case, we introduce the coupling between two neighbouring spins in the {\em reduced} Hamiltonian {\em after} performing symmetry reduction. We call this approach {\em momentum coupling} since the variables that are coupled is the momentum map.
We will see that this approach generalises the classical Heisenberg model to include general symmetry groups. For the latter, we consider a representation of $G$ on a vector space $V$ at each node and couple the neighbours in the {\em unreduced} Lagrangian that depend on vectors in $V^*$ that are acted on by $G$. Since the neighbours are coupled with the `positions' of a given vector in $V^*$ that is rotated around by the $G$-action, we call this approach {\em position coupling}. The vector that is taken here for position coupling breaks the full $G$-symmetry of our system and hence we apply the semi-direct product reduction theorem given in Holm et al. \cite{holm1998euler} to obtain the corresponding Lie-Poisson system on the semi-direct product Lie algebra $\mathfrak g^* \,\circledS \, V^*$ at each node. This gives rise to a new system that finds no analogues in classical lattice models. In the special case where the Lie algebra $\mathfrak g = \text{Lie}(G)$ is compact and semi-simple, we were able to obtain the equilibrium solutions for both the momentum-coupled and position-coupled systems as the eigenvectors of a generalised graph Laplacian, which we construct from the underlying graph. Furthermore, we were able to classify these stationary eigenvectors into ferromagnetic and anti-ferromagnetic states, which is consistent with known equilibrium solutions of the Ising model or the Heisenberg model and we also investigated their corresponding stability properties using the Energy-Casimir method. Relaxing the condition that each spin must have unit length, many of the equilibrium solutions that we find here are new. 

Our second goal is to investigate the statistical mechanical properties of the lattice model constructed above with noise and dissipation of the type considered in Arnaudon et al. \cite{arnaudon2016noise}. In particular, we study the phase transition exhibited in the two examples that we consider here: the rigid body lattice and the heavy top lattice which are the simplest systems that can be obtained via momentum-coupling and position-coupling respectively, by taking $G=SO(3)$. 
The invariant measure is found to be the Gibbs measure restricted to the coadjoint orbit for both of these systems, which allows us to introduce the notion of temperature and hence study their respective temperature phase transition behaviours.
For the rigid body network, we observe a standard second-order phase transition in both the mean-field simulation and the direct simulation as expected. However, in the direct simulation of the heavy top network, we observe a new type of phase transition behaviour where the spins become aligned to two intermediate metastable states as we decrease the temperature before settling down to the lowest energy state, instead of jumping directly to the lowest energy state. This `triple-humped' phase transition is not captured in our mean-field simulation, which follows a standard 'single-humped' phase transition. 
In this study, we only looked at numerical simulations of the phase transition but much work needs to be done in the analysis to understand this phenomenon.

\subsection{Structure of the paper}

In section \ref{background}, we begin with a quick review of statistical mechanics and stochastic geometric mechanics and study the statistics of a single coadjoint orbit system in section \ref{single}. Then, we introduce the notion of momentum coupling in section \ref{SS-section} to construct a lattice model with symmetry group $G$ and study in detail the equilibrium solutions and their stability properties. We will then add noise and dissipation to the system and show that the invariant measure is given by the Gibbs measure.
In section \ref{section-RB}, we study an example of this system by taking $G=SO(3)$, which we call the rigid body network.
In section \ref{SD-section}, we introduce the idea of position-coupling and derive the corresponding network Lie-Poisson model using semi-direct product reduction. Noise and dissipation are then added to the system in a similar way. This theory is then illustrated concretely in section \ref{section-HT} by considering the case $G=SO(3)$, which we call the heavy top network. In section \ref{PT-section}, we study the mean-field approximations and the phase transition behaviours of the rigid body network and the heavy top network. Finally, we give a conclusion and discuss further work in section \ref{conclusion}.

\subsection{List of symbols}
This work contains notations from three topics: geometric mechanics, graph theory and stochastic analysis. For clarity, we summarize the main notations we will be using here.  
\begin{itemize}
  \item[] $G\qquad $  a Lie group
  \item[] $g\qquad $  an element in a Lie group
\item[] $\mathfrak g\qquad $  Lie algebra corresponding to $G$
  \item[] $\mu\qquad $  an element in $\mathfrak g^*$
\item[] $\xi\qquad $  an element in $\mathfrak g$
  \item[] $k\qquad $  dimension of $\mathfrak g$
  \item[] $\mathcal O\qquad $ coadjoint orbit of $G$ in $\mathfrak g^*$
  \item[] $\mathcal N\qquad $  connected, undirected graph
  \item[] $N\qquad $  number of vertices in $\mathcal N$
\item[] $d_i\qquad $  number of edges stemming from vertex $i$ in $\mathcal N$
\item[] $A\qquad $  adjacency matrix
\item[] $D\qquad $  degree matrix
  \item[] $L_0\qquad $ the graph Laplacian
  \item[] $L\qquad $ the symmetric normalised graph Laplacian
\item[] $\mathbb A\qquad $  the extended adjacency matrix
\item[] $\mathbb D\qquad $  the extended degree matrix
  \item[] $\mathbb L\qquad $ the extended normalised Laplacian
\item[] $\mathbb 1\qquad $  the identity matrix
\item[] $\mathbb d\qquad $  stochastic time increment 
\item[] $\mathbb P\qquad $  probability distribution
\item[] $\Braket {\cdot }\qquad $ averaging 
\item[] $dW_t\qquad\hspace{-0.17in} $  standard Wiener process
\item[] $\sigma\qquad $ noise amplitude
\item[] $\theta \qquad $ dissipation amplitude
\item[] $\beta\qquad $ inverse temperature
\end{itemize}

\section{Background} \label{background}

\subsection{Statistical mechanics}

Statistical mechanics is a physical theory used to describe the macroscopic properties, such as temperature and entropy of a many-particle system (e.g. gas or lattice) evolving as a canonical Hamiltonian system that fluctuates around a mean state. This theory is useful for systems with a large degree of freedom where first of all, solving the full system is computationally expensive and second of all, the information that we need is independent of the motion of individual particles in the system. 
In equilibrium statistical mechanics, one deals with a system that is in {\em statistical equilibrium}, that is, the particles are distributed according to an invariant measure of the underlying system. The system of particles are often assumed to be in statistical equilibrium according to the following three thermodynamic ensembles

\begin{enumerate}
\item {\em Micro-canonical ensemble}: the system is isolated and there is no energy exchange or particle exchange with the exterior. Each state is equally probable to occur.
\item {\em Canonical ensemble}: the system is coupled to a heat bath of fixed temperature and fixed average energy and there is no particle exchange. The states are distributed according to a Gibbs measure.
\item {\em Grand-canonical ensemble}: the system is coupled to a heat bath of a fixed temperature and is subject to particle exchange.
\end{enumerate}

For systems with inter-particle interactions such as in lattice models, a phenomenon called {\em phase transition} is often observed as the temperature is varied, which describes an abrupt change of state in a system. A classic example of this is the transition from liquid to gas at the boiling point. Now, consider a simple 2-dimenstional lattice model with $N$ nodes that take the values {\em up} ($+1$) or {\em down} ($-1$) at each node. The energy of this system at a particular state is given by
\begin{align}
H = - \sum_{i \sim j}^N J_{ij} s_i s_j, \quad s_i, s_j \in \{1, -1\}\, ,
\end{align}
for constants $J_{ij}>0$ and $i \sim j$ means that nodes $i$ and $j$ are adjacent on the lattice. We denote by $\Omega := \{1, -1\}^N$ for the sample space of all possible configurations. Assuming that the system is in statistical equilibrium in a canonical ensemble with temperature $T$, the probability of selecting a configuration $\boldsymbol s = (s_1, \ldots, s_N) \in \Omega$ with a given energy level $H_0$ is given by the {\em Gibbs distribution}
\begin{align}
\mathbb P(X = \boldsymbol s \,|\,T) = Z_\beta^{-1} e^{-\beta H_0}\, , \quad H(\boldsymbol s) = H_0\, ,
\end{align}
where the normalising constant $Z_\beta$ is called the {\em partition function} and $\beta$ is the {\em inverse temperature}, defined by $\beta := \frac{1}{k_B T}$, where $k_B$ is the {\em Boltzmann constant}. It is well known that there exists some $T = T_c$, called the {\em critical temperature} such that the {\em average magnetisation} defined by
\begin{align}
\langle M \rangle = \frac1N \sum_{\boldsymbol s \in \Omega} \sum_{i=1}^N s_i \, \mathbb P(X = \boldsymbol s \,|\,T)
\end{align}
vanishes for $T > T_c$ and becomes strictly positive for $T < T_c$. This is an example of a so-called {\em second-order phase transition} and the model that we just considered is the {\em 2D Ising model} which is the simplest known system that admits a phase transition. Now, instead of taking the values $\{1, -1\}$, one can generalise this system so that the spins $s_i$ are unit vectors in $\mathbb R^3$ and we obtain the {\em classical Heisenberg model} which also exhibits a similar second-order phase transition only if the interaction is anisotropic, in which case it is often called the $XYZ$, or $XXZ$-model.

The ideas that we discussed here are elementary in statistical mechanics and can be found in most textbooks on the topic. Here, we will refer to Chandler \cite{chandler1987introduction} for a more thorough physical exploration of this topic and to \cite{ruelle2004thermodynamic} for a rigorous mathematical formulation of statistical mechanics.

\subsubsection{Other developments in geometric statistical mechanics}

Remarkably, in the early days of modern geometric mechanics, one of the founders of the field, Jean-Marie Souriau, had already attempted to apply ideas from geometric mechanics to statistical mechanics. 
His original idea was to impose natural symmetries on the Gibbs distribution and obtain a vector of inverse temperatures $\beta\in \mathfrak g$ in such a way that the symmetry Lie group $G$ corresponding to $\mathfrak g$ survives in the construction of statistical quantities, such as the entropy.
Some of his works in this direction can be found in \cite{souriau1966definition,souriau1974mecanique,souriau1969structure}.
However, this theory attracted very little attention at the time until only recently, where several authors modernized the original work of Souriau. 
We refer to \cite{barbaresco2014koszul,barbaresco2015symplectic,marle2016tools} as well as to \cite{barbaresco2016geometric} for a complete review on the history of Souriau's theory.

In parallel to this, there has also been recent work by Gay-Balmaz and Yoshimura \cite{gay2017lagrangianI,gay2017lagrangianII} based on infinite dimensional geometric mechanics and a certain class of dynamical constraints to describe various systems in the theory of non-equilibrium thermodynamics.

\subsection{Stochastic geometric mechanics}

A systematic theory for introducing stochastic processes into mechanics started with the work of Bismut \cite{bismut1982mecanique}, where the noise was introduced at the level of the variational principle and the corresponding stochastic Euler-Lagrange equations were derived. 
This theory has seen a resurgence recently and has been extended by several authors such as L\'azaro-Cam\'i and Ortega \cite{lazaro2008stochastic} and Bou-Rabee and Owhadi \cite{bourabee2009stochastic} based on more modern geometric mechanical techniques.
Also recently, Arnaudon, Chen and Cruzeiro \cite{arnaudon2014stochastic} and Holm \cite{holm2015variational} added noise in the variational principle that is compatible with symmetry reduction to obtain a corresponding stochastic Euler-Poincar\'e equation.
In the former, the noise was introduced at the level of the Lie group and using symmetry reduction together with taking an expectation, they obtained a dissipative deterministic equation with applications in infinite dimensional systems, such as the Navier-Stokes equation in fluid dynamics. 
In the latter, which we will base our work on, the noise is introduced directly into the reconstruction relation in the variational principle and followed by symmetry reduction, a fully stochastic Euler-Poincar\'e equation is derived. 
Based on this work, Arnaudon, De Castro and Holm \cite{arnaudon2016stochastic} studied in detail the finite dimensional analogue of this construction together with double bracket dissipation, and the implication of noise on the nonlinear stability of relative equilibria was investigated in Arnaudon, Ganaba and Holm \cite{arnaudon2017stability}.
We refer to \cite{arnaudon2016stochastic} and references therein for more details on other related works. 

In the present text, we will use the finite dimensional equations derived in Arnaudon, De Castro and Holm \cite{arnaudon2016stochastic}, hence we will refer to this work for more details about the derivations of the stochastic equations used here. 

For a dynamical system with configuration group $G$, phase space $TG$ and a $G$-invariant Lagrangian $L(g, \dot g)$, where $(g, \dot g) \in TG$, we can apply the theory of reduction by symmetry to show that the Euler-Lagrange equation on $TG$ associated to this Lagrangian is equivalent to the Euler-Poincar\'e equation 
\begin{align}
    \frac{d}{dt}\frac{\partial l}{\partial \xi} + \mathrm{ad}^*_\xi \frac{\partial l}{\partial \xi}=0\, , 
\end{align}
 on $\mathfrak g$, where $\xi= g^{-1} \dot g\in \mathfrak g$, the reduced Lagrangian is $l(\xi) := L(g, \dot g)$ and $\mathrm{ad}^*$ is the dual of the adjoint action. 
From the definition of the reduced variable $\xi$, it is possible to reconstruct the solution on the original configuration manifold $G$, using the formula
\begin{align}
    \dot g = g\xi\, , 
\end{align}
called the reconstruction relation. 
The introduction of noise appears at this level in the theory of reduction by symmetry by replacing the above relation by a stochastic process
\begin{align}
    \mathbb d g = g\xi \, dt + \sum_{l=1}^K g \sigma_l \circ dW_t^l\, , 
    \label{sto-rr}
\end{align}
where $W_t^l$ are $K$ independent standard Wiener processes, $\mathbb d$ is the stochastic evolution operator and $\sigma_l\in \mathfrak g$ are given Lie algebra elements which represents the amplitude and direction of the noise. 
We interpret this as a Stratonovich SDE, represented by the symbol $\circ$ so that we can use the standard rules of calculus. 

From the above stochastic reconstruction relation, one can compute the stochastic variations 
\begin{align}
    \delta \xi = \mathbb d\eta + \mathrm{ad}_\xi\eta \,dt + \sum_{l=1}^K \mathrm{ad}_{\sigma_l}\eta \circ dW_t^l\, ,     \quad \eta := g^{-1} \delta g\, , 
\end{align}
and the stochastic Euler-Poincar\'e equation will follow from the same variational problem exactly as in the deterministic case. 

As we mentioned in the introduction, the link between geometric mechanics and statistical mechanics requires a mechanism of dissipation to balance the energy input of the noise and hope for the system to reach a statistical equilibrium. 
As given in Arnaudon et al. \cite{arnaudon2016stochastic}, we use the double bracket dissipation introduced by Bloch et al, \cite{bloch1996euler} and extended in Gay-Balmaz and Holm \cite{gaybalmaz2013selective, gaybalmaz2014geometric} to obtain the following stochastic Lie-Poisson equation with dissipation 
\begin{align}        
    \mathbb d\mu &+ \mathrm{ad}^*_\frac{\partial h}{\partial \mu} \mu\, dt 
    + \theta\, \mathrm{ad}^*_\frac{\partial C}{\partial \mu} \left [ \frac{\partial C}{\partial \mu}, \frac{\partial h}{\partial \mu} \right ]^\flat \, dt + \sum_l\mathrm{ad}^*_{\sigma_l} \mu \circ dW_t^l = 0 \,,
    \label{SEP-Diss}
\end{align}
where $h:\mathfrak g^*\to \mathbb R$ is the reduced Hamiltonian, $C$ is a given Casimir function of the Lie-Poisson structure and the momentum variable $\mu\in \mathfrak g^*$ is conjugate to the reduced velocity $\xi \in \mathfrak g$ via the Legendre transform. 
We denote by $\flat:\mathfrak g \to \mathfrak g^*$ to be the canonical isomorphism between $\mathfrak g$ and its dual given an inner product, and $\theta \in \mathbb R$ parametrizes the strength of dissipation. 
This equation can be derived from a variational principle where the dissipation is inserted as a force in the Lagrange-d'Alembert framework, see \cite{arnaudon2016stochastic} for more details. 

One can see that the dissipative term and the noise preserves the coadjoint orbit
\begin{align}
    \mathcal O_\mu = \{ \mathrm{Ad}^*_g \mu, \quad  \forall g \in G\} \, , 
    \label{coadj-def}
\end{align}
where $\mu= \mu(0)$. 
In the following, we will denote a generic coadjoint orbit by $\mathcal O$, discarding the foot point. 
In particular, the Casimir functions $C:\mathfrak g^*\to \mathbb R$ (in general, the system admits several Casimir functions, but we only select one here) is conserved whereas the energy decays due to the dissipation and fluctuate due to the noise. 

The statistical information of the process can be captured via the Fokker-Planck equation, which describes the time evolution of the probability distribution of the process $\mathfrak \mu$. 
We refer to \cite{arnaudon2016stochastic} for the derivation of the Fokker-Planck equation for the process $\mathfrak \mu$ in geometric form, given by
\begin{align}
    \frac{d}{dt} \mathbb P(\mu) + \{h,\mathbb P\}  +\theta\, \left \langle\left [\frac{\partial \mathbb P}{\partial \mu}, \frac{\partial C}{\partial \mu}\right], \left [ \frac{\partial h}{\partial \mu}, \frac{\partial C}{\partial \mu}\right]^{\flat} \right \rangle - \frac12 \sum_l  \{\Phi_l,\{\Phi_l,\mathbb P\}\}=0\, ,
    \label{FP-Diss}
\end{align}
where $\langle \cdot, \cdot \rangle$ is the natural pairing on $\mathfrak g$, $\{f, h\}(\mu) =\left  \langle \mu, \left [ \frac{\partial f}{\partial \mu},  \frac{\partial h}{\partial \mu}\right ]\right \rangle $ is the Lie-Poisson bracket and $\Phi_l= \langle \sigma_l, \mu\rangle$ are the stochastic potentials.  
We now state the result in \cite{arnaudon2016stochastic} that forms the basis of our present work. 
\begin{theorem}\label{FP-diss-thm}
    The stationary distribution of the stochastic process \eqref{SEP-Diss} with an isotropic noise, that is $\sigma_i = \sigma e_i$ (for $e_i$ a basis of $\mathfrak g$), is the Gibbs measure on the coadjoint orbit, that is 
    \begin{align}
        \mathbb P_\infty(\mu) = Z^{-1} e^{-\beta h(\mu)}\, ,  
    \label{Gibbs-def}    
    \end{align}
    where $\beta= \frac{2\theta}{\sigma^2}= \frac{1}{k_B T}$ is the inverse temperature ($k_B$ is Boltzmann's constant) and $Z$ is the normalisation constant, or partition function 
    \begin{align}
        Z = \int_{\mathcal O}  e^{-\beta h(\mu)} d\mu \, . 
        \label{Z-def}
    \end{align}
\end{theorem}
\begin{proof}
One can prove this by a direct substitution into the Fokker-Planck equation \eqref{FP-Diss} and obtain $\frac{d}{dt} \mathbb P_\infty= 0$. 
\end{proof}

We need to stress an important point here.  
As opposed to the standard theory of statistical mechanics where the partition function is an integral over the full phase space, our partition function \eqref{Z-def} is defined as an integral over the coadjoint orbit and not over the whole space $\mathfrak g^*$. 
This is a result of the nonlinear dissipation and the multiplicative noise and makes the effective phase space compact if the Lie group is compact.
It also has the advantage of being an embedding in $\mathbb R^n$, rather than a general manifold where the definition of stochastic processes is more involved. 
Notice the requirement of isotropic noise, necessary to obtain the simple Gibbs form of the stationary distribution. 
In the non-isotropic case, the distribution would be close but different from the Gibbs distribution, see \cite{arnaudon2016stochastic} for more details. 

Hereafter, we will not use the Fokker-Planck equation as we will assume implicitly that the system is at statistical equilibrium, that is, our system is distributed according to the Gibbs measure. 
We can therefore apply results from the theory of equilibrium statistical mechanics to study our system. 
One of the fundamental results is the fluctuation-dissipation theorem, which states that for a system to be at a statistical equilibrium, the diffusion coefficient must be a function of the dissipation coefficient or vice versa. 
In our case, this relation is given by $\theta = \frac12 \beta \sigma^2$, and is often called {\em Einstein's relation}.
For our system, statistical equilibrium states exist for all pairs $(\theta,\sigma)$ but when the temperature is fixed, one of the variables depends on the other. 
We will not go further into this issue as it is rather simple as we are dealing with a classical system, but this theorem is important especially for quantum systems. 
We refer to Kubo \cite{kubo1966fluctuation} for an early theoretical exposition of this theorem, and the many references therein for more details on this topic.

\section{Statistics in a single coadjoint orbit}\label{single}

We illustrate our statistical analysis by first considering the statistics for a single coadjoint orbit.
The key object in equilibrium statistical mechanics is the partition function \eqref{Z-def} since most statistical quantities can be derived from it. 
Unfortunately, the partition function cannot be integrated analytically in general, except for a few simple examples.  

The average energy of the system is defined as
\begin{align}
    \Braket{E} := \int_{\mathcal O} h(\mu) \mathbb P_\infty (\mu) d\mu = Z^{-1}\int_\mathcal{O} h(\mu) e^{-\beta h(\mu)} d\mu\, ,
\end{align}
but can be obtained directly from the partition function as
\begin{align}
    \Braket{E} = - \frac{\partial}{\partial \beta} \mathrm{log}(Z)\, . 
\end{align}

In the case of compact coadjoint orbits, the average energy is bounded from above and saturtes rapidly. 
Indeed, as $T\to \infty$, the Gibbs distribution $\mathbb P_\infty(\mu) \to 1$  for all $\mu\in \mathcal O$ and we have the inequality
\begin{align}
    \Braket{E}= \int_\mathcal{O} h(\mu) \mathbb P_\infty(\mu) d\mu \leq \int_\mathcal{O} h(\mu) \, d\mu < \infty\, . 
\end{align}

Similarly, the energy fluctuation  around $\braket{E}$ is defined as 
\begin{align}
    \Delta_2 \Braket{E} := \braket{E^2}- \braket{E}^2 = \frac{\partial^2}{\partial \beta^2} \mathrm{log} Z= - \frac{\partial \Braket{E}}{\partial \beta}\, . 
\end{align}
As for the mean energy, it saturates rapidly as $T\to \infty$. 

We finally arrive at the definition of entropy $S$ on the coadjoint orbit, given by the famous relation
\begin{align}
    S := - k_B \int_\mathcal{O} \mathbb P_\infty(\mu) \mathrm{log} \left (\mathbb P_\infty(\mu)\right ) d\mu \, , 
\end{align}
or, in terms of the partition function,
\begin{align}
    S = k_B \frac{\partial}{\partial T} \left ( T\, \mathrm{log}(Z)\right )\, . 
    \label{S-Z}
\end{align}
If the orbit is compact, the entropy saturates as well for large $T$. 

The entropy is an important quantity as it is maximized by the Gibbs measure, under the constraint that the mean energy $\Braket{E}$ is fixed. 
This is, in a sense the definition of the canonical ensemble in statistical mechanics.  
\begin{theorem}
    The Gibbs distribution \eqref{Gibbs-def} is the distribution that maximizes the entropy \eqref{S-Z} of the system.     
    \label{S-thm}
\end{theorem}
\begin{proof}
We enforce $\Braket{E}=E_0$ with a Lagrange multiplier $\beta$ (this choice will be clear later) and consider the functional
\begin{align}
    \mathbb S(\mathbb P) = \int_\mathcal{O} \left [ \mathbb P(\mu) \mathrm{log}( \mathbb P(\mu)) + \beta \left ( h(\mu ) \mathbb P(\mu)  - E_0\right )\right ] d\mu\, . 
\end{align}
The variation with respect to $\mathbb P$ in the variational principle $\delta \mathbb S = 0$  gives the relation
\begin{align}
    \mathrm{log}\left (\mathbb P(\mu)\right) + 1 + \beta h(\mu)  = 0 \, . 
\end{align}
By normalizing the solution $\mathbb P(\mu)$ to $1$, we obtain the Gibbs measure \eqref{Gibbs-def} with inverse temperature $\beta$. 
\end{proof}

\section{Network of coadjoint orbits I: Momentum coupling}\label{SS-section}

The two important systems that can be described using the theory of statistical mechanics are gas and lattices. The latter, which includes the Ising model and the Heisenberg model are well studied in statistical physics and exhibit many interesting properties such as phase transition.
In this section, we will construct a lattice of interacting continuous spins, similar to the Heisenberg model, that take values on a coadjoint orbit, then study in details their deterministic properties and later introduce noise and dissipation to the system that preserve the coadjoint orbits.
By coupling the neighbours in the {\em momentum}, we will recover the classical Heisenberg model as the simplest example using the coadjoint orbit of $SO(3)$, that is the sphere $S^2$. 
We will mostly restrict our exposition to compact semi-simple Lie algebras here, but extensions to more general Lie algebras should be possible in some cases.

\subsection{Deterministic equations}

Our aim here is to construct the equations of motion of spins interacting on a general undirected connected network $\mathcal N$. 
At each node, we will assign a coadjoint orbit $\mathcal O_i$ for $i=1,\ldots, N$, associated to the Lie group $G$, which is coupled to its neighbours according to the structure of the graph.

\subsubsection{Coupling in the momentum}
Given an undirected, connected graph $\mathcal N$ and a left (or right) $G$-invariant canonical Hamiltonian system on $T^*G$, we introduce the notion of {\em coupling in the momentum} to construct a Lie-Poisson Hamiltonian system on $(\mathfrak g^*)^{\oplus N}$ such that at each node of $\mathcal N$, the system evolves as a Lie-Poisson system on $\mathfrak g^*$, with an additional interaction term arising from its neighbours. If we suppose for now that at each node, the system evolves independently as a canonical Hamiltonian system on $T^*G$ with left $G$-invariant Hamiltonians $H_i$ at node $i$, then according to the Lie-Poisson reduction theorem, we can collectivise the Hamiltonian as $H_i = h_i \circ \mathbf J_R$ where $\mathbf J_R : T^*G \rightarrow \mathfrak g^*$ is the momentum map corresponding to the right cotangent lift action of $G$ on $T^*G$ and we obtain a Lie-Poisson system on $\mathfrak g^*$ with reduced Hamiltonian $h_i$ at each node $i$. The idea of coupling in the momentum is to take the interaction between the neighbours at the level of the {\em reduced space} $\mathfrak g^*$, that is, we couple the momentum map $\mathbf J_R$ between neighbouring bodies. Introducing the symmetric and positive definite {\em interaction tensor}
\begin{align}
  \mathbb J_{ij} : \mathfrak g^* \to \mathfrak g\, , \quad \forall i,j = 1\, , \ldots, N\,,
\end{align}
we take the {\em momentum-coupled interaction potential energy} to be
\begin{align} \label{int-h}
h^{\text{int}}_{ij}(\mu_i, \mu_j) = -\frac{1}{2\sqrt{d_id_j}} \langle \mu_i, \mathbb J_{ij} \mu_j \rangle_{\mathfrak g^* \times \mathfrak g}, \quad \mu_i, \mu_j \in \mathfrak g^*\,,
\end{align}
where $d_i$ is the number of neighbours at node $i$ on graph $\mathcal N$ and the factor $\sqrt{d_id_j}^{-1}$ is taken to normalise the expression, which will be convenient later in our analysis. We take the total energy of the form
\begin{align}
h(\boldsymbol \mu) = \sum_{i=1}^N h_i(\mu_i) + \sum_{i=1}^N \sum_{j \sim i} h^{int}_{ij}(\mu_i, \mu_j), \quad \boldsymbol \mu = (\mu_1, \ldots, \mu_N) \in (\mathfrak g^*)^{\oplus N}\, , 
\end{align}
where $i \sim j$ means that nodes $i$ and $j$ are adjacent on the graph $\mathcal N$. Taking the $(-)$ Lie- Poisson structure on $(\mathfrak g^*)^{\oplus N}$ which is given by the sum
\begin{align} \label{LP-full}
\{f, g\}^-_{LP}(\boldsymbol \mu) = - \sum_{i=1}^N \left<\mu_i\, , \left[\frac{\delta f}{\delta \mu_i}, \frac{\delta g}{\delta \mu_i} \right] \right>_{\mathfrak g^* \times \mathfrak g}\, ,
\end{align}
we obtain the Lie-Poisson equation
\begin{align} \label{LP-eq-1}
\dot{\mu_i} = \text{ad}^*_{\delta h_i/\delta \mu_i} \mu_i - \sum_{j \sim i} \frac{1}{\sqrt{d_id_j}}\text{ad}^*_{\mathbb J_{ij}\mu_j} \mu_i, \quad \forall i = 1, \ldots, N\, .
\end{align}

\begin{remark}
To be more precise, the interaction tensor $\mathbb J_{ij}$ is a map from $\mathfrak g^*$ at site $i$ to $\mathfrak g$ at site $j$, but since the Lie algebra $\mathfrak g$ is identical at each node, there is no need to distinguish between them.
\end{remark}

\begin{remark}
If we consider a right $G$-invariant Hamiltonian instead, then the reduced variables are the left momentum map $\mathbf J_L$ and we take the $(+)$ Lie-Poisson structure. (See \cite{marsden1999book} for more details about the issue of left vs. right action.)
\end{remark}

For the case where $\mathfrak{g} = \text{Lie}(G)$ is a compact, semi-simple Lie algebra, we can identify $\mathfrak g^*$ with $\mathfrak g$ using the inner product $\langle \cdot, \cdot \rangle = - \kappa(\cdot, \cdot)$ where $\kappa$ is the Killing form, which also satisfies the associativity property $\langle a, [b, c] \rangle = \langle [a, b], c \rangle$. From this, one can check that the quadratic functions
\begin{align}
C_i (\mu_i) = \frac12 \langle \mu_i, \mu_i \rangle, \quad \forall i=1, \ldots, N,
\end{align}
are Casimirs of the Lie-Poisson bracket \eqref{LP-full} and the coadjoint orbits of $G^{\times N}$ on $\mathfrak (g^*)^{\oplus N}$ are contained in their level sets, i.e.
\begin{align}
\mathcal O \subset \left\{\boldsymbol \mu \in (\mathfrak g^*)^{\oplus N} : C_i(\mu_i) = c_i, \, i=1, \ldots, N\right\}\, ,
\end{align}
where $c_1, \ldots, c_N$ are constants, which follows from the $\text{Ad}^*$-invariance of the Killing form. The coadjoint orbits are preserved by the dynamics due to the equivariance of the momentum map $\mathbf J_R$.

We also consider
\begin{align}
  C (\boldsymbol \mu) = \sum_i C_i(\mu_i ) = \frac12 \langle \boldsymbol \mu, \boldsymbol \mu \rangle\,,
\end{align}
which is the sum of all the Casimirs. This is also a Casimir for this system and will be used in our analysis later. We will also use the following shorthand notation for a system of coadjoint orbits that are interconnected by a network.

\begin{definition}
Given a graph $\mathcal N$ with $N$ nodes, a Lie group $G$, and a coadjoint orbit $\mathcal O_{\boldsymbol \mu} = \left\{\text{\bf Ad}^*_{\boldsymbol g} {\boldsymbol \mu} : \boldsymbol g \in G^{\times N} \right\} \subset (\mathfrak g^*)^{\oplus N}$, where $\boldsymbol \mu \in (\mathfrak g^*)^{\oplus N}$ and $\text{\bf Ad}^*$ is the diagonal coadjoint action of $G^{\times N}$ on $(\mathfrak g^*)^{\oplus N}$, we call the triple $(\mathcal N, G, \mathcal O_{\boldsymbol \mu})$ a network of coadjoint orbits and equation \eqref{LP-eq-1} a network Lie-Poisson equation on $(\mathcal N, G, \mathcal O_{\boldsymbol \mu})$.
\end{definition}

Hereafter, we consider a simple mechanical system consisting only of a purely kinetic energy term
\begin{align}
  h^{KE}_i(\mu_i) = \frac12 \left \langle \mu_i, \mathbb I_i^{-1} \mu_i\right \rangle_{\mathfrak g^* \times \mathfrak g}\, , \qquad \forall i = 1, \ldots, N\,,  
\end{align}
where $\mathbb I_i:\mathfrak g\to \mathfrak g^*$ is the moment of inertia tensor assigned to node $i$, which is symmetric and positive definite, and an interaction potential energy term \eqref{int-h}. We take the total Hamiltonian to be
\begin{align} \label{full-h-1}
h(\boldsymbol \mu) &= \sum_i h^{KE}_i(\mu_i) + \sum_i \sum_{j \sim i} h^{\text{int}}_{ij}(\mu_i, \mu_j) \nonumber \\
&= \frac12 \sum_i \left \langle \mu_i, \mathbb I_i^{-1} \mu_i\right \rangle - \frac12 \sum_i \sum_{j \sim i} \frac{1}{\sqrt{d_id_j}} \langle \mu_i, \mathbb J_{ij} \mu_j \rangle.
\end{align}
In fact, this can be expressed in a more compact form using the language of graph theory, which we will review in the next section.

In the special case where the graph is regular, i.e. $d_1=\cdots=d_n=d$ and absorbing the $\frac1d$ facor in the interaction tensors $\mathbb J_{ij}$, we obtain the Hamiltonian  
\begin{align}
  h_\mathrm{reg} (\boldsymbol \mu) = \frac12 \sum_i \left \langle \mu_i, \mathbb I_i^{-1} \mu_i \right \rangle - \frac12 \sum_i\sum_{i\sim j} \left \langle \mu_i ,  \mathbb J_{ij} \mu_j \right\rangle \, .
\end{align}
which, ignoring the kinetic energy term, resembles the Hamiltonian for the Ising model or the classical Heisenberg model. Hence, the network Lie Poisson system we obtained by momentum coupling can be viewed as a generalisation of the Heisenberg model with spins taking values on a general coadjoint orbit and with an additional kinetic energy term at each node.

\begin{remark}[Ferromagnetic vs. anti-ferromagnetic states]
From this Hamiltonian, we see that the minimum energy configuration must be an aligned state, with all the $\mu_i$ pointing in the same direction, and the maximum energy state must be anti-aligned. 
In statistical mechanics, these two states are called ferromagnetic and anti-ferromagnetic states respectively and will be important in our discussion of equilibrium solutions later.
\end{remark}

\subsubsection{Review of graph theory}

The structure of a graph is described by the {\em adjacency matrix}, which for unweighted graphs, is given by 
\begin{align}
A_{ij}= 
  \begin{cases}
    1 & \mathrm{if }\quad   i\sim j\\
    0 & \mathrm{otherwise}\, , 
  \end{cases}
\end{align}
where $i\sim j$ means that the node $i$ and $j$ are adjacent, or share an edge on the graph. 
This matrix is symmetric if the graph is undirected, which is the case here. 
Each node has $d_i = (A \boldsymbol 1_N)_i$ neighbours, where $\boldsymbol 1_N = (1,\ldots 1)\in \mathbb R^N$, and this is called the {\em degree} at node $i$. 
From this, we define the {\em graph Laplacian} $L_0 = D-A$, where $D= \mathrm{diag}(d_1, \ldots, d_N)$ is the {\em degree matrix}, and its normalised version 
\begin{align}
  L= D^{-\frac12} L_0D^{-\frac12}= \mathrm{\mathbb 1}_{N} - D^{-\frac12} A D^{-\frac12} \,, 
\end{align}
which is a symmetric matrix. 
These Laplacians are usually used to define a random walk on a graph, where $\dot{\mathbf p} = \mathbf pL$ for a probability vector $\mathbf p\in \mathbb R^n$ of a random walker, which corresponds to a discrete version of the diffusion equation, where $L$ plays the role of the Laplace-de Rham operator.  

\subsubsection{Network of coadjoint orbits}

We will now extend this theory to write our system of interacting coadjoint orbits of dimension $k$ using the language of graph theory, in particular, using the graph Laplacian. 

We extend our notion of the normalised Laplacian by weighting its components by the inertia tensor and the interaction tensor to get an {\em extended normalised Laplacian}
\begin{align}
  \mathbb L &:= \overline {\mathbb I}^{-1} -\mathbb D^{-\frac12} \mathbb A \mathbb D^{-\frac12}\,\label{ext-Lapl} \\
  \mathrm{where }\quad \mathbb A_{ij} &= \mathbb J_{ij} A_{ij}, \quad \overline {\mathbb I}^{-1}= \mathrm{diag}(\mathbb I^{-1}_1, \ldots, \mathbb I^{-1}_n) \quad \mathrm{and} \quad \mathbb D = \mathrm{diag}(d_1 \mathbb 1_k, \ldots, d_n \mathbb{1}_k)\, . \nonumber 
\end{align}

The Hamiltonian \eqref{full-h-1} of our system can then be expressed compactly as
\begin{align} \label{simplified-h}
  h(\boldsymbol \mu ) =  \frac12 \left \langle \boldsymbol \mu , \mathbb L \boldsymbol \mu \right \rangle_{(\mathfrak g^*)^{\oplus N}\times \mathfrak g^{\oplus N}}\, .
\end{align}

Using this notation, the network Lie-Poisson equation \eqref{LP-eq-1} with Hamiltonian \eqref{simplified-h} can be written in the compact form
\begin{align}
\dot{\boldsymbol \mu} = \mathbf{ad}^*_{\mathbb L \boldsymbol \mu} \, \boldsymbol \mu \label{LP-network}\,, 
\end{align}
where $\mathbf{ ad}^*:   \mathfrak g^{\oplus N}\times (\mathfrak g^*)^{\oplus N}\to (\mathfrak g^*)^{\oplus n}$ is the diagonal coadjoint action $(\xi_1, \ldots, \xi_N) \times (\mu_1, \ldots, \mu_N) \mapsto (\text{ad}^*_{\xi_1} \mu_1, \ldots, \text{ad}^*_{\xi_N} \mu_N)$ where $\ad^*: \mathfrak g \times \mathfrak g^* \rightarrow \mathfrak g^*$ is the standard coadjoint action on $\mathfrak g^*$. 

\subsubsection{Legendre transformation and Lagrangian formulation}

We obtained equation \eqref{LP-network} by coupling the neighbours at the level of the reduced space {\em after} performing Lie-Poisson reduction at each node. It is then natural to ask whether we can recover the same equation by coupling the neighbours at the level of the group and then performing symmetry reduction. 
The answer is yes, except for a few cases where the graph Laplacian is not invertible on the full phase space. 

For $\mathbb L$ to be invertible, one has to ensure that $\text{ker}(\mathbb L) = \emptyset$. 
For the simple case where $\mathbb I_i = \mathbb I $ and $\mathbb J_{ij} = \mathbb J$, this can be reduced to showing that $\mathbb I$ and $\mathbb J$ has no eigenvalues in common.  
Indeed, if they share an eigenvalue, then at least one of the ferromagnetic states will have a $0$ eigenvalue since the extended Laplacian reduces to the standard graph Laplacian $L = D- A$ which has a single $0$ eigenvalue. 
However, choosing $\mathbb J< \mathbb I$, the minimum eigenvalue of $\mathbb L$ becomes strictly positive since the reduced $L$ corresponding to ferromagnetic states become strictly diagonally dominant. 

Assuming that the moment of inertia is invertible on the full space, we obtain the following reduced Lagrangian by the inverse Legendre transform
\begin{align}
  l(\boldsymbol \xi) = \frac12 \langle \boldsymbol \xi , \mathbb L^{-1} \boldsymbol \xi\rangle\, , \qquad  \boldsymbol \xi = \mathbb L \boldsymbol \mu \in \mathfrak g^{\oplus N},
\end{align}
which can also be obtained through the full Lagrangian on $TG^{\times N}$, given by
\begin{align}
L(g, \dot{g}) = \frac12 \left< \dot{\boldsymbol g}, \mathbb L^{-1}_{\boldsymbol g} \, \dot{\boldsymbol g} \right>_{T_{\boldsymbol g}(G^{\times N}) \times T^*_{\boldsymbol g}(G^{\times N})}, \label{Lag-full}
\end{align}
where we defined $\mathbb L^{-1}_{\boldsymbol g} := L^*_{{\boldsymbol g}^{-1}}\mathbb L^{-1}(L_{{\boldsymbol g}^{-1}})_* : T_{\boldsymbol g}(G^{\times N}) \to T^*_{\boldsymbol g}(G^{\times N})$ and $L_{\boldsymbol g}$ denotes left diagonal action of $\boldsymbol g$. We refer to appendix 2 of \cite{arnold89mechanics} for the reconstruction of the full Lagrangian from a reduced Lagrangian of this form. It is easy to check that the Lagrangian \eqref{Lag-full} is left invariant under $G$ so we can apply Euler-Poincar\'e reduction on \eqref{Lag-full} and perform Legendre transformation to obtain our network Lie-Poisson equation \eqref{LP-network}.
However, notice that the matrix $\mathbb L^{-1}$ is far from being sparse, thus the interaction between sites take a complicated nonlocal (more than neighbours interactions) form on the Lagrangian side. 

If one really wishes to interpret the system from the Lagrangian side with singular $\mathbb L$, the Lagrangian can be defined in terms of the Moore-Penrose pseudoinverse $\mathbb L^\dagger$. 
However by doing so, one has to restrict the reduced Lagrangian to the subspace $\text{ker}(\mathbb L^\dagger)^{\perp}$ and recover the reduced Hamiltonian on the subspace $\text{ker}(\mathbb L)^{\perp}$ by Legendre transformation.
We refer to \cite{cendra1998maxwell} for more details on how to obtain an equivalent Lagrangian system for degenerate Hamiltonians, but within the context of plasma physics. 

\subsection{Nonlinear stability results}

In this section, we will look for equilibrium solutions of the momentum-coupled network Lie-Poisson system \eqref{LP-network} on a compact, semi-simple Lie algebra and find their nonlinear stability properties. In particular, we will show that the eigenvectors of the extended graph Laplacian $\mathbb L$ correspond to equilibrium solutions of this system and they are either a ferromagnetic or an antiferromagnetic state. Furthermore, we will show that the eigenvectors with the lowest and highest eigenvalues correspond to nonlinearly stable equilibria.

\subsubsection{Equilibrium solutions}

Since we are on a compact semi-simple Lie algebra $\mathfrak g$, we can take $\text{\bf ad}^* = -\text{\bf ad}$, so equation \eqref{LP-network} becomes
\begin{align}
\dot{\boldsymbol \mu} = [\boldsymbol \mu, \mathbb L \boldsymbol \mu]\, .
\end{align}
In general, the solutions to $\dot{\boldsymbol \mu} = 0$ are given by all $\boldsymbol \mu_e \in \mathfrak g^*$ such that $\mathbb L \boldsymbol \mu_e$ is contained in the centralizer $Z(\boldsymbol \mu_e)$ of $\boldsymbol \mu_e$. However, here we will only consider the case where $\mathbb L \boldsymbol \mu_e = \lambda_e \boldsymbol \mu_e$ for some $\lambda_e \in \mathbb R$, which is clearly contained in $Z(\boldsymbol \mu_e)$. So fixing $a > 0$, a regular value of $C:\mathfrak g^* \rightarrow \mathbb R$, we consider relative equilibrium configurations on the level set $C^{-1}(a)$ that are given by the $kN$ eigenvectors of $\mathbb L$, rescaled appropriately to satisfy $C(\boldsymbol \mu_e) = a$. Now, pairing both sides of $\mathbb L \boldsymbol \mu_e = \lambda_e \boldsymbol \mu_e$ with $\frac12 \boldsymbol \mu_e$, we get
\begin{align}
\lambda_e = \frac{1}{C(\boldsymbol \mu_e)}h(\boldsymbol \mu_e)\, , \label{lambda-energy}
\end{align}
so the eigenvalue of $\mathbb L$ is proportional to the total energy of the equilibrium state.

In the special case where $\mathbb J_{ij} = \mathbb J$ and $\mathbb I_i = \mathbb I$ for all $i,j=1, \ldots, N$, we see that these $kN$ eigenvectors can be categorised into two types: ferromagnetic or anti-ferromagnetic states, as given in the following proposition.
\begin{proposition} \label{equib}
Consider a network Lie Poisson system \eqref{LP-network} on $(\mathcal N, G, \mathcal O)$ where $\mathfrak g = \text{Lie}(G)$ is compact, semi-simple and let $a > 0$ be a regular value of $C : \mathfrak g^* \rightarrow \mathbb R$.
If the interaction tensor is the same for all edges, i.e. $\mathbb J_{ij} = \mathbb J$ and the moment of inertia tensor is the same at all nodes, i.e. $\mathbb I_i = \mathbb I$, then there exists $kN$ linearly independent relative equilibrium configurations $\boldsymbol \mu^e = (\mu_1^e, \cdots, \mu_N^e)$ on the level set $C^{-1}(a)$ such that
\begin{enumerate}
  \item $k$ are ferromagnetic states, i.e. $\mu_i^e = \sqrt{d_i}\mu^e$ for all $i = 1, \ldots, N$, where $\mu^e$ is an eigenvector of the extended inertia matrix $\mathbb I_{\text{ext}} := \mathbb I^{-1} - \mathbb J$, and 
  \item the remaining $(N-1)k$ states are anti-ferromagnetic, i.e. $\sum_{i=1}^N\sqrt{d_i}\mu_i^e  = 0 $. 
\end{enumerate}
\end{proposition}
\begin{proof}
Consider the orthogonal decomposition of $(\mathfrak g^*)^{\oplus N} $ into ferromagnetic and anti-ferromagnetic states, i.e. $(\mathfrak g^*)^{\oplus N} = V \oplus V^{\perp}$, where
\begin{align*}
&V = \left\{\boldsymbol \mu = (\mu_1, \cdots, \mu_N) \in (\mathfrak g^*)^{\oplus N} : \mu_i = \sqrt{d_i} \mu, \quad \mu \in \mathfrak{g}^*\right\} \\
&V^{\perp} = \left \{\boldsymbol \mu = (\mu_1, \cdots, \mu_N) \in (\mathfrak g^*)^{\oplus N} :  \sum_i \sqrt{d_i} \mu_i = 0\right \}\,  .
\end{align*}
It is easy to check that $V^{\perp}$ is the orthogonal complement of $V$ with $\text{dim}(V) = k$ and $\text{dim}(V^{\perp}) = (N-1)k$. We claim that $V$ and $V^{\perp}$ are invariant subspaces under the linear operation $\mathbb L$.

To verify the former, take $\boldsymbol \mu \in V$, so $ \mu_i = \sqrt{d_i} \mu$ for all $i = 1, \ldots, N$. Then,
\begin{align*}
(\mathbb L \boldsymbol \mu)_i = \sqrt{d_i} \,\mathbb I^{-1} \mu - \sum_j \frac{A_{ij}}{\sqrt{d_i}} \mathbb J \mu = \sqrt{d_i}\left(\mathbb I^{-1} \mu - \sum_j \frac{A_{ij}}{d_i} \mathbb J \mu \right) \nonumber  = \sqrt{d_i} \,\widehat{\mu}\, ,
\end{align*}
where $\widehat{\mu} :=  \mathbb I_{\text{ext}}\,\mu \in \mathfrak{g} \cong \mathfrak{g}^*$ and we used the fact that $\sum_j A_{ij} = d_i$. So $\mathbb L \boldsymbol \mu \in V$ for all $\boldsymbol \mu \in V$ and therefore $V$ is invariant under $\mathbb L$.

For the other case, take $\boldsymbol \mu \in V^{\perp}$, so that $\sum_i \sqrt{d_i}\mu_i = 0$. 
Then,
\begin{align*}
  \sum_{i=1}^N  \sqrt{d_i} (\mathbb L \boldsymbol \mu)_i  &= \sum_{i=1}^N  \sqrt{d_i}\left(\mathbb I^{-1} \mu_i - \sum_j \frac{A_{ij}}{\sqrt{d_id_j}} \mathbb J \mu_j \right) \\
&= \mathbb I^{-1} \sum_i \sqrt{d_i} \mu_i - \sum_{i,j} \frac{A_{ij}}{d_j} \mathbb J \sqrt{d_j} \mu_j\\
 &=  \mathbb I_{\text{ext}}  \sum_i \sqrt{d_i} \mu_i = 0\, , 
\end{align*}
where we used the fact that $\sum_i A_{ij} = d_j$ and the fact that $\boldsymbol \mu \in V^{\perp}$ in the last line. Hence, $\mathbb L \boldsymbol \mu \in V^{\perp}$ for all $\boldsymbol \mu \in V^{\perp}$ as expected.

This implies that there is an appropriate change of basis such that $\mathbb{L}$ becomes block diagonal,
\begin{align*}
\mathbb{L}\rightarrow 
    \begin{pmatrix}
    \mathbb{L}_1 & 0 \\
    0 & \mathbb{L}_2
    \end{pmatrix}\, ,
\end{align*}
where $\mathbb{L}_1 : V \rightarrow V^*$ and $\mathbb{L}_2: V^{\perp} \rightarrow (V^{\perp})^*$. It is easy to see that $\mathbb L_1 \equiv \mathbb I_{\text{ext}}$. Hence, $\mathbb{L}$ has $kN$ eigenvectors $\{\boldsymbol \mu^e_i\}_{i=1}^{kN}$, which are equilibrium solutions, where $\boldsymbol \mu^e_i= (\mathbf v_i, \mathbf 0_{V^{\perp}})$ for $i=1,\ldots k$ and $\boldsymbol \mu^e_i = (\mathbf 0_{V}, \mathbf w_i)$ for $i = k+1, \ldots, kN$, where $\mathbf v_i \in V$ are the $k$ eigenvectors of $\mathbb{L}_1\equiv \mathbb I_{\text{ext}}$ and $\mathbf w_i \in V^{\perp}$ are the $(N-1)k$ eigenvectors of $\mathbb{L}_2$ and this proves our result.
\end{proof}

\begin{remark}
Rescaling appropriately, all of the ferromagnetic and anti-ferromagnetic equilibrium states given above exist on a level set of the summed Casimir $C$, but not all of them exist if we restrict to a single coadjoint orbit.
\end{remark}

As we will see later with the $SO(3)$ example, most of the eigenvalues $\lambda_e$ of $\mathbb L$ have algebraic multiplicity $n_\lambda = \text{mult}(\lambda_e) > 1$, so if $\boldsymbol \mu^e_1$ and $\boldsymbol \mu^e_2$ are two eigenvectors of $\mathbb L$ sharing the same eigenvalue $\lambda_e$, then their linear combination is also an eigenvector and therefore an equilibrium solution. Hence, every point in the eigenspace $E(\lambda_e) := \text{ker}(\mathbb L - \lambda_e \mathbb 1)$ is an equilibrium solution with $\text{dim}(E(\lambda_e)) = n_\lambda$ (since $\mathbb L$ is symmetric, the algebraic and geometric multiplicity are equivalent). It is therefore possible to construct more complicated equilibrium states that are neither ferromagnetic nor anti-ferromagnetic by taking a linear combination of a ferromagnetic eigenvector and an anti-ferromagnetic eigenvector that share the same eigenvalue.

\subsubsection{Nonlinear stability}

We now investigate the stability properties of the equilibrium solutions found above via the energy-Casimir method, as given in \cite{arnold1966priori,holm1985nonlinear}. Our main result is stated as follows.

\begin{theorem} \label{equib-config}
Fixing a level set $C^{-1}(a)$ for $a>0$, the equilibrium solutions $\boldsymbol \mu_e$ of the network Lie-Poisson equation \eqref{LP-network} corresponding to the highest and lowest eigenvalues $\lambda_e$ of $\mathbb L$ are nonlinearly stable provided $\text{mult}(\lambda_e) = 1$.
\end{theorem}
\begin{proof}
Consider an equilibrium solution $\boldsymbol \mu_e$, which is an eigenvector of $\mathbb L$. In accordance with the energy-Casimir method, we first consider the augmented Hamiltonian
\begin{align}
  h_\Phi = h+ \Phi\left (C\right )\, , 
\end{align}
where $\Phi$ is an arbitrary real-valued function such that $\boldsymbol \mu_e$ is a critical point of $h_\Phi$.
That is,
\begin{align}
 D_{\boldsymbol \mu} h_\Phi(\boldsymbol \mu_e) \cdot \delta \boldsymbol \mu = \left(\mathbb L \boldsymbol \mu_e + \Phi'\left (a \right)\boldsymbol \mu_e\right) \cdot \delta \boldsymbol \mu = 0 \, ,
\end{align}
This equation holds for any free variations $\delta \boldsymbol \mu \in \mathfrak{g}$, provided $ \Phi'(a)  = -\lambda_e$, where $\lambda_e$ is the eigenvalue of $\mathbb L$ corresponding to the eigenvector $\boldsymbol \mu_e$. 

Next, we compute the second variation of $h_\Phi$.
\begin{align}
 \left<\delta \boldsymbol \mu, D^2_{\boldsymbol \mu} h_\Phi(\boldsymbol \mu_e) \cdot \delta \boldsymbol \mu \right>&=  \left<\delta \boldsymbol \mu, \left(\mathbb L + \Phi' (a)\right) \delta \boldsymbol \mu\right> + \Phi''\left (a\right ) (\boldsymbol \mu_e \cdot \delta \boldsymbol \mu)^2 \nonumber\\ 
  &=  \left<\delta \boldsymbol \mu, (\mathbb L- \lambda_e \mathbb{1}) \delta \boldsymbol \mu \right> + \Phi''(a) (\boldsymbol \mu_e \cdot \delta \boldsymbol \mu)^2  \, . 
\label{hessian-RB}
\end{align}
Since $\mathbb L$ is symmetric, we can change the basis such that $\mathbb L$ is diagonal with ordered eigenvalues, i.e. $\mathbb L \rightarrow \text{diag}(\lambda_1, \ldots, \lambda_{kN})$ with $\lambda_1 \geq \ldots \geq \lambda_{kN}$. Now, let $\boldsymbol \mu_e$ be an equilibrium solution with the highest eigenvalue, so that $\lambda_e = \lambda_1$ and assume that $\text{mult}(\lambda_e) = 1$. Then we have $\boldsymbol \mu_e = \sqrt{2a} \,\widehat{\boldsymbol e}_1$, where $\{\widehat{\boldsymbol e}_1, \ldots, \widehat{\boldsymbol e}_{kN}\}$ is the new basis and \eqref{hessian-RB} reduces to
\begin{align}
 \left< \delta \boldsymbol \mu, D^2_{\boldsymbol \mu} h_\Phi(\boldsymbol \mu_e)\delta \boldsymbol \mu \right> &=  \sum_{i=1}^{kN} (\lambda_i - \lambda_1 ) \delta \hat \mu_i^2 + 2a\Phi''(a) \delta \hat{\mu}_1^2  \, ,
\end{align}
where $\delta \hat{\mu}_i$ are the components of $\delta \boldsymbol \mu$ in the new basis. One can directly check that choosing $\Phi''(a) < 0$, the Hessian matrix $D^2_{\boldsymbol \mu} h_\Phi(\boldsymbol \mu_e)$ becomes strictly negative definite, so by the energy-Casimir method, the equilbrium configuration with the highest energy is nonlinearly stable. Similarly, choosing $\lambda_e = \lambda_{kN}$ (lowest energy) and $\Phi''(a) > 0$, the Hessian matrix $D^2_{\boldsymbol \mu} h_\Phi(\boldsymbol \mu_e)$ becomes strictly positive definite, so the equilibrium configuration with the lowest energy is also nonlinearly stable.
\end{proof}

\begin{remark}
From \eqref{lambda-energy}, we see that this corresponds to the highest and lowest energy equilibrium states on the level set $C^{-1}(a)$.
\end{remark}

We cannot deduce further about the nonlinear stability of the other equilibrium states, but as we will see with the $SO(3)$ example in section \ref{section-RB}, most of them are linearly unstable and the remaining few are linearly stable at best. To investigate the linear stability of the equilibrium solutions in the general setting, we first linearize the equation \eqref{LP-network} by setting $\boldsymbol \mu(t) = \boldsymbol \mu_e + \epsilon \, \delta \boldsymbol \mu(t)$ and dropping terms of $O(\epsilon^2)$ to get
\begin{align} \label{lin-eq}
  \dot{\delta \boldsymbol \mu} = \mathrm{\bf ad}_{\boldsymbol \mu_e} \left( (\mathbb L - \lambda_e \mathbb{1}) \delta \boldsymbol \mu \right) \,,  
\end{align}
where we have used $\mathbb L \boldsymbol \mu_e = \lambda_e \boldsymbol \mu_e$. We can then investigate the linear instability of the equilibrium solution $\boldsymbol \mu_e$ by looking at whether the eigenvalues of the linear operator $\mathrm{\bf ad}_{\boldsymbol \mu_e} \left( (\mathbb L - \lambda_e \mathbb 1)\, \cdot \right)$ has a positive real part. 

\subsection{Noise and dissipation}

We now perturb our system \eqref{LP-network} with noise and dissipation that preserves the Casimir $C$, following \cite{arnaudon2016noise}. 
The equation is similar to the single coadjoint orbit system with equation \eqref{SEP-Diss}, that is 
\begin{align}        
    \mathbb d\mu_i &+ \mathrm{ad}^*_\frac{\partial h}{\partial \mu_i} \mu_i\, dt 
    + \theta\, \mathrm{ad}^*_\frac{\partial C}{\partial \mu_i} \left [ \frac{\partial C}{\partial \mu_i}, \frac{\partial h}{\partial \mu_i} \right ]^\flat \, dt + \sum_l\mathrm{ad}^*_{\sigma_{i,l}} \mu_i \circ dW_t^{i,l} = 0 \,,
    \label{SEP-Diss-N}
\end{align}
for $i=1, \ldots, N$, where $h = \frac12 \left<\boldsymbol \mu, \mathbb L \boldsymbol \mu \right>$  and $C$ is a given Casimir for the Lie-Poisson bracket. 
The noise term has $kN$ independent Lie algebra vectors $\sigma_{i,l}$ associated to independent Wiener processes $W_t^{i,l}$ and the dissipation considered here is a double bracket dissipation (See \cite{gaybalmaz2013selective, bloch1996euler}).
One could generalise this equation further by letting $\theta_i$ be node dependent, but we will not do this here to keep our equations simple and we also choose our noise to be {\em isotropic}, that is, $\sigma_{i,l} = \sigma \boldsymbol e_l$, where $\boldsymbol e_l$ is a basis vector of $\mathfrak g$, which are the same at every node.  

Notice that without noise, the energy dissipates as 
\begin{align}
  \frac{d}{dt} h_{\text{int}}(\boldsymbol \mu) = -\theta\,  \sum_{i=1}^N  \left [ (\mathbb L \boldsymbol \mu)_i , \mu_i \right ] ^2\, 
\end{align}
and will tend towards the equilibrium configuration with minimum energy on the coadjoint orbit. 
If the noise is isotropic, the stationary distribution can be computed in the same way as we did with the single orbit case.

\begin{theorem}
The stationary distribution of \eqref{SEP-Diss-N} with  isotropic noise is the Gibbs distribution 
\begin{align}
  \mathbb P_\infty(\boldsymbol \mu ) = \frac{1}{Z} e^{-\beta h(\boldsymbol \mu)}  \, , 
  \label{P_infty}
\end{align}
where the inverse temperature is given by  
\begin{align}    
  \frac{1}{T} = \beta = \frac{\theta}{2\sigma^2} \, . 
\end{align}
The partition function is  given by
\begin{align}
    Z=  \int_{\boldsymbol {\mathcal O}} e^{-\beta h(\boldsymbol \mu)} d\boldsymbol \mu\, , 
    \label{Z-int}
\end{align}
where $\boldsymbol{ \mathcal O}= \mathcal O_1\times \ldots \times \mathcal O_{N}$ is the direct product of the coadjoint orbits at each lattice site,  or the total coadjoint orbit of the lattice. 
\end{theorem}
\begin{proof}
  We refer to \cite{arnaudon2016noise} for the proof of this formula, which is done by direct substitution into the Fokker-Planck equation \eqref{FP-Diss} of the stochastic process \eqref{SEP-Diss-N}. 
\end{proof}

The Gibbs distribution \eqref{P_infty} provides us with the notion of {\em temperature} in this system and as we will see in section \ref{PT-section}, we observe a phase transition for the case $G=SO(3)$ as we vary the temperature, similar to the classical Heisenberg model.
The mean field approximation can also be derived using the partition function \eqref{Z-int} and will allow us to help detect the phase transition. 

\section{Example I: Networks of rigid bodies} \label{section-RB}

In this section, we study in more detail the case $G=SO(3)$, corresponding to a network of interacting rigid bodies. The corresponding Lie algebra $\mathfrak{so}(3)$ is compact and semi-simple so we are in the setting of the general theory studied in section \ref{SS-section}.

\subsection{Equation of motion}

Consider a single free rigid body with configuration group $SO(3)$. The reduced Hamiltonian is given by the kinetic energy $h(\mathbf \Pi) = \frac12 \mathbf \Pi \cdot \mathbb I^{-1} \mathbf \Pi$ where $\mathbf \Pi \in \mathfrak{so}^*(3) \cong \mathbb R^3$ is the {\em angular momentum vector} and the corresponding equation, given by the $\mathfrak{so}^*(3)$ Lie-Poisson structure, is $\dot{\mathbf \Pi} = \mathbf \Pi \times \mathbf \Omega$, where $\mathbf \Omega := \mathbb I^{-1} \mathbf \Pi$ is the {\em angular velocity vector}. We will not discuss the dynamics of a single body further here, as it is standard in geometric mechanics and instead refer to chapter 15 of Marsden and Ratiu \cite{marsden1999book} for a more detailed exposition of the system.

Extending this to a network $(\mathcal N, SO(3), \mathcal O)$ of interacting rigid bodies via momentum coupling with inertia tensor $\mathbb I_i$ and interaction tensor $\mathbb J_{ij} $, we get the Hamiltonian 
\begin{align}
    h(\overline {\mathbf \Pi})= \frac12 \sum_i \mathbf \Pi_i\cdot \mathbf  \Omega_i - \frac12 \sum_i \sum_{j \sim i} \frac{1}{\sqrt{d_id_j}}\mathbf \Pi_i \cdot \mathbb J_{ij} \mathbf \Pi_j\, , 
    \label{RB-H-lattice}
\end{align}
where $\overline {\mathbf \Pi} = (\mathbf \Pi_1, \ldots, \mathbf \Pi_N)$ is the network extended angular momentum vector and $\mathbf \Omega_i := \mathbb I_i^{-1} \mathbf \Pi_i$ is the angular velocity at site $i$. Taking the $(-)$ Lie-Poisson structure
\begin{align} \label{RB-bracket}
\{f, g\}_{LP}^-(\overline {\mathbf \Pi}) =  -\sum_i \mathbf \Pi_i \cdot \frac{\partial f}{\partial \mathbf \Pi_i} \times \frac{\partial g}{\partial \mathbf \Pi_i},
\end{align}
we obtain the corresponding network Lie-Poisson equation
\begin{align}
    \dot{ \mathbf{\Pi}}_i = \mathbf \Pi_i\times \mathbf \Omega_i - \sum_{j \sim i} \frac{1}{\sqrt{d_id_j}} \, \mathbf \Pi_i\times \mathbb J_{ij} \mathbf \Pi_j\quad \forall i = 1, \ldots, N\, .
    \label{RB-Det-lattice}
\end{align}

One can check that the Casimirs for the Lie-Poisson bracket \eqref{RB-bracket} are given by
\begin{align}
C_i(\mathbf \Pi_i) = \frac12 \| \mathbf \Pi_i\|^2, \quad i=1, \ldots, N\, , 
\end{align}
and we denote the sum of the Casimirs by
\begin{align}
C(\overline {\mathbf \Pi}) = \sum_{i=1}^N C_i(\mathbf \Pi_i) = \frac12 \|\overline{\mathbf \Pi}\|^2\, .
\end{align}

The coadjoint orbit in this special case is given as follows.
\begin{theorem}
The coadjoint orbit $\boldsymbol{\mathcal O} = \mathcal O_1 \times \cdots \times \mathcal O_N$ of $SO(3)^{\times N}$ on $\mathfrak{so}^*(3)^{\oplus N} \cong \mathbb R^{3N}$ is given by the level sets of the Casimirs
\begin{align}
\mathcal O_i = \left\{\mathbf \Pi_i \in \mathbb R^{3} :  C_i(\mathbf \Pi_i) = c_i  \right\} \cong S^2\, ,
\end{align}
for some constants $c_1, \ldots, c_N$.
\end{theorem}
\begin{proof}
Since the action of $SO(3)$ on the two-sphere $S^2$ (viewed as a submanifold in $\mathbb R^3$) is transitive, the coadjoint orbit is exactly given by the level sets of the Casimirs.
\end{proof}

We observe that removing the kinetic energy term in \eqref{RB-H-lattice} and taking $d_i=d$ for all $i=1, \ldots, N$, we recover the Hamiltonian for the Heisenberg model. Hence, the rigid body network can be viewed as a Heisenberg model with mass-carrying spins.

\subsection{Equilibrium solutions}

We now investigate the relative equilibrium configurations of the deterministic rigid body network and its corresponding stability properties. Since the Lie algebra $\mathfrak{so}(3)$ is compact and semi-simple, we can apply propositions \ref{equib} and \ref{equib-config} to obtain information about the equilibrium configurations of the rigid body network and its corresponding nonlinear stability properties, which we summarise in the following proposition.
\begin{proposition} \label{stab-RB}
The relative equilibrium solutions $\overline {\mathbf \Pi}^e = (\mathbf \Pi^e_1, \ldots, \mathbf \Pi^e_N)$ of a rigid body network $(\mathcal N, SO(3), \mathcal O)$ correspond to the eigenvectors of $\mathbb L$ and those with the lowest and highest energy on a level set $C^{-1}(a)$ ($a > 0$) are nonlinearly stable. Furthermore, if $\mathbb J_{ij} = \mathbb J$ and $\mathbb I_i = \mathbb I$, then there exists $3N$ linearly independent relative equilibrium configurations such that
\begin{enumerate}
\item three are ferromagnetic states, i.e. $\mathbf \Pi_i^e = \sqrt{d_i} \mathbf \Pi^e$ for $i=1, \ldots, N$, where $\mathbf \Pi^e$ is an eigenvector of the extended inertia matrix $\mathbb I_{\text{ext}}:=\mathbb I^{-1} - \mathbb J$ and
\item the remaining $3N-3$ are anti-ferromagnetic states, i.e. $\sum_{i=1}^N \sqrt{d_i} \mathbf \Pi_i^e = \mathbf 0$
\end{enumerate}
\end{proposition}

Now, from \eqref{lin-eq}, the linearized equation for the rigid body network around an equilibrium configuration $\overline {\mathbf \Pi}^e$ is given by
\begin{align} \label{linearised-RB}
\frac{d}{dt}\delta \overline {\mathbf \Pi} = \widehat{\mathbf \Pi}^e(\mathbb L - \lambda_e \mathbb 1) \cdot \delta \overline {\mathbf \Pi}, \quad \widehat{\mathbf \Pi}^e := \text{diag}(\widehat{\mathbf \Pi}_1^e, \ldots, \widehat{\mathbf \Pi}_N^e)\, ,
\end{align}
where $\,\widehat{ } : (\mathbb R^3, \times) \rightarrow \mathfrak{so}^*(3)$ is the standard Lie algebra isomorphism that takes a vector in $\mathbb R^3$ to a skew-symmetric matrix in $\mathfrak{so}^*(3)$. 

\begin{figure}[htpb]
  \centering
  \subfigure[$\mathbb I=\text{diag}(1,1,1)$, $\mathbb J=\text{diag}(1,2,3)$]{\includegraphics[scale=0.52]{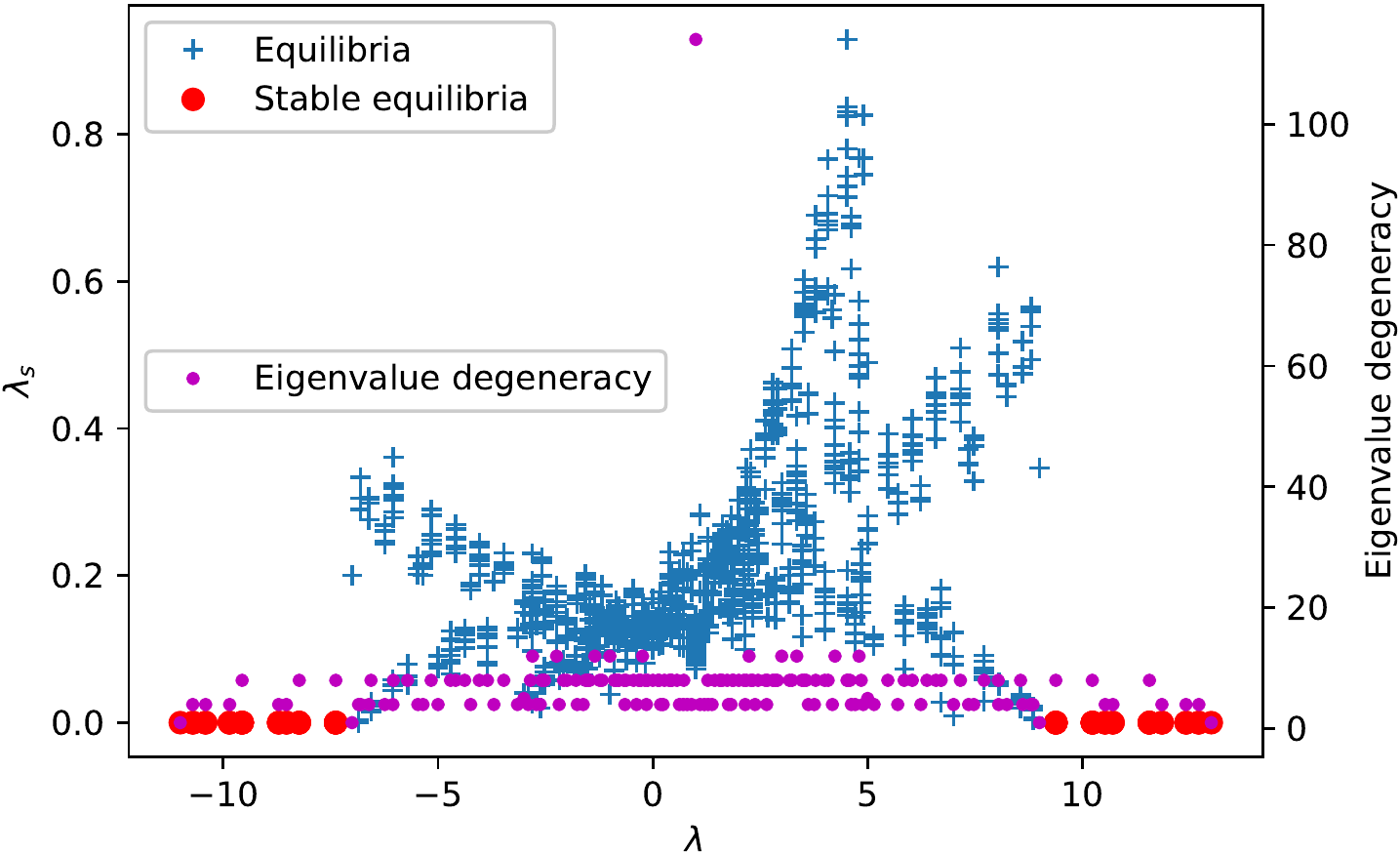}\label{fig:RB-stability1}}
  \subfigure[$\mathbb I=\text{diag}(1,2,3)$, $\mathbb J=\text{diag}(1,1,1)$]{\includegraphics[scale=0.52]{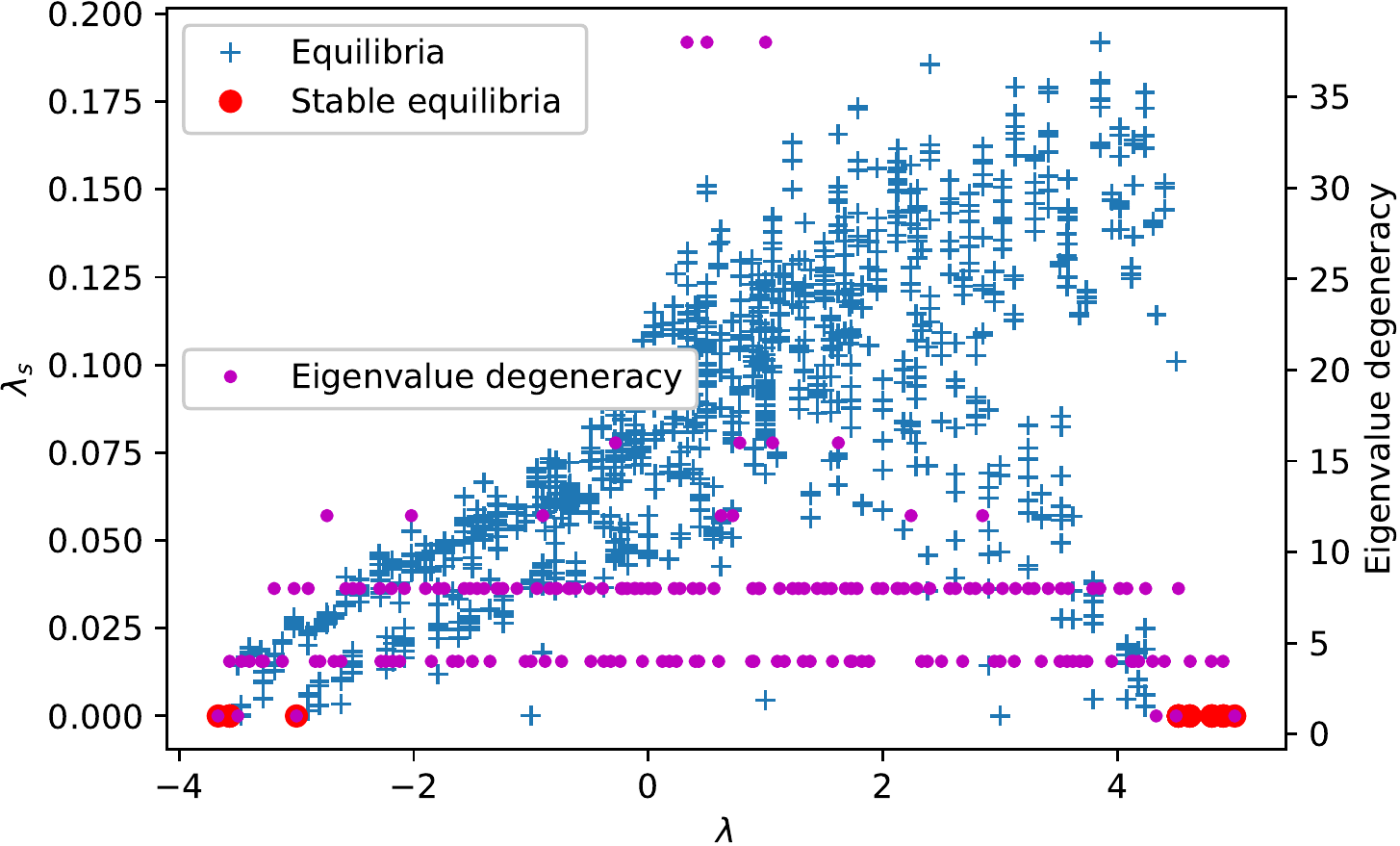}\label{fig:RB-stability2}}
  \caption{In this figure, we plot the maximum real part of the eigenvalues $\lambda_S$ of the linearised system \eqref{linearised-RB} corresponding to equilibrium solutions $\boldsymbol \Pi_e$ (eigenvectors of $\mathbb L$) with energy $\lambda_e$ (eigenvalue of $\mathbb L$ corresponding to eigenvector $\boldsymbol \Pi_e$). The blue crosses indicate the maximum real eigenvalues that are non-zero, corresponding to unstable solutions, and the red dots indicate those that are zero, corresponding to linearly stable solutions. The purple dots indicate the algebraic multiplicity $n$ of $\lambda_e$, whose values are to be read from the right axis. As expected, the equilibria corresponding to the highest and lowest energy are stable in both cases.}
  \label{fig:RB-stability}
\end{figure}

In figure \ref{fig:RB-stability}, we computed the eigenvalues $\lambda_s$ of the linearised system \eqref{linearised-RB} plotted against its corresponding energy $\lambda_e$ for two cases: (a) $\mathbb I=\text{diag}(1,1,1)$, $\mathbb J=\text{diag}(1,2,3)$ and (b) $\mathbb I=\text{diag}(1,2,3)$, $\mathbb J=\text{diag}(1,1,1)$, where we took a $20 \times 20$ lattice with periodic boundary condition. As expected from proposition \ref{stab-RB}, we see that in both cases, the highest and lowest energy configurations are linearly stable (both have multiplicity 1), however the majority of the configurations that lie between these two states are unstable. Interestingly, we also see a few linearly stable states towards the high and low end of the energy, with significantly more of them in \ref{fig:RB-stability1} than in \ref{fig:RB-stability2}.

\begin{figure}[htpb]
  \centering
  \subfigure[$\mathbb I = \mathrm{diag}(1,1,1)$, $\mathbb J = \mathrm{diag}(1,2,3)$, low energy position with  $\lambda= -9.8$]{\includegraphics[scale=0.30]{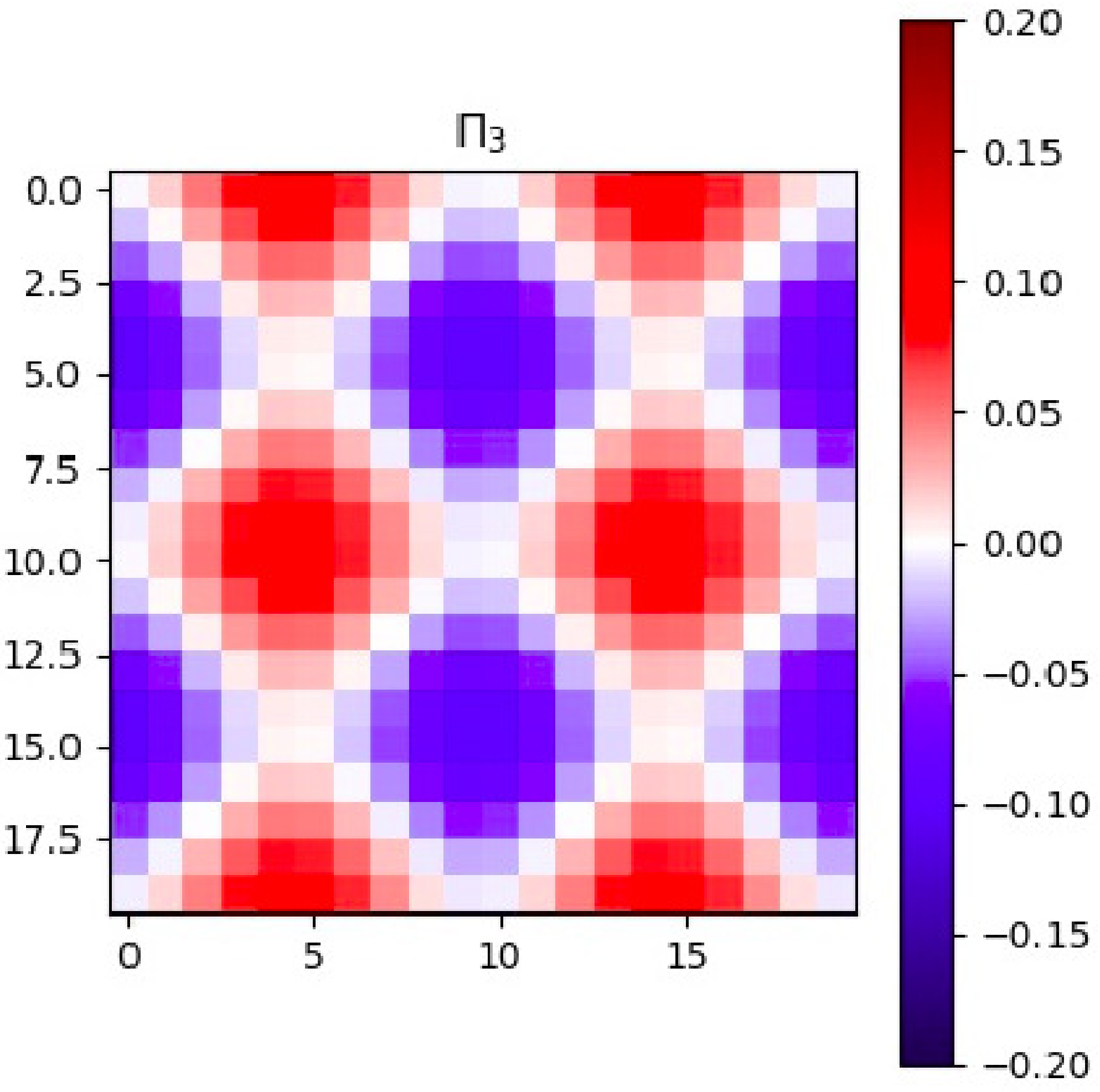}\label{fig:RB-sol1a}}
  \subfigure[$\mathbb I = \mathrm{diag}(1,2,3)$, $\mathbb J = \mathrm{diag}(1,1,1)$, high energy position with $\lambda= 4.6$]{\includegraphics[scale=0.30]{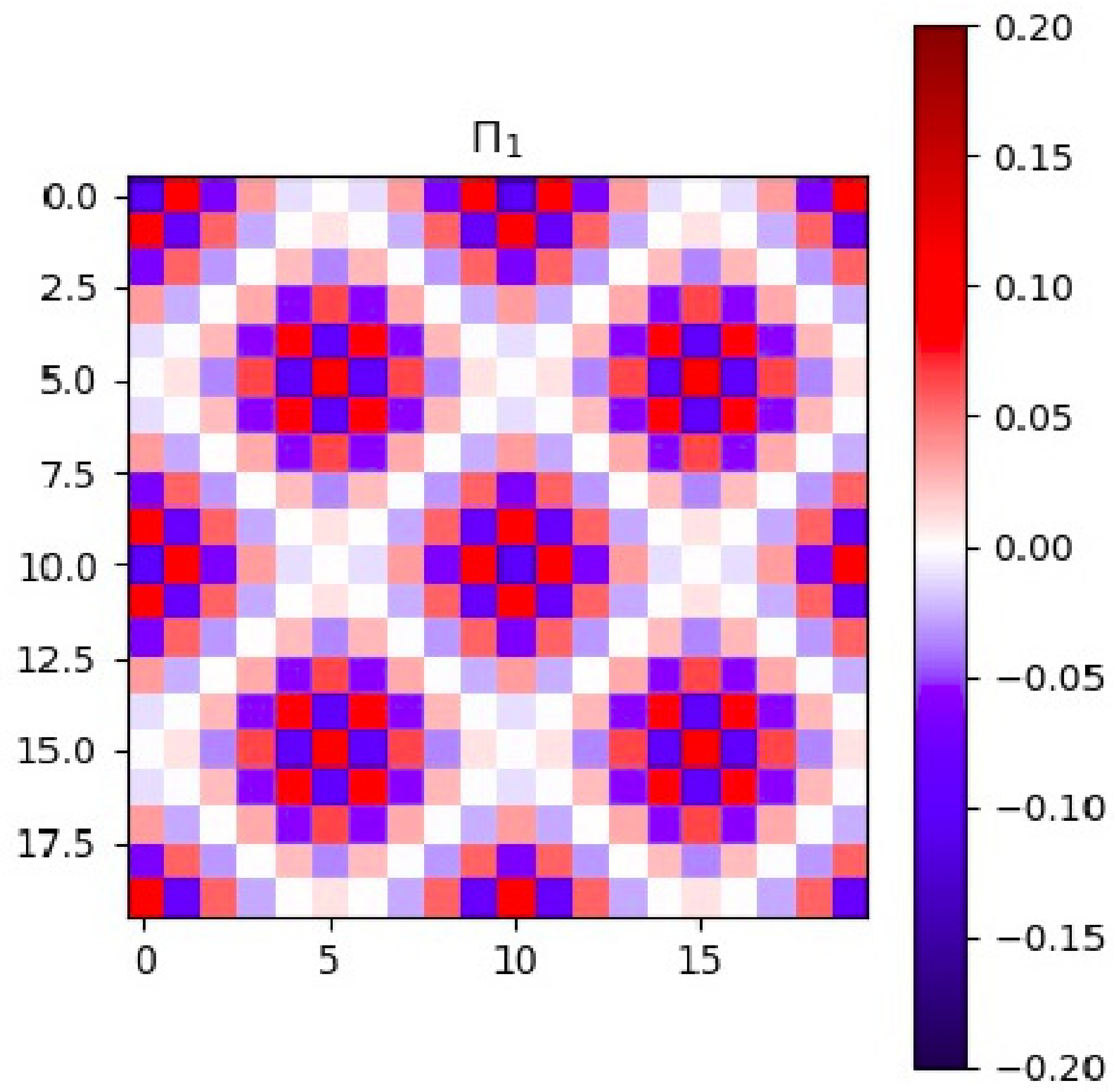}\label{fig:RB-sol1b}}
  \subfigure[$\mathbb I = \mathrm{diag}(1,1,1)$, $\mathbb J = \mathrm{diag}(1,2,3)$, high energy positions with  $\lambda= 12.7$]{\includegraphics[scale=0.30]{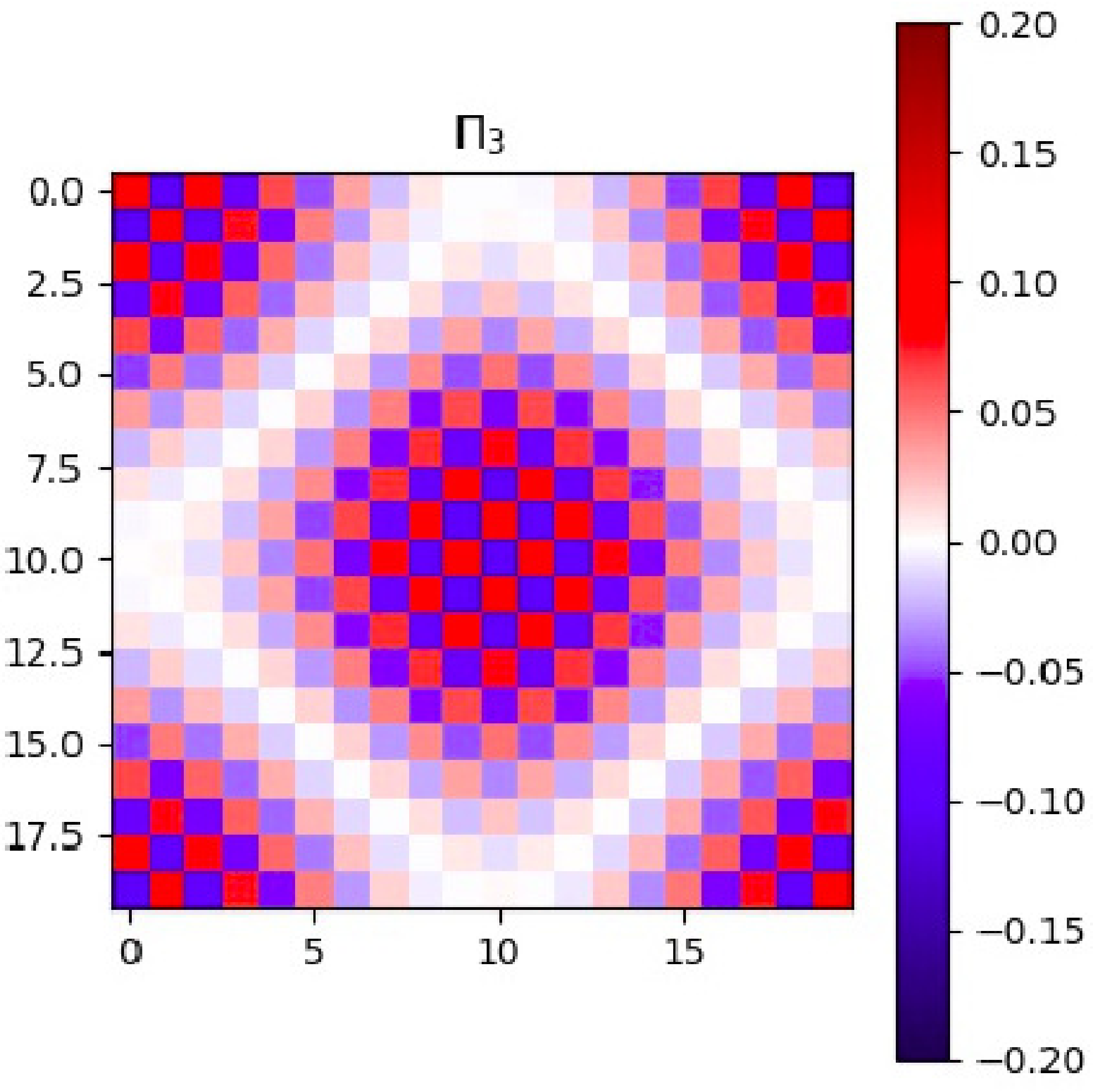}\label{fig:RB-sol2b}}
  \subfigure[$\mathbb I = \mathrm{diag}(1,2,3)$, $\mathbb J = \mathrm{diag}(1,1,1)$, low energy position with $\lambda= -3.5$]{\includegraphics[scale=0.30]{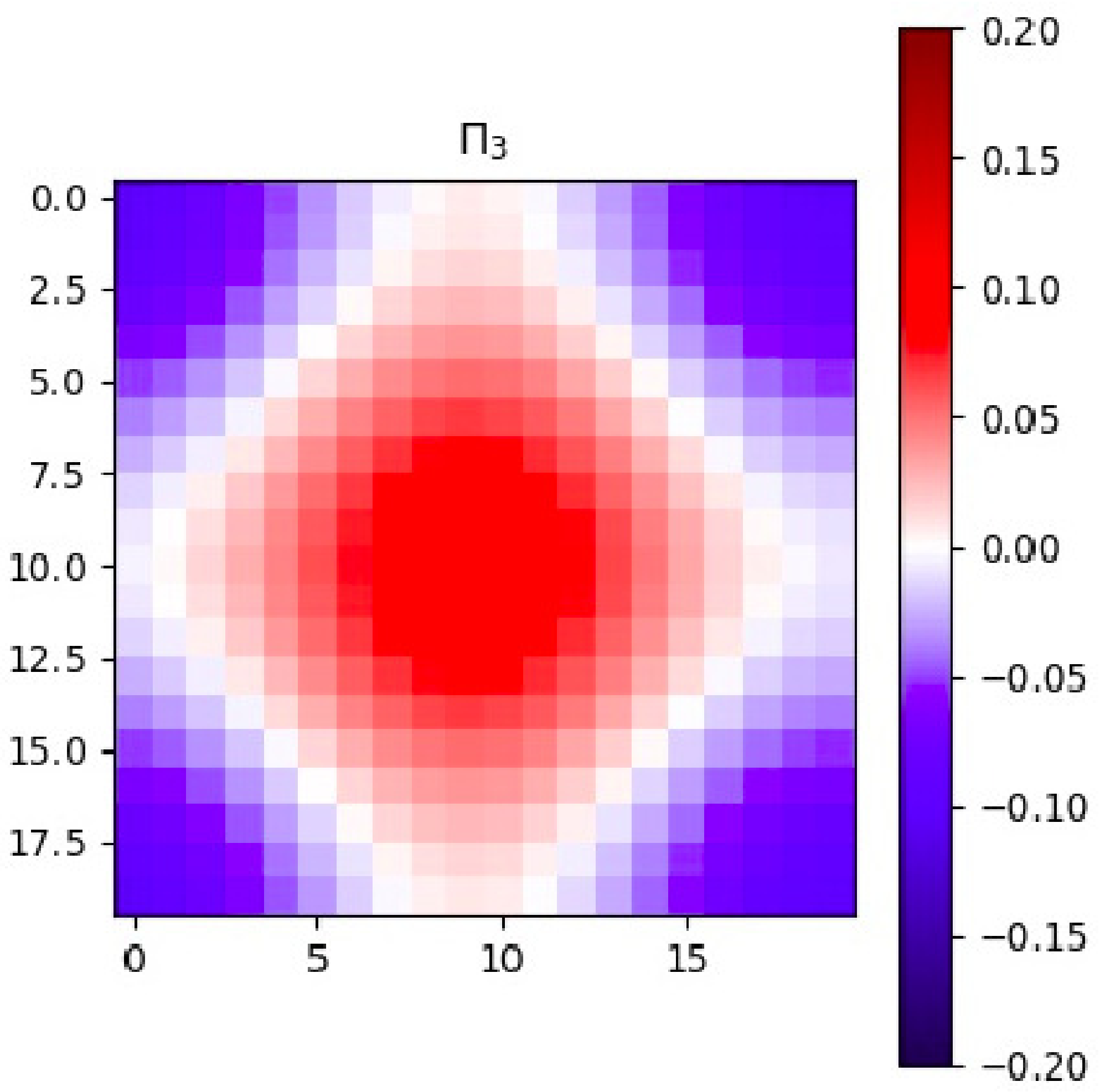}\label{fig:RB-sol2a}}
  \caption{In this figure, we display four non-trivial examples of linearly stable equilibrium solutions of the $20\times 20$ rigid body lattice. 
    All the solutions are parallel to one of the axes (either $\Pi_1$ or $\Pi_3$) and have an `argyle' pattern of aligned spins (\ref{fig:RB-sol1a}, \ref{fig:RB-sol2a}) or anti-aligned spins (\ref{fig:RB-sol1b}, \ref{fig:RB-sol2b}).
    The former pattern corresponds to configurations with low energy (negative values of $\lambda$), and the latter corresponds to configurations with high energy (positive values of $\lambda$). 
    All four solutions given here correspond to the red dots in figure \ref{fig:RB-stability}.}
  \label{fig:RB-sols}
\end{figure}

We observed numerically that the linearly stable states (red dots) given in figure \ref{fig:RB-stability1}, are parallel to the $\Pi_3$-axis, which tend to be characterised by a regular `argyle' pattern of aligned spins for low energy configurations (figure \ref{fig:RB-sol1a}) and anti-aligned spins for high energy configurations (figure \ref{fig:RB-sol2b}). As expected, the lowest energy state is a ferromagnetic state that is completely aligned with the $\Pi_3$-axis and the highest energy state is an anti-ferromagnetic state given by a fine checkerboard pattern. These last two configurations are furthermore nonlinearly stable by proposition \ref{stab-RB}.

The linearly stable states in figure \ref{fig:RB-stability2} can be described as follows, from lowest to highest energy. The lowest energy configurations is a ferromagnetic state that is parallel to the $\Pi_3$-axis as expected, the next ($\lambda = -3.5$) is a four dimensional subspace consisiting of a similar `argyle' pattern of aligned spins seen in the other case (figure \ref{fig:RB-sol2a}), the next ($\lambda = -3$) is again a ferromagnetic state, along the $\Pi_1$-axis this time, the last few ($4 < \lambda < 5$) consist of `argyle' patterns with anti-aligned spins parallel to the $\Pi_1$-axis (figure \ref{fig:RB-sol1b}) and the highest energy state is a completely anti-ferromagnetic checkerboard patterned state as expected.

\section{Network of coadjoint orbits II: Position coupling}\label{SD-section}

In this section, we consider a Lie-Poisson network of coadjoint orbits $(\mathcal N, G, \mathcal O)$ that arise from a different approach to coupling the neighbours. We will call this {\em position coupling}. In contrast to momentum coupling, where coupling occurs at the level of the reduced space $\mathfrak g^*$, in position coupling, the coupling occurs at the level of the group by considering a representation of $G$, followed by symmetry reduction, which yields a Euler-Poincar{\'e}/Lie-Poisson equation. This introduces a semi-direct product group structure into our system, analogous to the heavy top. We saw earlier that the momentum-coupled equations arise somewhat unnaturally from the Euler-Poincar{\'e} framework, however, there is no such issue for the position coupling approach. Furthermore, the equations that we obtain by position coupling are completely new to our knowledge and further investigation is necessary to understand the new phase transition behaviour that we will describe in section \ref{PT-section}.

\subsection{Deterministic equations}

Consider a network $\mathcal N$, a group $G$ and a left representation of $G$ on a vector space $V$. Fixing a vector $\boldsymbol a_0 = (A_1, \ldots, A_N) \in (V^*)^{\oplus N}$ and introducing the interaction tensor $\mathbb J_{ij} : V^* \rightarrow V$, we define the {\em position coupled interaction potential energy}
\begin{align} \label{h-int-HT}
h^{\text{int}}_{ij}(g_i, g_j ; A_i, A_j) = -\frac{1}{2\sqrt{d_id_j}} \langle g_i^{-1} A_i, \mathbb J_{ij} \, g_j^{-1} A_j \rangle_{V^* \times V}\, , \quad g_i, g_j \in G,
\end{align}
where nodes $i$ and $j$ are adjacent on the graph $\mathcal N$ and we denote by $gA$ to be the left action of $g \in G$ on $a \in V^*$. At node $i$, we assign a Lagrangian $L_i(g_i, \dot{g}_i; A_i)$ depending on the parameter $A_i$ that satisfies the following properties
\begin{enumerate}
\item $L_i(g_i, \dot{g}_i; A_i)$ is (left) invariant under $G_{A_i}$, namely the isotropy subgroup of $G$ which fixes the vector $A_i$.
\item The extended Lagrangian $L_i(g_i, \dot{g}_i, A_i)$ is invariant under the diagonal (left) action of $G$ on $TG \times V^*$.
\end{enumerate}
We take our full Lagrangian to be
\begin{align}
L(\boldsymbol g, \dot{\boldsymbol g}, \boldsymbol a_0) = \sum_i L_i(g_i, \dot{g}_i, A_i) - \sum_i \sum_{j \sim i} h^{\text{int}}_{ij}(g_i, g_j, A_i, A_j)\, ,
\end{align}
where $\boldsymbol g = (g_1, \ldots, g_N) \in G^{\times N}$.
We can check that this Lagrangian is invariant under the (left) diagonal action of $G^{\times N}$ on $TG^{\times N} \times (V^*)^{\times N}$, so we obtain the reduced Lagrangian
\begin{align}
L(\boldsymbol g, \dot{\boldsymbol g}, \boldsymbol a_0) &= L(\boldsymbol e, \boldsymbol g^{-1}\dot{\boldsymbol g}, \boldsymbol g^{-1}\boldsymbol a_0) \nonumber \\
&:= l(\boldsymbol \xi, \boldsymbol a)=\sum_i l_i(\xi_i, a_i) + \frac12 \sum_{i} \sum_{j \sim i}\frac{1}{\sqrt{d_id_j}} \langle a_i, \mathbb J_{ij} a_j \rangle\, ,
\end{align}
where $ \xi_i := g_i^{-1} \dot{g_i}, \quad a_i := g_i^{-1} A_i$, $\boldsymbol \xi = (\xi_1, \ldots, \xi_N)$ and  $\boldsymbol a = (a_1, \ldots, a_N)$.
This allows us to apply the semi-direct product reduction theorem given in Holm et al. \cite{holm1998euler} to obtain the following network Euler-Poincar\'e equations.
\begin{theorem}
The following statements are equivalent.
\begin{enumerate}
\item For a fixed vector $\boldsymbol a_0 \in (V^*)^{\oplus N}$, Hamilton's variational principle
    \begin{align}
    \delta \int_{t_1}^{t_2} L(\boldsymbol g, \dot{\boldsymbol g}; \boldsymbol a_0) \, dt = 0
    \end{align}
holds for variations $\delta \boldsymbol g$ vanishing at the endpoints.
\item $\boldsymbol g(t)$ satisfies the Euler-Lagrange equations for $L|_{\boldsymbol a_0}$ on $G^{\times N}$.
\item The constrained variational principle
    \begin{align}
    \delta \int_{t_1}^{t_2} l(\boldsymbol \xi(t), \boldsymbol a(t)) \, dt = 0
    \end{align}
holds on $(\mathfrak g \oplus V^*)^{\oplus N}$ with variations taking the form
    \begin{align}
    \delta \xi_i = \dot{\eta_i} + [\xi_i, \eta_i], \quad \delta a_i = - \eta_i a_i\, ,
    \end{align}
for all $i=1,\ldots,N$, where $\eta_i(t) := \delta g_i(t) \in \mathfrak g$ vanishes at the endpoints.
\item The following Euler-Poincar{\'e} equations hold on $(\mathfrak g \oplus V^*)^{\oplus N}$
    \begin{align}
    \frac{d}{dt} \frac{\delta l_i}{\delta \xi_i} &= \ad^*_{\xi_i} \frac{\delta l_i}{\delta \xi_i} + \frac{\delta l_i}{\delta a_i} \diamond a_i + \left(\mathbb D^{-\frac12} \mathbb A \mathbb D^{-\frac12} \boldsymbol a\right)_i \diamond a_i, \label{EP-semidirect} \\
\frac{da_i}{dt} &= - \xi_i a_i
    \end{align}
for all $i=1, \ldots, N$, where $\diamond: V^* \times V \rightarrow \mathfrak{g}^*$  is the momentum map defined by
\begin{align}
\left< \left< q, \, \xi p \right> \right> = - \langle \xi, p \,\diamond \, q \rangle\, ,
\end{align}
where $\langle \langle \cdot, \cdot \rangle \rangle$ is the natural pairing on $V^* \times V$.
\end{enumerate}
\end{theorem}

\subsubsection{Legendre transformation and Lie-Poisson equation}

We find the reduced Hamiltonian by taking the partial Legendre transform in the velocity variable,
\begin{align}
h(\boldsymbol \mu, \boldsymbol a) &= \langle \boldsymbol \mu, \boldsymbol \xi \rangle - l(\boldsymbol \xi, \boldsymbol a) =: \sum_i h_i(\mu_i, a_i) - \frac12\left<\boldsymbol  a, \mathbb D^{-\frac12} \mathbb A \mathbb D^{-\frac12} \boldsymbol a \right>\, ,
\label{Legendre}
\end{align}
where $\boldsymbol \mu = (\mu_1, \ldots, \mu_N)$ and $\quad \mu_i = \frac{\delta l}{\delta \xi_i}$.
The Euler-Poincar\'e equation  \eqref{EP-semidirect} then becomes a Lie-Poisson equation 
\begin{align}
  \begin{split}
\frac{d \mu_i}{dt} &= \ad^*_{\delta h_i / \delta \mu_i} \mu_i - \frac{\delta h_i}{\delta a_i} \diamond a_i - \left(\mathbb D^{-\frac12} \mathbb A \mathbb D^{-\frac12} \boldsymbol a\right)_i \diamond a_i \\
\frac{d a_i}{dt} &= - \frac{\delta h_i}{\delta \mu_i} \cdot a_i\, ,
  \end{split}
\end{align}
which is indeed Lie-Poisson, with bracket given by the $(-)$ Lie-Poisson bracket on the semi-direct product algebra $(\mathfrak{g}^* \circledS V^*)^{\times N} $, which is given by
\begin{align}
\{f,g\}^-_{LP}(\boldsymbol \mu, \boldsymbol a) &= -\sum_i \left(\left< \mu_i, \left[\frac{\delta f}{\delta \mu_i}, \frac{\delta g}{\delta \mu_i} \right] \right> + \left<a_i, \frac{\delta f}{\delta \mu_i}\frac{\delta g}{\delta a_i} - \frac{\delta g}{\delta \mu_i} \frac{\delta f}{\delta a_i} \right>\right)\, . 
\label{LP-semidirect}
\end{align}
\begin{remark}
Likewise, for a right invariant Lagrangian and right representation, we recover the $(+)$ Lie-Poisson bracket. Again, we refer the readers to Holm et al. \cite{holm1998euler} for more details about left vs. right group action.
\end{remark}

For the special case where $\mathfrak g = \text{Lie}(G)$ is compact, semi-simple and given a coadjoint representation of $G$ on $V^* = \mathfrak g^*$, the Lie-Poisson structure becomes
\begin{align}
\{f,g\}^-_{LP}(\boldsymbol \mu, \boldsymbol a) &= -\sum_i \left(\left< \mu_i, \left[\frac{\delta f}{\delta \mu_i}, \frac{\delta g}{\delta \mu_i} \right] \right> + \left<a_i, \left[\frac{\delta f}{\delta \mu_i},\frac{\delta g}{\delta a_i}\right] - \left[\frac{\delta g}{\delta \mu_i}, \frac{\delta f}{\delta a_i}\right] \right>\right)\, , 
\label{LP-semidirect-2}
\end{align}
with the corresponding Lie-Poisson equation
\begin{align}
  \begin{split}
\frac{d \mu_i}{dt} &= \left[\mu_i, \frac{\delta h_i}{\delta \mu_i} \right] + \left[a_i, \frac{\delta h_i}{\delta a_i} - \left(\mathbb D^{-\frac12} \mathbb A \mathbb D^{-\frac12} \boldsymbol a\right)_i\,\right] \\
\frac{d a_i}{dt} &= \left[a_i,\frac{\delta h_i}{\delta \mu_i} \right]\, . 
\label{LP-eq-semidirect}
  \end{split}
\end{align}
We refer to Ratiu et al. \cite{ratiu1981euler,ratiu1982lagrange} for more details on similar systems and in particular their complete integrability in the case of the Lagrange top (which will be lost here, due to the interaction terms). 

One can check that the Casimirs for the bracket \eqref{LP-semidirect-2} are given by
\begin{align}
  C_{i,1} = \langle \mu_i ,a_i\rangle \qquad \mathrm{and} \qquad C_{i,2}   \langle a_i ,a_i\rangle\, , 
\end{align}
for $i=1, \ldots, N$, where we take $\langle \cdot, \cdot \rangle = -\kappa(\cdot, \cdot)$ and the coadjoint orbits of the semi-direct product group $(G \,\circledS \,\mathfrak g)^{\times N}$ are contained in the level sets of the Casimirs, i.e.
\begin{align}
  \mathcal O \subset \left\{(\boldsymbol \mu, \boldsymbol a) \in (\mathfrak g^* \,\circledS \,\mathfrak g^*)^{\times N} :  C_{i,1}(\boldsymbol \mu, \boldsymbol a) = c_{i,1}, C_{i,2}(\boldsymbol \mu, \boldsymbol a) = c_{i,2}, \, \forall i = 1, \ldots, N \right\}\, ,
\end{align}
for given $2N$ constants $c_{1,1}, \ldots, c_{N,1}$ and $c_{1,2}, \ldots, c_{N,2}$. Again, this follows from the $\text{Ad}^*$-invariance of the Killing form. The preservation of coadjoint orbits under the dynamics of the position coupled network Lie-Poisson system again follows from the equivariance of the momentum map $\mathbf J_R : T^*(G \,\circledS \,\mathfrak g)^{\times N} \rightarrow (\mathfrak g^* \,\circledS \,\mathfrak g^*)^{\times N}$ giving rise to the Lie-Poisson structure \eqref{LP-semidirect}.
We will also denote by 
\begin{align}
  C_1(\boldsymbol \mu, \boldsymbol a) = \sum_{i=1}^N C_{i,1}(\mu_i, a_i)\qquad \mathrm{and} \qquad C_2( \boldsymbol a) = \sum_{i=1}^N C_{i,2}(a_i)\,, 
\end{align}
as the sum of the Casimirs. 

\subsection{Equilibrium positions}\label{HT-eq-general}

We now seek for equilibrium solutions of \eqref{LP-eq-semidirect} that correspond to the critical points of the Hamiltonian restricted to the level sets of the summed Casimirs $C_1$ and $C_2$. Here, we specialise to Hamiltonians of the form
\begin{align}
 h(\boldsymbol \mu,\boldsymbol a) = h^{KE}(\boldsymbol \mu) + h^\mathrm{int}(\boldsymbol a)\, , 
\end{align}
where
\begin{align}
h^{KE}(\boldsymbol \mu) = \frac12 \langle \boldsymbol \mu, \overline {\mathbb I}^{-1} \boldsymbol \mu \rangle, \quad \overline {\mathbb I}^{-1} := \text{diag}(\mathbb I^{-1}_1, \ldots, \mathbb I^{-1}_N)\, ,
\end{align}
for inertia tensors $\mathbb I_i : \mathfrak g \rightarrow \mathfrak g^*$ at each node $i$ is the kinetic energy, and
\begin{align}
h^{\text{int}}(\boldsymbol a) = - \frac12 \left< \boldsymbol a, \mathbb D^{-\frac12} \mathbb A \mathbb D^{-\frac12} \,\boldsymbol a \right>\, .
\end{align}
is the interaction potential energy. Now consider the augmented Hamiltonian
\begin{align}
  h_{\phi,\psi}(\boldsymbol \mu,\boldsymbol a) = h(\boldsymbol \mu,\boldsymbol a) + \phi(C_1) + \psi(C_2)\, , 
\end{align}
where $\phi, \psi$ are arbitrary smooth functions and take its first variations, which set to $0$ gives the condition for $(\boldsymbol \mu_e, \boldsymbol a_e) $ to be an equilibrium solution. 
One could also choose a more general function of the two Casimirs, but this form turns out to be general enough for our purpose. Taking the variation, we get
\begin{align}
  \begin{split}
    Dh_{\phi,\psi}(\boldsymbol \mu_e, \boldsymbol a_e) \cdot (\delta \boldsymbol \mu,\delta \boldsymbol a)^T &= \sum_i \left( \left<\mathbb I^{-1}_i \mu_i^e  + \lambda_1 \, a_i^e, \, \delta \mu_i \right> \right.\\
  &\left. + \left< -\left(\mathbb D^{-\frac12} \mathbb A \mathbb D^{-\frac12} \boldsymbol a_e\right)_i+ \lambda_1 \, \mu_i^e  + \lambda_2 \, a_i^e, \,\delta a_i \right> \right) = 0  \, , 
  \end{split}
\end{align}
where we defined
\begin{align}
  \lambda_1 &:= \phi'(c_1)\qquad \mathrm{and} \qquad \lambda_2 := 2\psi'(c_2)\, . 
\end{align}

The conditions for the equilibrium solutions are hence given by
\begin{align}
&\delta \boldsymbol \mu: \overline {\mathbb I}^{-1} \boldsymbol \mu_e + \lambda_1 \boldsymbol a_e = 0 \label{mu-eq} \\
&\delta \boldsymbol a: -\mathbb D^{-\frac12} \mathbb A \mathbb D^{-\frac12} \boldsymbol a_e + \lambda_1 \boldsymbol \mu_e + \lambda_2 \,\boldsymbol a_e = 0\, , 
 \label{a-eq}
\end{align}
and pairing \eqref{mu-eq} with $\boldsymbol \mu_e$ and \eqref{a-eq} with $\boldsymbol a_e$, we find
\begin{align}
&\lambda_1 = -\frac{1}{c_1}\left<\boldsymbol \mu_e, \overline {\mathbb I}^{-1} \boldsymbol \mu_e \right> \label{lambda1-eq} \\
&\lambda_2 = \frac{1}{c_2}\left(\left<\boldsymbol \mu_e, \overline {\mathbb I}^{-1} \boldsymbol \mu_e \right> + \left<\boldsymbol a_e,  \mathbb D^{-\frac12} \mathbb A \mathbb D^{-\frac12} \boldsymbol a_e\right> \right)\, . 
\label{lambda2-eq}
\end{align} 

Now, substituting \eqref{mu-eq} into \eqref{a-eq}, we obtain the following equation
\begin{align}
\mathbb L(\lambda_1) \,\boldsymbol a_e = -\lambda_2 \,\boldsymbol a_e\, ,
\label{L_l1l2}
\end{align}
where $\mathbb L(\lambda_1) := -\lambda_1^2 \,\overline {\mathbb I} - \mathbb D^{-\frac12} \mathbb A \mathbb D^{-\frac12}$ is our new extended graph Laplacian. Hence, fixing $\lambda_1$, the equilibrium solutions $\boldsymbol a_e$ and $\boldsymbol \mu_e = -\lambda_1 \overline {\mathbb I} \,\boldsymbol a_e$ again correspond to the $kN$ eigenvectors of the graph Laplacian $\mathbb L(\lambda_1)$ with eigenvalue $-\lambda_2$.
From this, we immediately deduce the following result.

\begin{proposition}
Fixing $\lambda_1$ and one of the summed Casimirs $C_1$ or $C_2$, and furthermore, assuming that $\mathbb J_{ij} = \mathbb J$ and $\mathbb I_i = \mathbb I$ for all $i,j = 1, \ldots, N$, there exists $kN$ equilibrium solutions $\boldsymbol a_e$ and $\boldsymbol \mu_e = -\lambda_1 \overline {\mathbb I} \,\boldsymbol a_e$ such that
\begin{enumerate}
\item $k$ are ferromagnetic states, i.e. $a_i = \sqrt{d_i} a$ for all $i$, where $a$ is an eigenvector of the extended inertia tensor $\mathbb I_{\text{ext}}(\lambda_1) := -\lambda_1^2 \mathbb I - \mathbb J$,
\item The remaining $(N-1)k$ are anti-ferromagnetic states, i.e $\sum_{i=1}^N \sqrt{d_i} a_i = 0$.
\end{enumerate}
\end{proposition}
The proof is similar to that of proposition \ref{equib}.

\begin{remark}
It is important to note that the ferromagnetic and anti-ferromagnetic equilbrium solutions given above are only found if we fix $\lambda_1$ and only one of the summed Casimirs $C_1$ or $C_2$. If instead, we want to find the equilibrium solutions on the level sets of both $C_1$ and $C_2$ while keeping the parameter $\lambda_1$ free, then one has to solve the full nonlinear equation \eqref{mu-eq}, \eqref{a-eq}, \eqref{lambda1-eq}, \eqref{lambda2-eq} which in general is difficult to solve.
\end{remark}

\subsection{Nonlinear stability analysis}

We saw that the eigenvectors of the graph Laplacian $\mathbb L(\lambda_1)$ correspond to equilibrium solutions of \eqref{LP-eq-semidirect}, similar to the momentum-coupled case. We will now show that the equilibrium configuration corresponding to the eigenvector of $\mathbb L(\lambda_1)$ with the {\em lowest} eigenvalue is nonlinearly stable. This differs slightly from the momentum coupled case where we were able to prove that both the lowest and highest eigenvalue configurations are stable.

\begin{theorem}
Fixing $\lambda^e_1$, the equilibrium configuration corresponding to the eigenvector of $\mathbb L(\lambda_1^e)$ with the lowest eigenvalue $-\lambda_2^e$ is nonlinearly stable, provided $\text{mult}(-\lambda_2^e) = 1$.
\end{theorem}

\begin{proof}
We apply the energy-Casimir method to assess the stability of the lowest eigenvalue configuration. The Hessian $D^2 h_{\phi,\psi}(\boldsymbol \mu_e, \boldsymbol a_e)$ of the augmented Hamiltonian is a symmetric matrix given by
\begin{align*}
D^2 h_{\phi,\psi}(\boldsymbol \mu_e, \boldsymbol a_e) =
\begin{pmatrix}
    \overline {\mathbb I}^{-1}  + \hat \lambda_1 \,\boldsymbol a_e^T \boldsymbol a_e & \lambda^e_1 \,\mathbb 1 +  \hat \lambda_1 \,\boldsymbol a_e^T \boldsymbol \mu_e \\
    \lambda^e_1 \,\mathbb 1 +  \hat \lambda_1 \,\boldsymbol \mu_e^T \boldsymbol a_e & \lambda^e_2 \,\mathbb 1 - \mathbb D^{-\frac12} \mathbb A \mathbb D^{-\frac12} + \hat \lambda_2 \,\boldsymbol a_e^T \boldsymbol a_e + \hat \lambda_1 \,\boldsymbol \mu_e^T\boldsymbol \mu_e
  \end{pmatrix}\, ,
\end{align*}
where $\hat{\lambda}_1 := \phi''(c_1)$ and $\hat{\lambda}_2 := 4\psi''(c_2)$.
Setting $\hat{\lambda}_1 = 0$, this simplifies to
\begin{align*}
D^2 h_{\phi,\psi}(\boldsymbol \mu_e, \boldsymbol a_e) =
\begin{pmatrix}
    \overline {\mathbb I}^{-1} & \lambda^e_1 \,\mathbb 1 \\
    \lambda^e_1 \,\mathbb 1 & \lambda^e_2 \,\mathbb 1 - \mathbb D^{-\frac12} \mathbb A \mathbb D^{-\frac12} + \hat \lambda_2 \,\boldsymbol a_e^T \boldsymbol a_e
  \end{pmatrix}
=:
\begin{pmatrix}
X & Y \\
Y^T & Z
\end{pmatrix}\, .
\end{align*}
It is well-known that block matrices of this form are positive definite if and only if the upper left block $X$ and the Schur complement $B := Z - Y^T X^{-1} Y$ are both positive definite. Since we take $\mathbb I$ to be positive definite, it follows that $X$ is positive definite so we only need to show that $B$ is positive definite. Written in full, in terms of the graph Laplacian $\mathbb L(\lambda^e_1) = -(\lambda^e_1)^2 \,\overline {\mathbb I} - \mathbb D^{-\frac12} \mathbb A \mathbb D^{-\frac12}$, one can check that
\begin{align*}
\delta \boldsymbol a^T B \,\delta \boldsymbol a = \delta \boldsymbol a^T \left(\mathbb L(\lambda^e_1) + \lambda^e_2 \mathbb 1 \right) \delta \boldsymbol a + \hat \lambda_2 (\boldsymbol a_e \cdot \delta \boldsymbol a)^2\, .
\end{align*}
Since $\mathbb L(\lambda^e_1)$ is symmetric, we can diagonalise it so that $\mathbb L(\lambda^e_1) \rightarrow \text{diag}(\alpha_1, \ldots, \alpha_{kN})$ with $\alpha_1 \geq \ldots \geq \alpha_{kN}$. Recalling that the equilibrium solution satisfies the eigenvalue problem $\mathbb L(\lambda^e_1)\, \boldsymbol a_e = -\lambda_2^e \,\boldsymbol a_e$, we take $-\lambda^e_2 = \alpha_{kN}$ which is the lowest eigenvalue of $\mathbb L(\lambda^e_1)$ and we assume that $\text{mult}(-\lambda_2^e)=1$. We then have $\boldsymbol a_e = \sqrt{c_2}\,\widehat{\boldsymbol e}_{kN}$ and get
\begin{align*}
\delta \boldsymbol a^T B \,\delta \boldsymbol a = \sum_{i=1}^{kN} (\alpha_i - \alpha_{kN})\delta \hat{a}_i^2 + \hat{\lambda}_2 c_2 \delta \hat{a}_{kN}^2\, ,
\end{align*}
where $\delta \hat{a}_i$ for $i=1, \ldots,kN$ are the components of $\delta \boldsymbol a$ in this new-basis. From this, it is easy to see that the Schur complement $B$ is positive definite if we choose $\hat{\lambda}_2 > 0$. Hence, $D^2 h_{\phi,\psi}(\boldsymbol \mu_e, \boldsymbol a_e)$ becomes positive definite and this configuration is nonlinearly stable, by the energy-Casimir method.
\end{proof} 

\subsection{Noise and dissipation}

Following Arnaudon et al. \cite{arnaudon2016noise}, the general stochastic equations for this system is given by
\begin{align}
    \begin{split}
      \mathbb d\mu_i  &+ [\xi_i,  \mu_i] +\left [\chi(  \boldsymbol a)_i,a_i\right ]\, dt 
      + \theta\, \left [\frac{\partial C}{\partial \mu_i}, \left [  \frac{\partial C}{\partial \mu_i}, \xi_i \right ]\right]  dt
      +\, \theta\, \left [ \frac{\partial C}{\partial a_i}, \left [ \frac{\partial C}{\partial \mu}, \chi(\boldsymbol a)_i \right ] +\left [\frac{\partial C}{\partial a_i} ,\xi_i\right ]  \right ] dt \\
    &\hspace{50mm}+  \sum_l [\sigma_l,\mu_i]  \circ dW_t^{i,l} = 0\\
    \mathbb da_i &+ [\xi_i, a_i]\, dt + \theta\,\left  [\frac{\partial C}{\partial \mu_i},  \left [\frac{\partial C}{\partial \mu_i},\chi(  \boldsymbol a)_i\right ]- \left [\frac{\partial C}{\partial a_i}, \xi_i\right ]  \right ]  dt   + \sum_l [{\sigma_l},a_i]  \circ dW_t^{i,l} = 0\,,
    \end{split}
    \label{SP-SD}
\end{align}
where 
\begin{align*}
  \chi(\boldsymbol a)_i := \frac{\partial h}{\partial a_i} -\left ( \mathbb D^{-\frac12} \mathbb A \mathbb D^{-\frac12} \boldsymbol a\right )_i\, . 
\end{align*}
The dissipative terms parametrized by $\theta$ are of double bracket form, and preserve the structure of the coadjoint orbit. 
Their complicated form is due to the semi-direct product structure, see \cite{gaybalmaz2013selective, gaybalmaz2014geometric, arnaudon2016noise} for more details. 
The important difference here is in the $\chi(\boldsymbol a)$ term which contains interactions between the neighbouring spins on the network and appears in the dissipative terms. 
These terms are crucial for the existence of the Gibbs distribution \eqref{P_infty}. 

\section{Example II: Networks of heavy tops} \label{section-HT}

We now consider the case $G=SO(3)$, and take the coadjoint representation of $SO(3)$ on $\mathfrak{so}^*(3) \cong \mathbb R^3$. The neighbours are coupled by the orientation of a fixed vector $\boldsymbol \Gamma_0 \in \mathbb R^3$ rotated around by the $SO(3)$ action. This breaks the symmetry in our Lagrangian so we extend our configuration manifold to $SO(3) \times \mathbb R^3$ such that the extended Lagrangian is invariant under the diagonal $SO(3)$ action on $SO(3) \times \mathbb R^3$. The Lie-Poisson equation that we obtain via symmetry reduction has the Lie-Poisson structure of the semi-direct product group $SE(3) \cong SO(3) \circledS \mathbb R^3$, hence we call this system the ``heavy top network" as opposed to the rigid body network obtained via momentum coupling.

\subsection{Equations of motion}

Starting with the full Lagrangian
\begin{align}
L(\boldsymbol R, \dot{\boldsymbol R}, \mathbf \Gamma_0) = \frac12 \sum_i \langle \dot{R}_i, \mathbb I_i^{R_i} \,\dot{R}_i \rangle + \frac12 \sum_i \sum_{j \sim i} \frac{1}{\sqrt{d_id_j}} \langle R_i^{-1} \mathbf \Gamma^i_0, \mathbb J_{ij}R_j^{-1} \mathbf \Gamma^j_0 \rangle\, ,  
\label{HT-Lagrangian-full}
\end{align}
where $(R_i,\dot{R}_i) \in TSO(3), \boldsymbol \Gamma_0^i \in \mathbb R^3$ and $ \mathbb I_i^{R_i} = R_i \,\mathbb I_i R_i^{-1}$ (interpreted as matrix multiplication), we get the reduced Hamiltonian
\begin{align} \label{h-HT}
&h(\overline {\boldsymbol \Pi}, \overline {\boldsymbol \Gamma}) = \frac12 \left( \overline {\mathbf \Pi} \cdot \overline{\mathbb I}^{-1} \overline {\mathbf \Pi}  - \overline {\mathbf \Gamma}\cdot \mathbb D^{-\frac12} \mathbb A \mathbb D^{-\frac12} \,\overline {\mathbf \Gamma} \right)\,,
\end{align}
where $ \overline {\mathbf \Pi} = (\mathbf \Pi_1, \ldots, \mathbf \Pi_N)$, $ \overline {\mathbf \Gamma} = (\mathbf \Gamma_1, \ldots, \mathbf \Gamma_N)$, $\mathbf \Pi_i = R_i^{-1}\dot{R}_i$ and $\mathbf \Gamma_i = R_i^{-1} \mathbf \Gamma^i_0$.
Taking the semi-direct product Lie-Poisson structure
\begin{align}
\{f, g\}^-_{LP}(\overline {\mathbf \Pi}, \overline {\mathbf \Gamma}) = -\sum_i \left(\mathbf \Pi_i \cdot \frac{\delta f}{\delta \mathbf \Pi_i} \times  \frac{\delta g}{\delta \mathbf \Pi_i} + \mathbf \Gamma_i \cdot \left(\frac{\delta f}{\delta \mathbf \Pi_i} \times  \frac{\delta g}{\delta \mathbf \Gamma_i} - \frac{\delta g}{\delta \mathbf \Pi_i} \times  \frac{\delta f}{\delta \mathbf \Gamma_i}\right)\right)\, ,
\end{align}
we get the Lie-Poisson equation
\begin{align}
  \begin{split}
\dot{\mathbf \Pi}_i &= \mathbf \Pi_i\times \mathbf \Omega_i  - \mathbf \Gamma_i \times \left( \mathbb D^{-\frac12} \mathbb A \mathbb D^{-\frac12} \,\overline {\mathbf \Gamma} \right)_i\\ 
\dot{\mathbf \Gamma}_i &= \mathbf \Gamma_i\times \mathbf \Omega_i \label{HT-eq}\, , 
  \end{split}
\end{align}
for $i=1, \ldots, N$, where $\mathbf \Omega_i := \mathbb I_i^{-1} \mathbf \Pi_i$ is the angular velocity at node $i$.
One can check that the Casimirs for this bracket are given by
\begin{align}
 C_{i,1} = \mathbf \Pi_i \cdot \mathbf \Gamma_i\, ,\quad  C_{i,2} = \|\mathbf \Gamma_i \|^2\, , 
\end{align}
so in particular, the sums
\begin{align}
  C_1 = \sum_i C_{i,1} = \sum_i \mathbf \Pi_i \cdot \mathbf \Gamma_i\, , \quad    C_2 = \sum_i C_{i,2} = \sum_i  \| \mathbf \Gamma_i\|^2\, , 
\end{align}
are conserved by the dynamics. Now, the coadjoint orbits of the heavy top network are given as follows.

\begin{theorem}
The coadjoint orbits $\mathcal O = \mathcal O_1 \times \cdots \times \mathcal O_N$ of the heavy top network are given by,
\begin{align}
\mathcal O_i = \left\{(\mathbf \Pi_i, \mathbf \Gamma_i) \in \mathbb R^{3}\times \mathbb R^{3} : C_{i,1} = c_{i,1}, \, C_{i,2} = c_{i,2} \right\} \cong TS^2\, ,
\end{align}
if $c_{i,2} \neq 0$, which is a four-dimensional submanifold and
\begin{align}
\mathcal O_i = \left\{(\mathbf \Pi_i, \boldsymbol 0) \in \mathbb R^{3}\times \mathbb R^{3} : \|\mathbf \Pi_i\|^2 = const \right\} \cong S^2\, ,
\end{align}
if $c_{i,2} = 0$, which is a two-dimensional submanifold, unless $\|\mathbf \Pi_i\|^2 = 0$.
\end{theorem}
\begin{proof}
We refer to theorem {\color{red} 1.2} in \cite{ratiu1981euler} for the proof in the single body case. The general multi-body case considered here is an easy extension.
\end{proof}

\subsection{Noise and dissipation}

From the general equation \eqref{SP-SD}, we use $C_{i,1} = \boldsymbol \Pi_i \cdot \boldsymbol \Gamma_i$ as the Casimir for the double bracket dissipation to obtain the stochastic equation 
\begin{align}
    \begin{split}
      \mathbb d\boldsymbol \Pi_i &+\left(\boldsymbol \Omega_i \, dt + \sum_l \boldsymbol{\sigma}_l\circ dW^l_t \right) \times\boldsymbol \Pi_i - \left(\boldsymbol \Gamma_i \times \boldsymbol \chi(\boldsymbol \Gamma)_i \right) dt \\
      &+\theta\, \boldsymbol\Gamma_i\times(\boldsymbol\Omega_i\times\boldsymbol\Gamma_i)\, dt +\theta\,  \left [ \boldsymbol \Pi_i\times ( \boldsymbol\chi(\boldsymbol \Gamma)_i\times \boldsymbol \Gamma_i) - \boldsymbol\Pi_i\times (\boldsymbol\Pi_i\times \boldsymbol\Omega_i)\right]\, dt= 0 \\
      \mathbb d\boldsymbol \Gamma_i &+ \left(\boldsymbol \Omega_i \, dt + \sum_l\boldsymbol{\sigma}_l\circ dW^l_t \right) \times\boldsymbol \Gamma_i + \theta\, \left [ \boldsymbol\Gamma_i\times ( \boldsymbol \chi(\boldsymbol \Gamma)_i\times \boldsymbol\Gamma_i) - \boldsymbol \Gamma_i\times(\boldsymbol\Pi_i\times \boldsymbol\Omega_i)\right ]\, dt = 0 \,.
    \end{split}
\end{align}
where
\begin{align}
  \boldsymbol \chi(\boldsymbol \Gamma) = \chi+ \sum_{j\sim i}\frac{1}{\sqrt{d_i d_j}}\mathbb J_{ij} \boldsymbol \Gamma_j\, . 
\end{align}
Notice that the other Casimir $C_{i,2}$ is also preserved by this dissipation, thus we did not include it here.
However, it is possible that including the Casimir $C_{i,2}$ would change the behaviour of the system, but we leave this investigation for future work. 

\subsection{Equilibrium solutions}

We now study the equilibrium solutions of the heavy top network and its corresponding stability properties. The classification of equilibrium solutions into ferromagnetic and anti-ferromagnetic states and its nonlinear stability property are summarised in the following proposition.

\begin{proposition} \label{HT-equib}
The relative equilibrium solutions $\overline {\mathbf \Gamma}_e$ of a heavy top network correspond to the eigenvectors of $\mathbb L(\lambda_1^e)$ with
$\overline {\mathbf \Pi}_e = -\lambda_1^e \,\overline {\mathbb I} \,\overline {\mathbf \Gamma}_e$ and the equilibrium solution corresponding to the lowest eigenvalue $-\lambda_2^e$ of $\mathbb L(\lambda_1^e)$ is nonlinearly stable. Furthermore, if $\mathbb J_{ij} = \mathbb J$ and $\mathbb I_i = \mathbb I$, then there exists $3N$ linearly independent equilibrium configurations such that
\begin{enumerate}
\item three are ferromagnetic, i.e. $\mathbf \Gamma_i^e = \sqrt{d_i} \mathbf \Gamma^e$ for $i=1, \ldots, N$, where $\mathbf \Gamma^e$ is an eigenvector of $\mathbb I_{\text{ext}} := -(\lambda_1^e)^2 \,\mathbb I^{-1} - \mathbb J$ and
\item the remaining $3N-3$ are anti-ferromagnetic, i.e. $\sum_{i=1}^N \sqrt{d_i} \mathbf \Gamma_i^e = 0$.
\end{enumerate}
\end{proposition}

In order to investigate the stability of the other equilibrium configurations, we move to linear stability analysis.

\subsubsection{Linear stability}

Linearising equations \eqref{HT-eq} by taking $\overline {\mathbf \Pi}(t) = \overline {\mathbf \Pi}_e + \epsilon \,\delta \overline {\mathbf \Pi}(t)$ and $\overline{\mathbf \Gamma}(t) = \overline {\mathbf \Gamma}_e + \epsilon \,\delta \overline{\mathbf \Gamma}(t)$ and dropping terms of $O(\epsilon^2)$, we get in term of $\boldsymbol \Gamma_e$ only, 
\begin{align} 
\frac{d}{dt}
\begin{pmatrix}
\delta \overline {\mathbf \Pi} \\
\delta \overline {\mathbf \Gamma}
\end{pmatrix}
=
\begin{pmatrix}
  \lambda_1 \widehat {\mathbf \Gamma}^e -\lambda_1 \mathbb I \widehat{\mathbf \Gamma}^e \,\overline {\mathbb I}^{-1} &  (-\lambda_1^2 \,\overline{\mathbb I} +\lambda_2 \mathbb 1) \widehat{\mathbf \Gamma}^e -\widehat{\mathbf \Gamma}^e  \,\mathbb D^{-\frac12}\mathbb A \mathbb D^{-\frac12} \\
\widehat{\mathbf \Gamma}^e \,\overline {\mathbb I}^{-1}  & \lambda_1 \, \widehat{\mathbf \Gamma}^e
\end{pmatrix}
\begin{pmatrix}
\delta \overline{\mathbf \Pi} \\
\delta \overline{\mathbf \Gamma}
\end{pmatrix}\, , 
\label{lin-stab-HT}
\end{align}
where $\widehat{\mathbf \Pi}^e := \text{diag}(\widehat{\mathbf \Pi}^e_1, \ldots, \widehat{\mathbf \Pi}^e_N)$, $\widehat{\mathbf \Gamma}^e := \text{diag}(\widehat{\mathbf \Gamma}^e_1, \ldots, \widehat{\mathbf \Gamma}^e_N)$ and $\,\widehat{} : (\mathbb R^{3},\times) \rightarrow \mathfrak{so}(3)$ is the standard Lie algebra isomorphism that sends a vector in $\mathbb R^3$ to a skew symmetric matrix in $\mathfrak{so}(3)$. Hence, we can assess the linear stability of an equilibrium solution $(\overline {\mathbf \Pi}_e, \overline{\mathbf \Gamma}_e)$ by looking at the eigenvalues of the matrix on the right hand side of \eqref{lin-stab-HT}.

\begin{figure}[htpb]
  \centering
  \subfigure[$\mathbb I=\mathrm{diag}(1,2,3)$, $\mathbb J=\mathrm{diag}(1,1,1)$]{\includegraphics[scale=0.55]{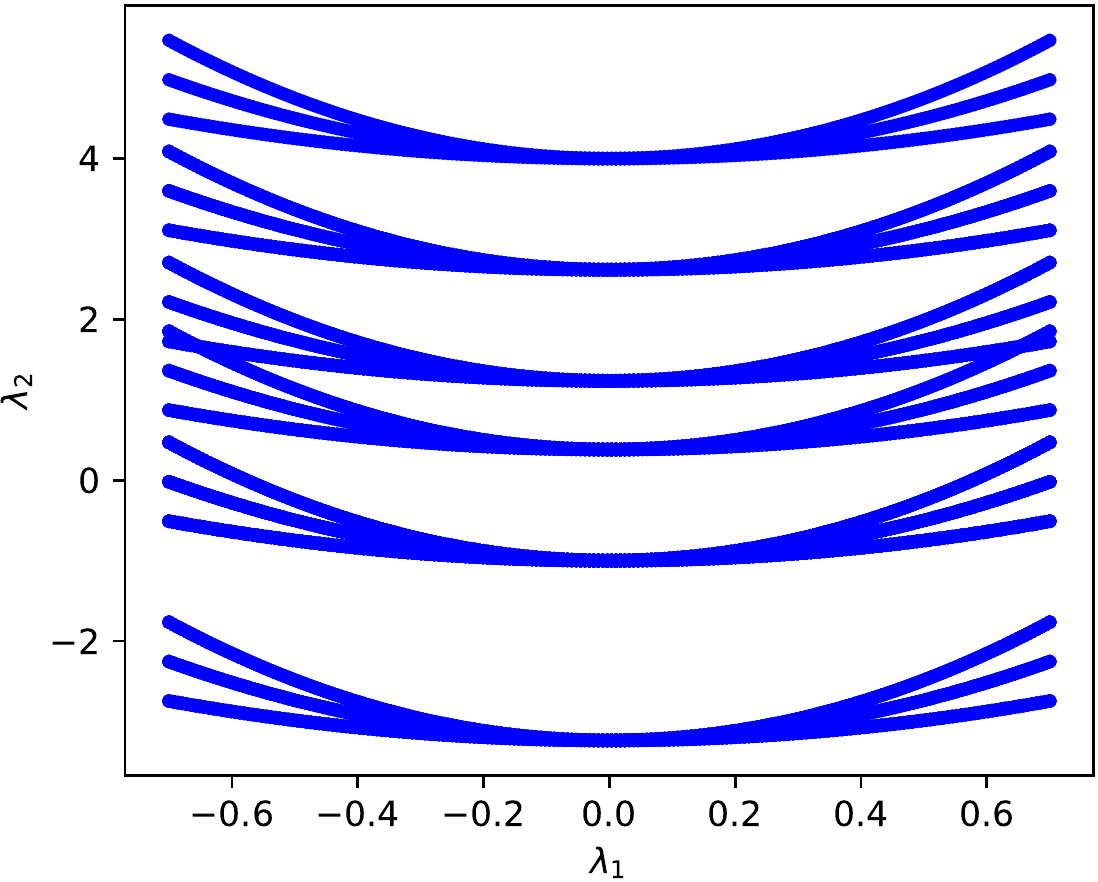}\label{fig:l1l2-HT1}}
\hspace{10pt}
  \subfigure[$\mathbb I=\mathrm{diag}(1,1,1)$, $\mathbb J=\mathrm{diag}(1,2,3)$]{\includegraphics[scale=0.55]{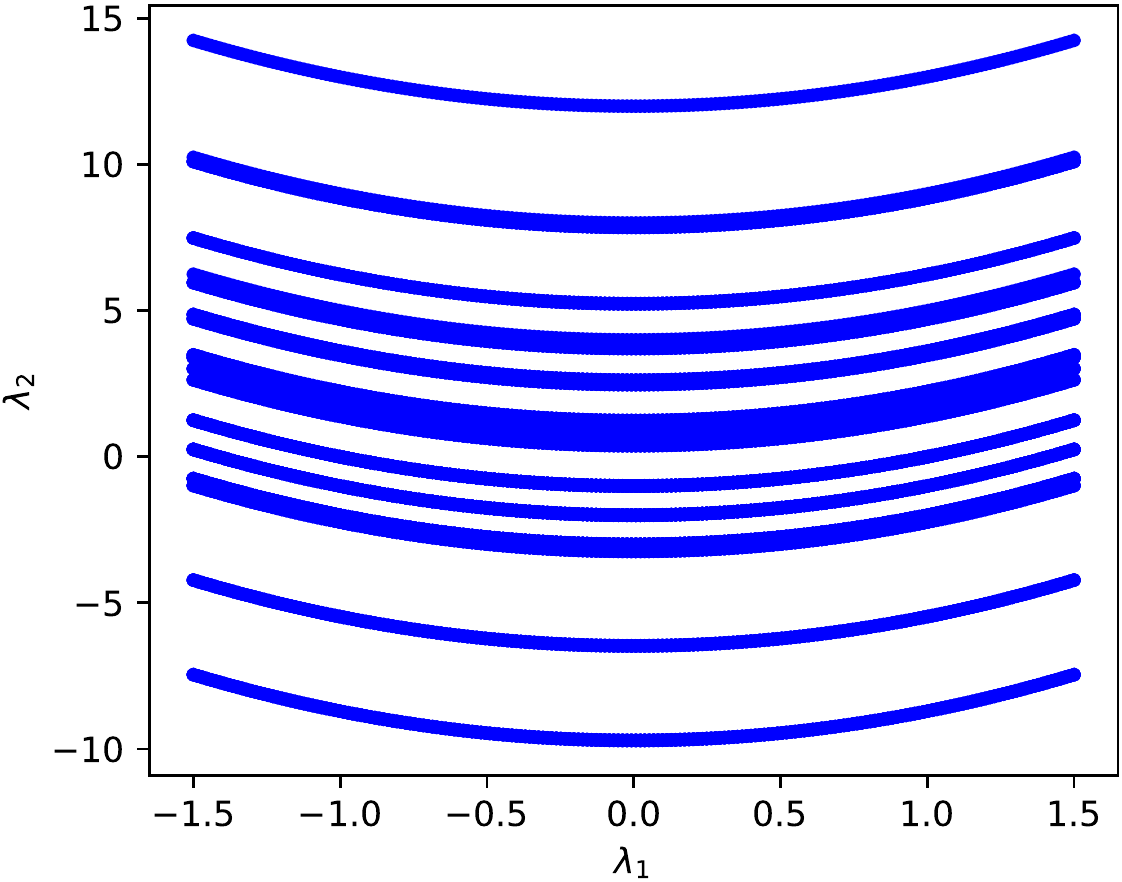}\label{fig:l1l2-HT2}}
  \caption{This figure shows the possible $(\lambda_1, \lambda_2)$-pairs that one can have from solving equation \eqref{L_l1l2}. Fixing $\lambda_1$, the possible values of $\lambda_2$ that solve \eqref{L_l1l2} are the eigenvalues of the graph Laplacian $\mathbb L(\lambda_1)$.
We see that the eigenvalues $\lambda_2$  in the left panel \ref{fig:l1l2-HT1} collapse as $\lambda_1 \rightarrow 0$, which does not occur in the right panel \ref{fig:l1l2-HT2}. }
  \label{fig:l1l2-HT}
\end{figure}

Determining the linear stability of all the equilibrium states of this system is impossible to do analytically, so we will only discuss the numerical results here. 
First, we solve \eqref{L_l1l2} to find all the equilibria of the lattice given a $\lambda_1$, and in figure \ref{fig:l1l2-HT}, we plotted the values of all possible $\lambda_2$ (i.e. minus the eigenvalues of $\mathbb L(\lambda_1)$) as a function of $\lambda_1$. 
In the case $\mathbb I= \mathrm{diag}(1,2,3)$ and $\mathbb J=\mathrm{diag}(1,1,1)$ (figure \ref{fig:l1l2-HT1}), we see that as $\lambda_1\to 0$, all the solutions become degenerate with multiplicity $3$ (i.e. the eigenvalues $\lambda_2$ collapse) but this does not happen in the other case $\mathbb I= \mathrm{diag}(1,1,1)$ and $\mathbb J=\mathrm{diag}(1,2,3)$ (figure \ref{fig:l1l2-HT2}).
This can be explained by the fact that at $\lambda_1=0$, which corresponds to the zero momentum case $\bar{\mathbf \Pi} = 0$, the system is isotropic in case (a) and anisotropic in case (b).


\begin{figure}[htpb]
  \centering
  \subfigure[$\mathbb I=\mathrm{diag}(1,2,3)$, $\mathbb J=\mathrm{diag}(1,1,1)$]{\includegraphics[scale=0.63]{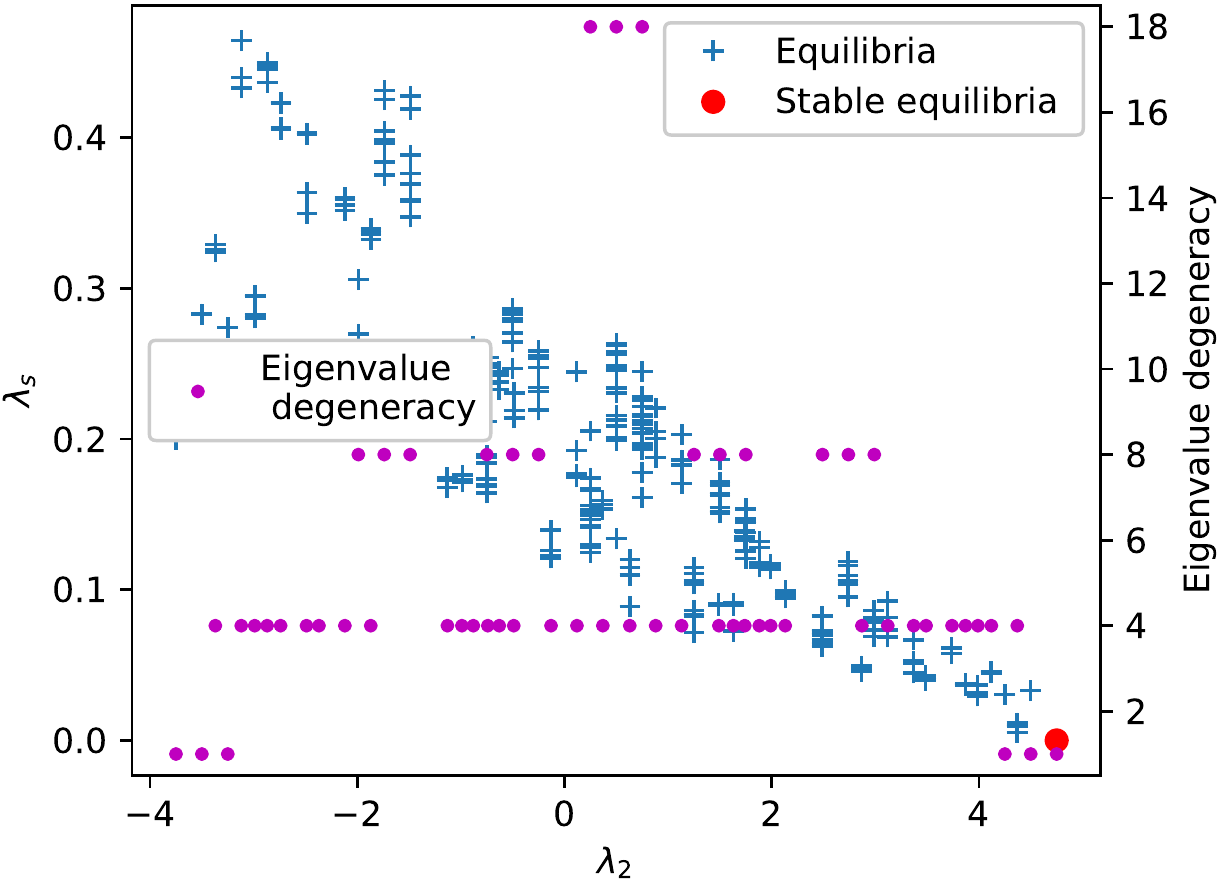}\label{stability_HT1}}
  \subfigure[$\mathbb I=\mathrm{diag}(1,1,1)$, $\mathbb J=\mathrm{diag}(1,2,3)$]{\includegraphics[scale=0.63]{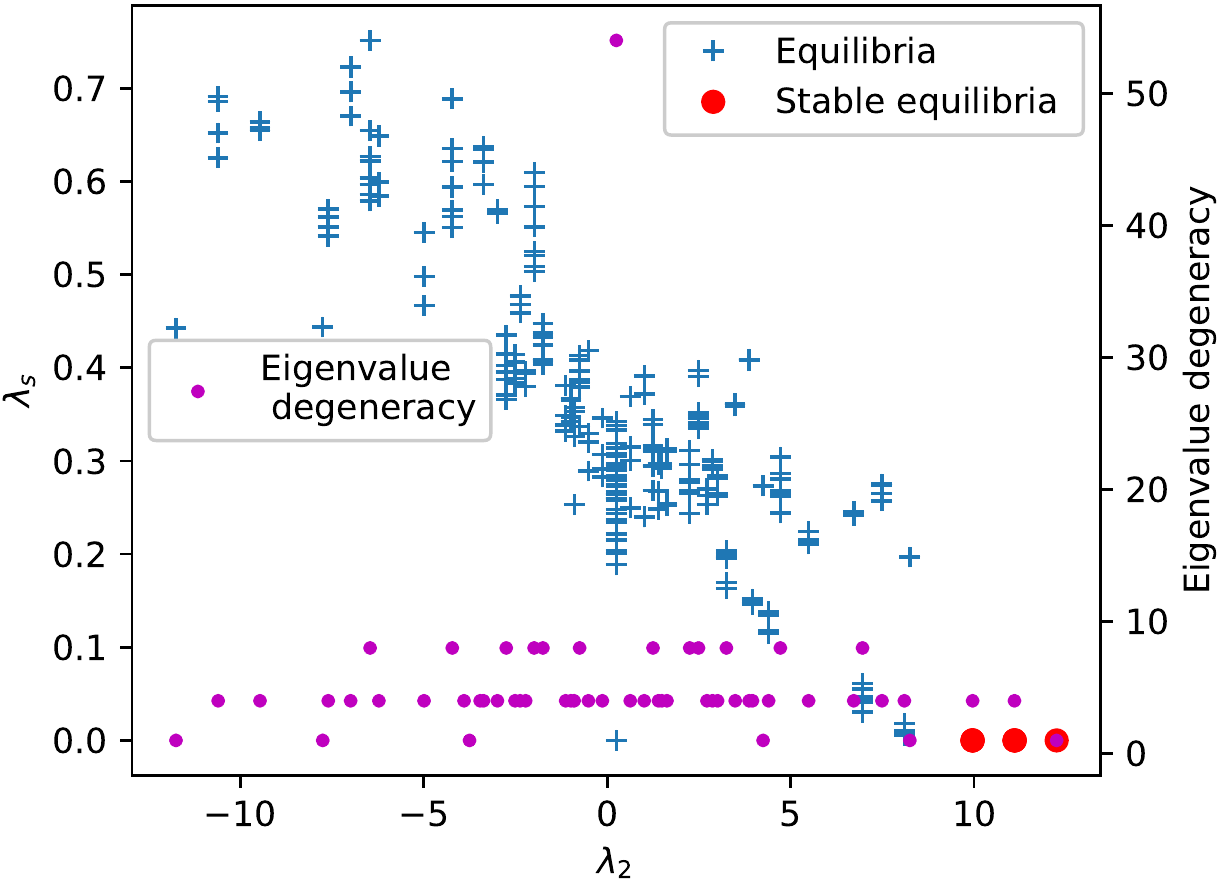}\label{stability_HT2}}
  \caption{The two panels in this figure display the maximum, real part of the eigenvalues $\lambda_S$ of the linearised system \eqref{lin-stab-HT} around the equilibrium configuration $\boldsymbol \Gamma_e$ (eigenvectors of $\mathbb L(\lambda_1)$) plotted against $\lambda_2$ (minus the eigenvalue of $\mathbb L(\lambda_1)$ corresponding to $\boldsymbol \Gamma_e$) for $\lambda_1 = 0.5$.
Again, the blue crosses correspond to unstable equilibria, the red dots correspond to linearly stable equilibria and the purple dots correspond to the multiplicity of $\lambda_2$, whose values are to be read from the right axis. As expected, the configuration corresponding to the highest $\lambda_2$ (or, the lowest eigenvalue of $\mathbb L(\lambda_1)$) is stable in both cases.}
  \label{fig:stability-HT}
\end{figure}

In figure \ref{fig:stability-HT}, we display the largest, real part of the eigenvalues $\lambda_s$ of the linearised system \eqref{lin-stab-HT} plotted against $\lambda_2$ (minus the eigenvalue of $\mathbb L(\lambda_1)$) corresponding to equilibrium solutions $\mathbf \Gamma_e$ of \eqref{lin-stab-HT} for $\lambda_1=0.5$. As expected from proposition \ref{HT-equib}, the lowest eigenvalue state (or, the states with highest $\lambda_2$) in both cases \ref{stability_HT1} and \ref{stability_HT2} are linearly stable. This corresponds to a ferromagnetic equilibria along the $\Gamma_3$-axis in both cases. For $\mathbb I=\mathrm{diag}(1,2,3)$ and $\mathbb J=\mathrm{diag}(1,1,1)$ (figure \ref{stability_HT1}), 
the two states near this configuration with multiplicity 1 ($\lambda_2 \approx 4.2$ and $\lambda_2 \approx 4.5$) correspond to the other two ferromagnetic equilibria given in proposition \ref{HT-equib}, which we call states A and B and these are seen to be unstable. There also exists a subspace of anti-ferromagnetic equilibria between these two states ($\lambda_2 \approx 4.35$) with multiplicity $4$ that have a very small but positive $\lambda_s$. We call this state $C$. As we took $\lambda_1 \rightarrow 0$, we saw that state $C$ becomes linearly stable for some small $\lambda_1$, and states $A$ and $B$, while they were always found to be unstable, their corresponding values of $\lambda_s$ tended to $0$ (i.e. approaching a linearly stable state).
In the other case $\mathbb I=\mathrm{diag}(1,1,1)$ and $\mathbb J=\mathrm{diag}(1,2,3)$ (figure \ref{stability_HT2}), the linearly stable equilibria excluding the nonlinearly stable state were found to be anti-ferromagnetic states with the `argyle' pattern similar to that in figure \ref{fig:RB-sols}.

\subsection{Numerical experiments}\label{HT-numerics}

Before looking at phase transitions, we will present here an interesting phenomenon that arises in the heavy top lattice when $\mathbb I= \mathrm{diag}(1,2,3)$ and $\mathbb J = \mathrm{diag}(1,1,1)$. 
In figure \ref{fig:HT-stab-diss}, we show several simulations of the deterministic $20\times 20$ heavy top lattice starting from different initial conditions, with or without the double bracket dissipation.
The initial conditions are taken to be nearly ferromagnetic, with spins $\mathbf \Gamma_i$ aligned to a fixed direction, except for a small perturbation at two nodes. 
This way, we can numerically assess the stability of the ferromagnetic equilibria. 
From proposition \ref{HT-equib}, we know that position $(0,0,1)$ (i.e. the $\Gamma_3$-axis) is nonlinearly stable, and we also observed this in our simulations. Hence, we will not display the corresponding plots here as it is not very interesting.  
Instead, we will display the simulations starting close to the two other ferromagnetic equilibria $(0,1,0)$  in figure \ref{g_010} and $(1,0,0)$ in figure \ref{g_100}, which were seen to be unstable with or without dissipation. 

There are two interesting behaviors of the system observed from the plots. 
The first, which is displayed in panel \ref{g_100}, is that the solution starting close to the shortest axis $(1,0,0)$ with double bracket dissipation is stuck in this position for a while,
then gets stuck in the middle axis position $(0,1,0)$ briefly, before relaxing to the lowest energy state $(0,0,1)$. This can also be seen in panel \ref{g_010} where the solution starting close to the middle axis $(0,1,0)$ is stuck there for a while before relaxing to the lowest energy state $(0,0,1)$. 
Although the double bracket dissipation can never stabilize an unstable equilibrium solution, as proven in Bloch et al. \cite{bloch1994dissipation}, this indicates some kind of transient stability, or metastability of the shortest and middle axis equilibrium, which will be observed again in the phase transition plots in section \ref{PT-section}.
We note that this phenomenon is not apparent in the absence of dissipation and is only observed clearly in the presence of dissipation.
Apart from these special cases, for instance when we start from $(1,1,1)$, which is not an equilibrium, the dissipation will always drive the system towards the lowest energy position without getting stuck (see figure \ref{g_111}). 

\begin{figure}[htpb]
  \centering
  \subfigure[Position, $ \Gamma(0)\approx  (0,1,0)$ ] {\includegraphics[scale=0.45]{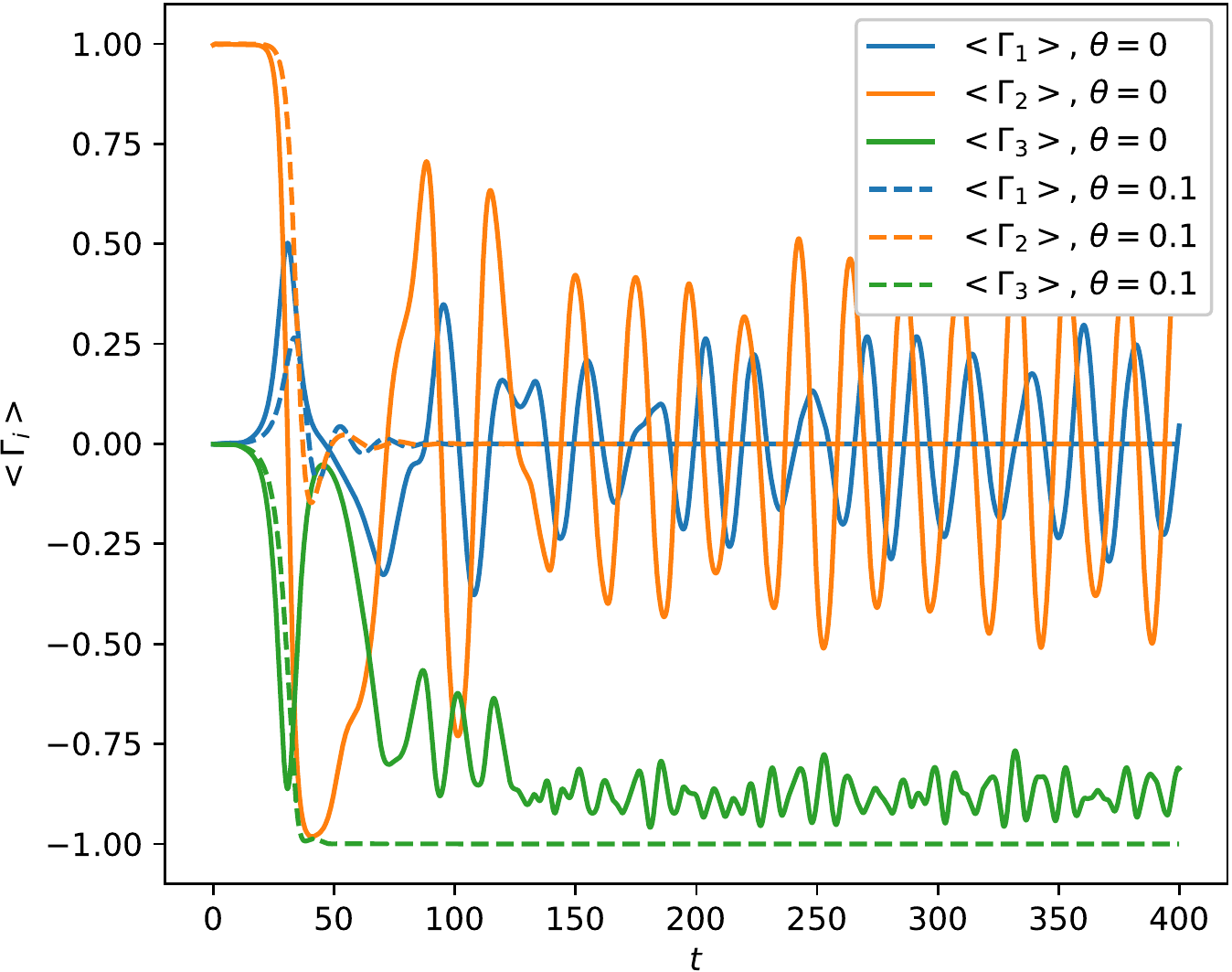}\label{g_010}}
\hspace{10pt}
  \subfigure[Energy, $\Gamma(0)\approx (0,1,0)$ ]{\includegraphics[scale=0.45]{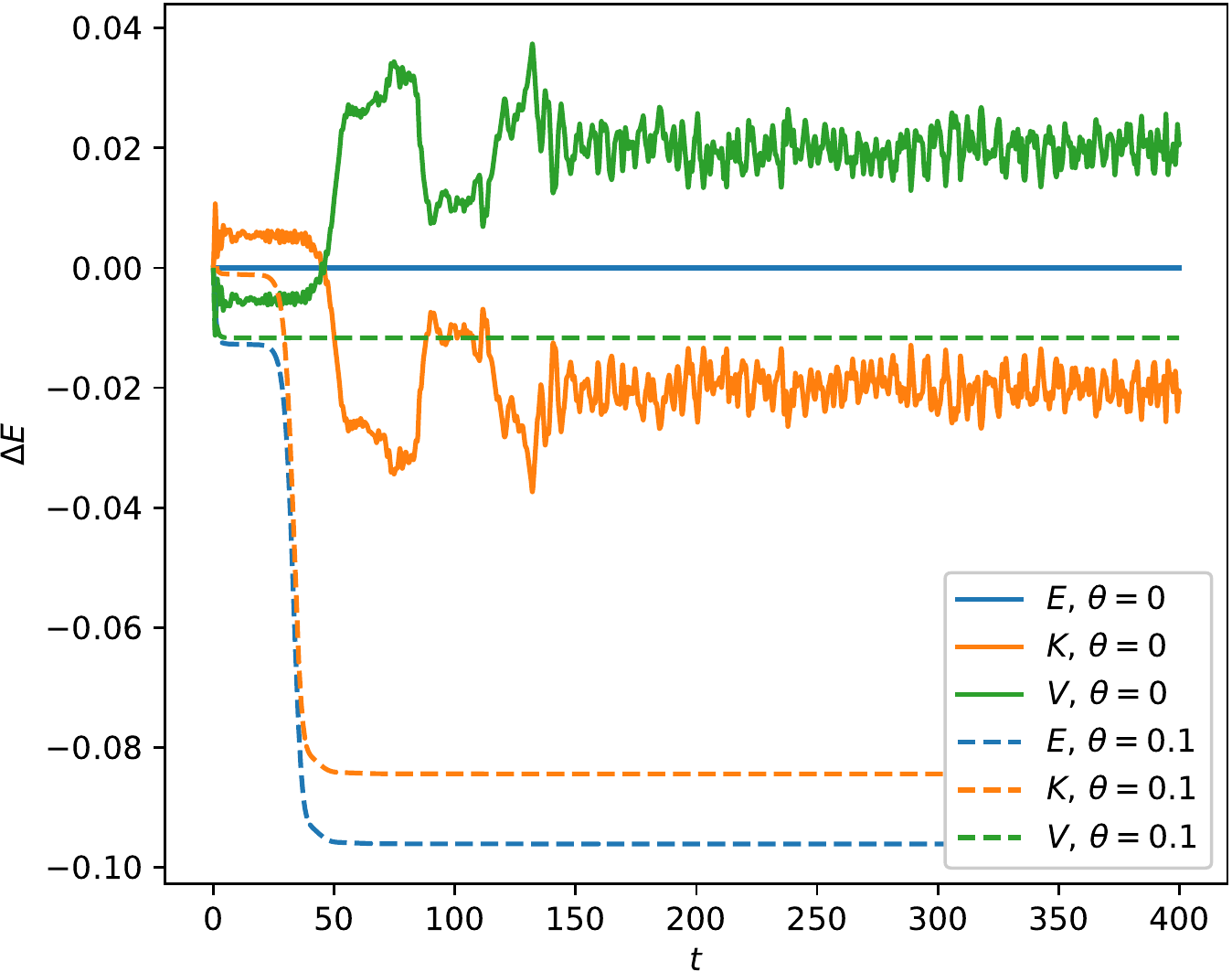}\label{e_010}}
  \subfigure[Position, $\Gamma(0)\approx (1,0,0)$]{\includegraphics[scale=0.45]{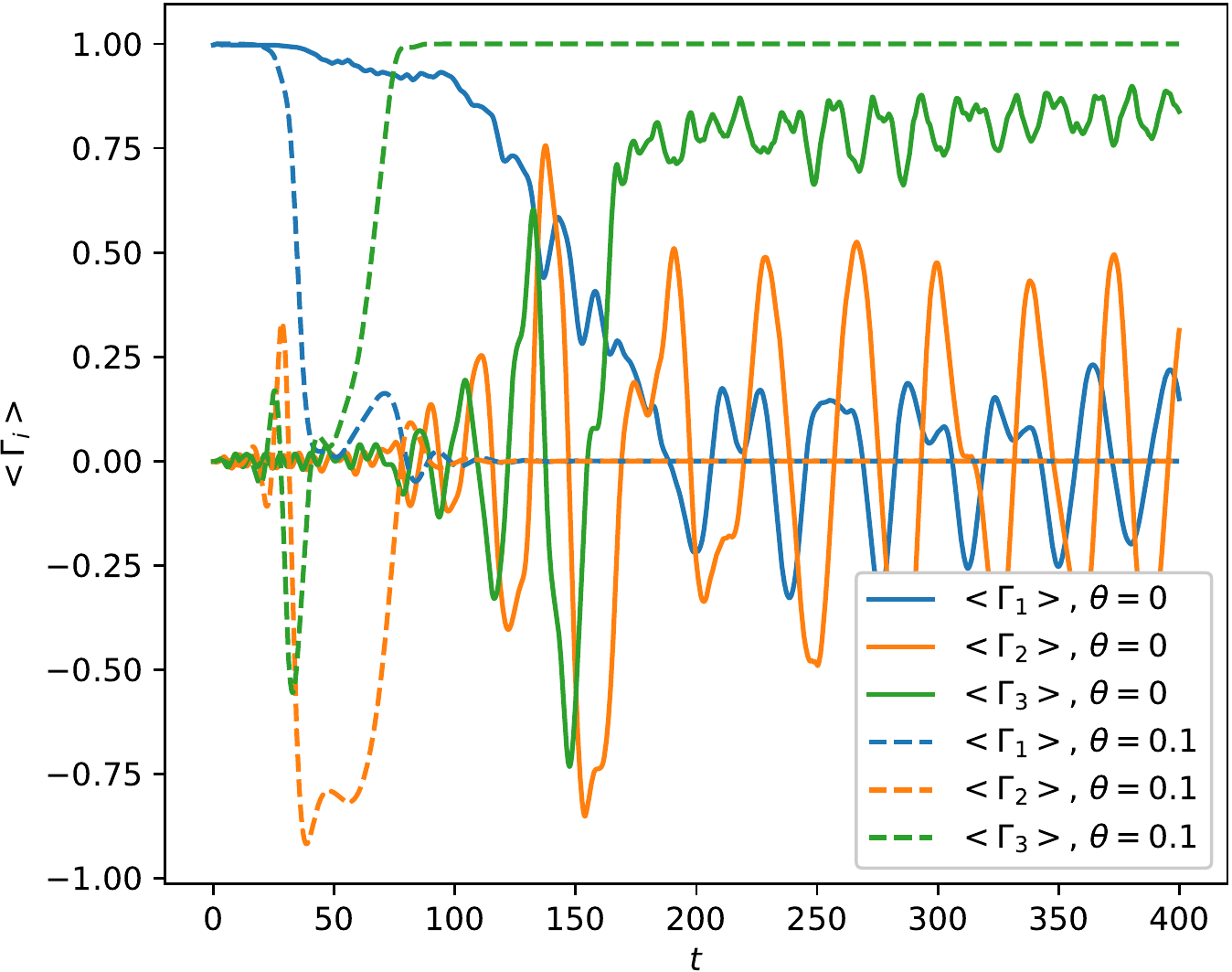}\label{g_100}}
\hspace{10pt}
  \subfigure[Energy, $\Gamma(0)\approx (1,0,0)$]{\includegraphics[scale=0.45]{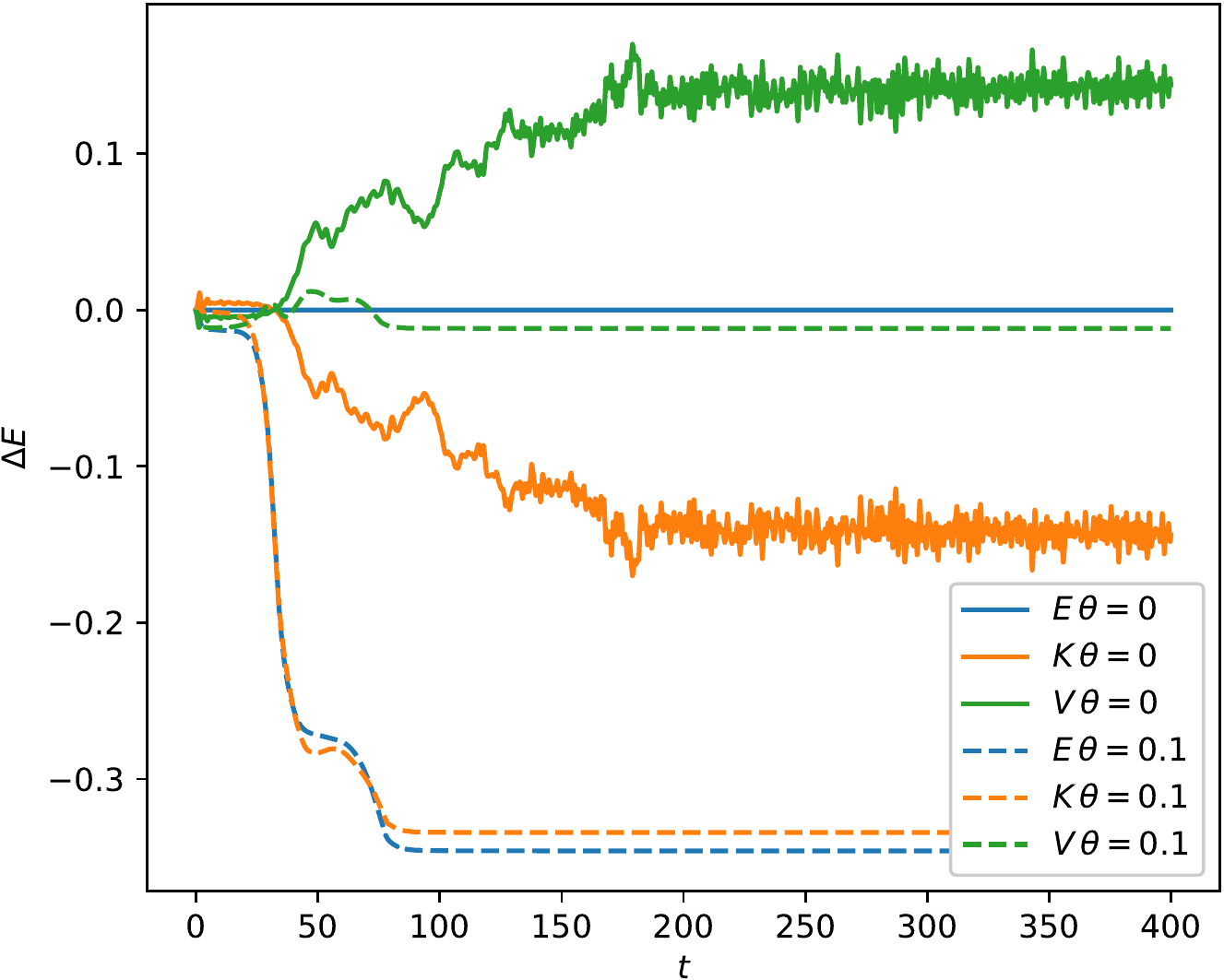}\label{e_100}}
  \subfigure[Position, $\Gamma(0)\approx (1,1,1)$ ]{\includegraphics[scale=0.45]{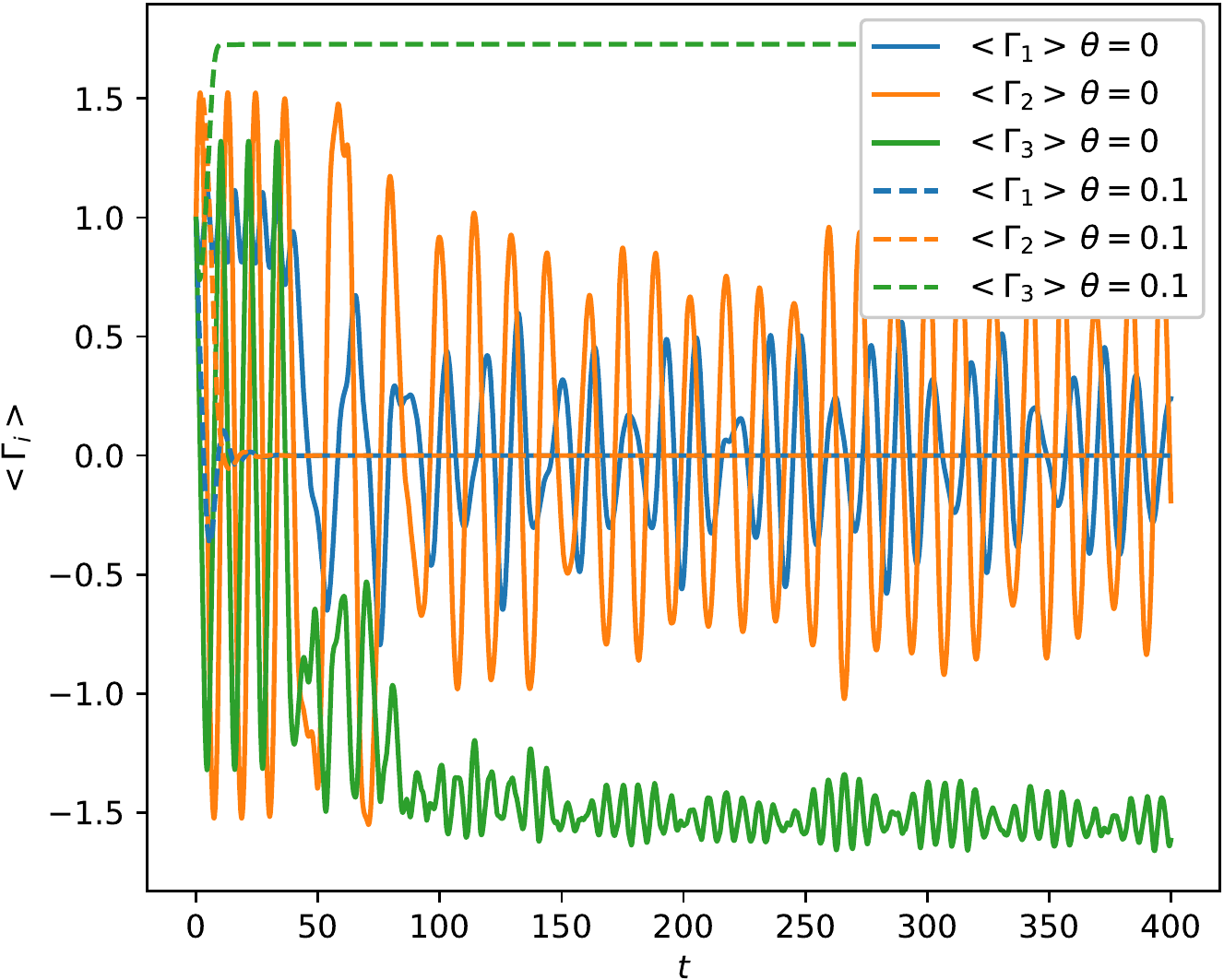}\label{g_111}}
\hspace{10pt}
  \subfigure[Energy, $\Gamma(0)\approx (1,1,1)$]{\includegraphics[scale=0.45]{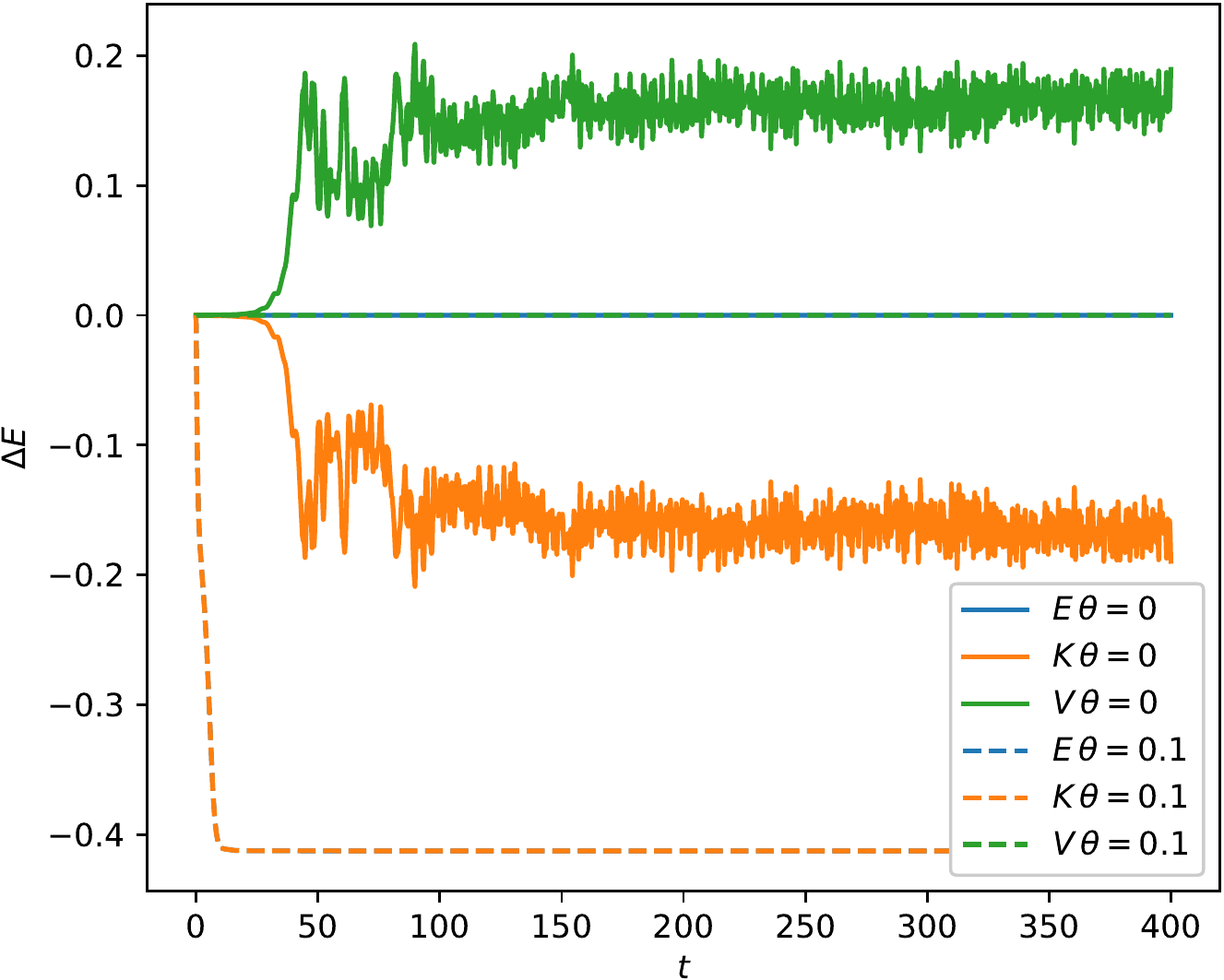}\label{e_111}}
  \caption{This figure shows the time evolution of the averged spins (\ref{g_010}, \ref{g_100}, \ref{g_111}) and energy (\ref{e_010}, \ref{e_100}, \ref{e_111}) in a $20\times 20$ heavy top lattice, starting from different initial conditions, with dissipation (dashed lines) or without dissipation (bold lines). 
The blue, orange and green lines in the left panel correspond to the $\Gamma_1$, $\Gamma_2$ and $\Gamma_3$-components of the spins averaged over the lattice respectively and the same colors in the right panel correspond to the total, kinetic and potential energy.
  We observe metastability of the $\Gamma_1$ and $\Gamma_2$ axes in panels \ref{g_100} and \ref{g_010} in the dissipative case, and a partial synchronisation around the $\Gamma_3$-axis in the non-dissipative case in \ref{g_010}, \ref{g_100} and \ref{g_111}. 
  }

  \label{fig:HT-stab-diss}
\end{figure}

The second observation that we made is the {\em partial synchronisation} phenomenon seen in the non-dissipative simulations (see the bold lines in \ref{g_010}, \ref{g_100} and \ref{g_111}).
The corresponding energy (total, kinetic and potential) plots are given by the bold lines in \ref{e_010}, \ref{e_100} and \ref{e_111} to demonstrate the validity of the simulations. 
In all of the cases considered, we observe that after some time, the spins relax to a state where on average, it oscillates closely around the lowest energy configuration (i.e. ferromagnetic state along the $\Gamma_3$-axis), despite the absence of dissipation. We call this phenomenon ``partial synchronisation''.
This result is rather surprising as, by Liouville's theorem, the volume of the phase space is conserved by the dynamics, so by the Poincar\'e recurrence theorem, any states starting near the equilibrium should return sufficiently close to it after a finite time.
However, in our simulations, the trajectories after some time seem to get stuck in another part of the phase space without returning to a region close to the initial conditions. 
This is very counter-intuitive from the original dynamics of the rigid body, where an initial condition near the unstable (middle) axis will eventually come back near it, after traversing a long trajectory on the momentum sphere. 
One possible avenue of investigation of this phenomenon is to study the local dynamics of the lattice where almost periodic motions were observed (see the videos in \url{http://wwwf.imperial.ac.uk/~aa10213/}). This phenomenon may also be tied to the complex interaction between the $\mathbf \Gamma$ and $\mathbf \Pi$ variables and the non-compactness of the phase space.
However, we will not study this further in the present work and instead leave it for future investigations. 

\section{Temperature phase transitions}\label{PT-section}

In this last section, we will study phase transitions in the two examples that we constructed above, namely, the rigid body network and the heavy top network.
There exist various types of phase transitions, but here we will focus on second-order phase transitions that arise from varying the temperature of the system. This is characterized by a transition from an orderly state of the spins, measured by how much they are aligned with each another, to a completely disordered state.
In order to detect a phase transition, one can either apply a mean field approximation of the model and try to solve it analytically or perform direct numerical simulations of the underlying dynamics of the lattice. 
Many other methods are available but are out of the scope of this first investigation. 

\subsection{Mean field approximation}
The mean field approximation relies on the assumptions that (1) the microscopic system is identical at each node and (2) the spins (momentum or position) of the rigid bodies in statistical equilibrium are close to its mean. 
This approximation is increasingly accurate if each site has more neighbours, which is the case for high dimensional lattices. 
In two dimensions, the approximation fails to properly assess the critical temperature and the corresponding critical exponents, but still, give a good indication of the presence of a phase transition. 

\subsubsection{Mean field approximation of the rigid body network}
We assume that the system is identitcal at each node, i.e. $\mathbb I_i = \mathbb I$ for all $i=1, \ldots, N$, and that the interactions between the neighbours are identical, that is, $\mathbb J_{ij} = \mathbb J$ for all $i,j = 1, \ldots, N$ for the mean field approximation to be valid. Now, define the {\em averaged momentum}
\begin{align}
\langle \boldsymbol \Pi \rangle := \frac1N \sum_{i=1}^N \int_{\boldsymbol{\mathcal O}} \mathbf \Pi_i \, \mathbb P_{\infty} (\overline {\mathbf \Pi}) d \overline {\mathbf \Pi}, \quad \mathbb P_{\infty} (\overline {\mathbf \Pi}) = Z_{RB}^{-1} e^{-\beta h(\overline {\mathbf \Pi})},
\end{align}
where $\beta = \frac{2 \theta}{\sigma^2} = T^{-1}$ is the inverse temperature, $\mathbb P_{\infty}(\cdot)\,d\overline {\mathbf \Pi}$ is the Gibbs measure \eqref{P_infty} and $\boldsymbol{\mathcal O} = \mathcal O_1 \times \cdots \times \mathcal O_N$ is the total coadjoint orbit. Since the system is assumed to be identical at each node, we have $\mathcal O_i = \mathcal O$ for all $i=1, \ldots, N$.

Since we assume that the spins $\mathbf \Pi_i = \Braket{\boldsymbol \pi} + \delta \mathbf \Pi_i$ are close to the mean $\langle \boldsymbol \pi \rangle$, we linearise the interaction Hamiltonian around $\Braket{\boldsymbol \Pi}$ to get 
\begin{align}
      h^{\text{int}}(\overline {\mathbf \Pi})&=  -\sum_{i} \sum_{j \sim i} \frac{1}{2d} ( \boldsymbol \Pi_i - \Braket{\boldsymbol \Pi} + \Braket{\boldsymbol \Pi}) \cdot \mathbb J (\boldsymbol \Pi_j-\Braket{\boldsymbol \Pi}+\Braket{\boldsymbol \Pi}) \nonumber\\
      &=  -\frac{1}{2d} \sum_{i} \sum_{j \sim i} \left(\delta \mathbf \Pi_i \cdot \mathbb J \, \delta \mathbf \Pi - \frac{1}{2d} \left ( \Braket{\boldsymbol \Pi } \cdot \mathbb J \boldsymbol\Pi_j +\boldsymbol\Pi_i \cdot \mathbb J \Braket{\boldsymbol\Pi}\right) \right) - \frac12 \sum_i \Braket{\boldsymbol\Pi}\cdot \mathbb J \Braket{\boldsymbol\Pi}\, \nonumber\\
&\approx -\sum_{i} \Braket{\boldsymbol \Pi } \cdot \mathbb J \boldsymbol\Pi_i, \nonumber
\end{align}
where the first term in the second line is neglected since it is quadratic in $\delta \mathbf \Pi_i$ and is therefore very small, and the last term is also neglected since it is constant and does not contribute to the dynamics.
We therefore obtain the complete mean field Hamiltonian
\begin{align}
  h_{\rm mf} =\sum_{i} \left( \frac12  \boldsymbol \Pi_i  \cdot \mathbb I^{-1} \boldsymbol\Pi_i  - \Braket{\boldsymbol \Pi } \cdot \mathbb J \boldsymbol\Pi_i \right) \, . 
\end{align}
Notice that the kinetic energy is still exact. 
From this Hamiltonian, the Gibbs distribution and the partition function take a simpler form,
\begin{align}
\mathbb P_{\infty}^{\text{mf}} (\overline {\mathbf \Pi}) = \frac{1}{Z_{RB}^{\text{mf}}} e^{-\beta h_{\text{mf}}(\overline {\mathbf \Pi})}, \quad 
  Z^{\rm mf}_{RB}= \left (\int_{\mathcal O} e^{-\beta \left ( \frac12 \boldsymbol \Pi\cdot \mathbb I^{-1} \boldsymbol \Pi - \boldsymbol \Pi\cdot  \mathbb J\Braket{\boldsymbol \Pi }\right )}d\boldsymbol \Pi\right)^N\,,
\end{align}
and one can check that the average momentum $\Braket{\boldsymbol \Pi}$ simplifies to
\begin{align}
  \Braket{\boldsymbol \Pi} = \frac{\int_{\mathcal O} \mathbf \Pi \, e^{-\beta \left ( \frac12 \boldsymbol \Pi\cdot \mathbb I^{-1} \boldsymbol \Pi - \boldsymbol \Pi\cdot  \mathbb J\Braket{\boldsymbol \Pi }\right )}d\boldsymbol \Pi}{\int_{\mathcal O} e^{-\beta \left ( \frac12 \boldsymbol \Pi\cdot \mathbb I^{-1} \boldsymbol \Pi - \boldsymbol \Pi\cdot  \mathbb J\Braket{\boldsymbol \Pi }\right )}d\boldsymbol \Pi}, 
\label{mean-momentum}
\end{align}
which is now an implicit equation for the order parameter $\Braket{\boldsymbol \Pi}$. 
This equation is difficult to solve analytically, as it involves integrals over the momentum sphere but can be numerically estimated using Monte-Carlo integration. 
We will display the numerical approximation of $\Braket{\boldsymbol \Pi}$ in the next section, compared to the full simulations.

\subsubsection{Mean field approximation of the heavy top network}
In a similar fashion, we can derive the mean field approximation of the heavy top network. In this case, we take our order parameter to be the {\em averaged position}, defined by
\begin{align}
\langle \boldsymbol \Gamma \rangle := \frac1N \sum_{i=1}^N \int_{\boldsymbol{\mathcal O}_1} \int_{\boldsymbol{\mathcal O}_2} \mathbf \Gamma_i \, \mathbb P_{\infty} (\overline {\mathbf \Pi}, \overline {\mathbf \Gamma}) d\overline {\mathbf \Pi} d \overline {\mathbf \Gamma}, \quad \mathbb P_{\infty} (\overline {\mathbf \Pi}, \overline {\mathbf \Gamma}) = Z_{HT}^{-1} e^{-\beta h(\overline {\mathbf \Pi}, \overline {\mathbf \Gamma})},
\end{align}
where $\boldsymbol{\mathcal O}_1 = \mathcal O_{1,1} \times \cdots \times \mathcal O_{N,1}$ and $\boldsymbol{\mathcal O}_2 = \mathcal O_{1,2} \times \cdots \times \mathcal O_{N,2}$ are coadjoint orbits corresponding to level sets of the Casimirs $C_{i,1}$ and $C_{i,2}$ respectively. Again, since we assume the system to be identical at each node, we take $\mathcal O_{i,1} = \mathcal O_1$ and $\mathcal O_{i,2} = \mathcal O_2$ for all $i=1, \ldots, N$.

The partition function in the mean field approximation is found to be
\begin{align}
  Z_{HT}^{\rm mf} = \left( \int_{\mathcal O_1 \times \mathcal O_2} e^{-\beta \left(\frac12 \mathbf \Pi \cdot \mathbf \Omega - \mathbf \Gamma \cdot (\mathbb J \, \Braket{ \boldsymbol \gamma} )\right)} d \mathbf \Pi \, d \mathbf \Gamma \right)^N\, , 
\end{align}
and the equation for $\Braket{\boldsymbol \Gamma}$ simplifies to
\begin{align} \label{mean-position}
\langle \boldsymbol \Gamma \rangle &=
\frac{\int_{\mathcal O_2} \mathbf \Gamma \,e^{\beta \boldsymbol  \Gamma \cdot (\mathbb J \, \Braket{\boldsymbol \gamma})} \left( \int_{\mathcal O_1} e^{-\beta \frac12 \mathbf \Pi \cdot \mathbf \Omega} \, d \mathbf \Pi  \right) d \mathbf \Gamma
}
{\int_{\mathcal O_2} e^{\beta \mathbf \Gamma \cdot (\mathbb J \, \Braket{\boldsymbol  \gamma})} \left( \int_{\mathcal O_1} e^{-\beta \frac12 \mathbf \Pi \cdot \mathbf \Omega} \, d \mathbf \Pi  \right) d \mathbf \Gamma
}.
\end{align}

\begin{remark}  \label{mf-equivalence}
Notice that if one chooses $\mathbb I = \mathrm{diag}(1,1,1)$ so that $\boldsymbol \Omega = \boldsymbol \Pi$, the integral over the $\mathbf \Pi$-variable cancels out so equation \eqref{mean-position} becomes equivalent to \eqref{mean-momentum}.  
\end{remark}

\subsection{Numerical simulations}

The simplest (but computationally expensive) way to detect phase transitions in lattices is by a direct simulation of the full stochastic equation, sampling from the Gibbs measure. 
This is equivalent to a classical Monte-Carlo simulation of the Ising model for example but in the more general setting of continuous spins on coadjoint orbits. 
The phase transitions we observe are all second-order and is detected as we increase the temperature $T = \frac{\sigma^2}{2 \theta}$.
The order parameter for the rigid body and the heavy top network are given by the averaged momentum and position respectively.

\begin{figure}[htpb]
  \centering
  \subfigure[$\mathbb I=\mathrm{diag}(1,1,1)$ and $\mathbb J=\mathrm{diag}(1,2,3)$]{\includegraphics[scale=0.65]{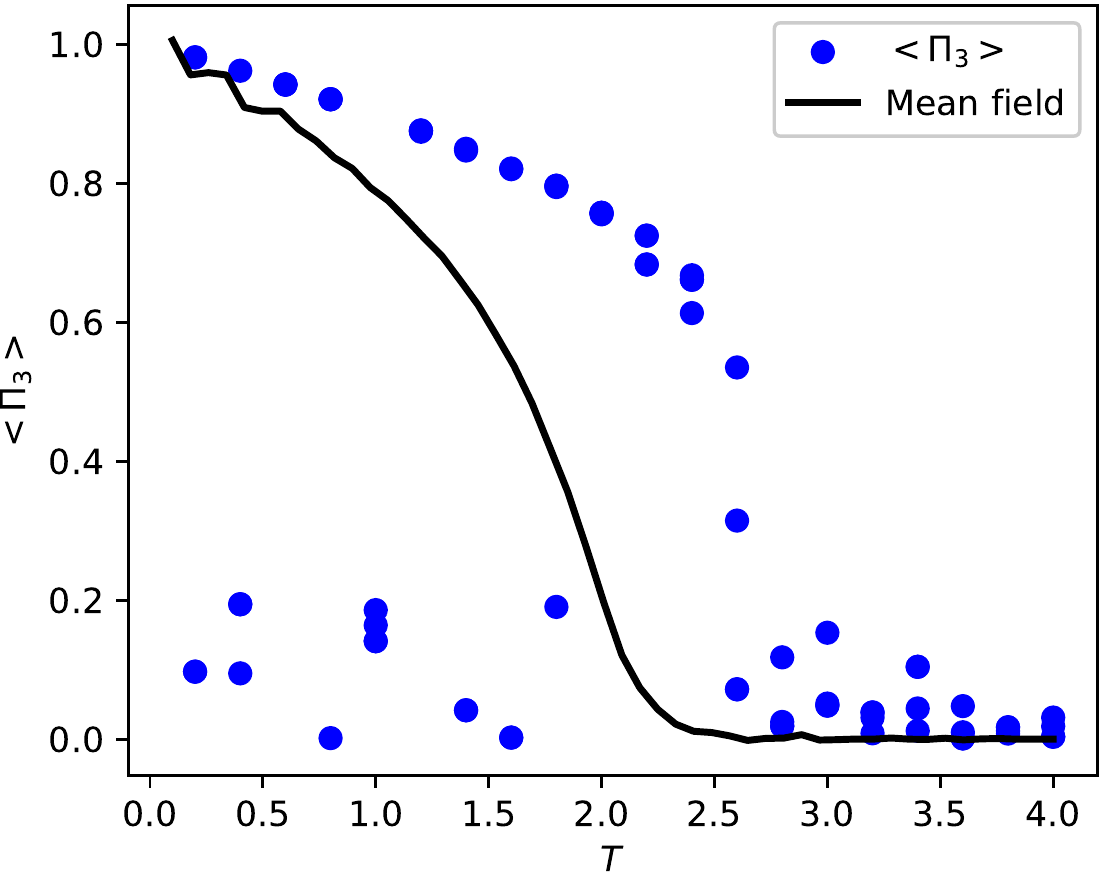}\label{PT-RB1}}
  \subfigure[$\mathbb I=\mathrm{diag}(1,2,3)$ and $\mathbb J=\mathrm{diag}(1,1,1)$]{\includegraphics[scale=0.65]{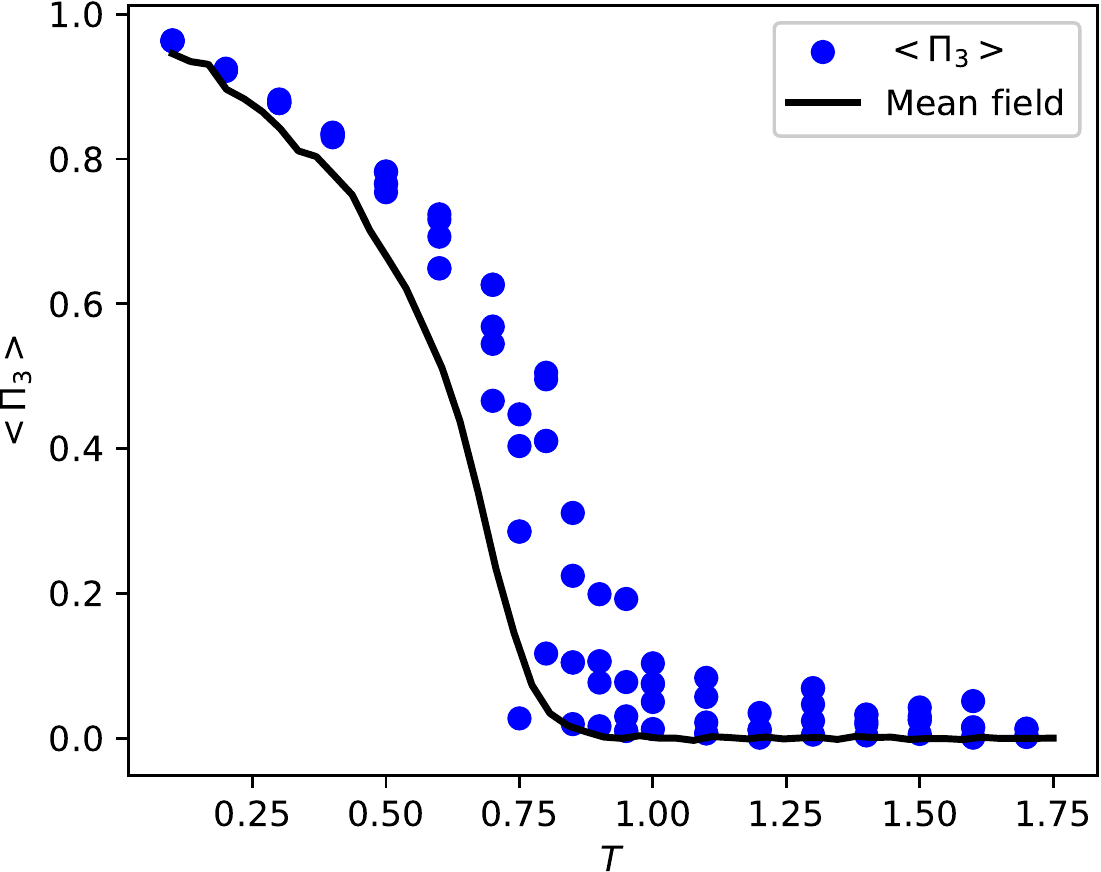}\label{PT-RB2}}
  \caption{The two panels in this figure show the temperature phase transition of the rigid body lattice for two different cases.  The blue dots correspond to averaged magnetisation computed via direct simulation of the full stochastic equation \eqref{SEP-Diss-N} and the black line corresponds to the mean field approximation \eqref{mean-momentum}.
  Both cases exhibit a second order phase transition, characterised by a sudden loss of magnetisation (magnitude of the order parameter) at the critical temperature. }
  \label{fig:RB1}
\end{figure}

We summarise here the phase transition behaviours of the four systems we simulated, displayed in figures \ref{fig:RB1} and \ref{fig:HT-PT}. 
\begin{enumerate}
  \item{\it Rigid body with $\mathbb I=\mathrm{diag}(1,1,1)$ and $\mathbb J=\mathrm{diag}(1,2,3)$ (Figure \ref{PT-RB1}) }
    This case corresponds to the classical Heisenberg model with anisotropic interactions (also known as $XYZ$-model). 
    Notice that some simulations have anomalous low magnetisation at low temperatures. 
    This can be explained by the fact that in some simulations, the system got stuck in a state consisting of large regions of opposite spins.
A longer simulation would be necessary to see these two domains merge into one to reach the minimum energy state.  
  \item {\it  Rigid body with $\mathbb I=\mathrm{diag}(1,2,3)$ and $\mathbb J=\mathrm{diag}(1,1,1)$.(Figure \ref{PT-RB2}) } 
    This case corresponds to a massive isotropic Heisenberg model with non-uniform mass. 
    Our simulations did not get stuck in the state consisting of opposite spins, as seen in the previous case, and the mean field approximation is closer to the direct simulations. 
    The reason that the mean field approximation is more precise in this case can be explained by the fact the interaction term, which is approximated, is isotropic but the kinetic term, which is exact, is anisotropic. 
  \item {\it  Heavy top with $\mathbb I=\mathrm{diag}(1,1,1)$ and $\mathbb J=\mathrm{diag}(1,2,3)$. (Figure \ref{PT-HT1}) } 
    This case also corresponds to the classical Heisenberg model, at least from the mean field point of view (by remark \ref{mf-equivalence}), even if the full dynamics is different. Again, we observe solutions getting stuck in the state consisting of opposite spins and hence the existence of anomalies. 
  \item {\it  Heavy top with $\mathbb I=\mathrm{diag}(1,2,3)$ and $\mathbb J=\mathrm{diag}(1,1,1)$. (Figures \ref{PT-HT2},\ref{PT-HT3} and \ref{PT-HT4})} 
    This last case is the most interesting as it shows a more complex phase transition, which we will describe in details below. 
\end{enumerate}

\subsubsection{Triple-humped phase transition in the heavy top network}
We now discuss the last case in more details. 
From the direct numerical simulations, we observed a phase transition from a strongly magnetised state (large value of $\Braket{\boldsymbol \Gamma}$) along the $\Gamma_3$-axis to a non-magnetized state (that is, $\Braket{\boldsymbol \Gamma} = 0$) as we increased the temperature. But between these two states, we also observed {\em two intermediate phase transitions}, where  magnetisation along the other two axes also occur before becoming completely disordered. We call this a `trimple-humped' phase transition.
These intermediate phase transitions indicate that these unstable ferromagnetic equilibria along the $\Gamma_1$ and $\Gamma_2$ axes can still support magnetisation, which, in statistical physics is generically called a {\em meta-stable state}. We note that this phenomenon is not captured in our mean-field simulations.

From our linear stability analysis in section \ref{section-HT}, we observed that for small values of $\lambda_1$, equivalent to a small ratio $c_1/c_2$, the ferromagnetic equilibria along the $\Gamma_1$ and $\Gamma_2$ axes are unstable, but are close to being linearly stable. 
This is compatible with the observation that  these intermediate phase transitions exist only for small values of $c_1/c_2$. 
For example, in panel \ref{PT-HT4} we took $c_1/c_2 = 1.15$ and saw that the third phase transition is almost negligible. 
Increasing the value of this ratio further will also remove the other intermediate phase transition. On the other hand, decreasing $c_1/c_2$, the intermediate phase transitions persist and furthermore, will shift their respective critical temperature towards zero. 
In figure \ref{PT-HT3}, where we took $c_1/c_2 = 0.8$, we see that both critical temperatures $T_1$ and $T_2$ are smaller compared to those in figure \ref{PT-HT2}, where we took $c_1/c_2 = 1$. 

We will not investigate this phase transition further here, and leave its mathematical understanding as a challenging open problem. 

\begin{figure}[htpb]
  \centering
  \subfigure[$\mathbb I=\mathrm{diag}(1,1,1)$, $\mathbb J=\mathrm{diag}(1,2,3)$, $\frac{c_1}{c_2}=1$]{\includegraphics[scale=0.67]{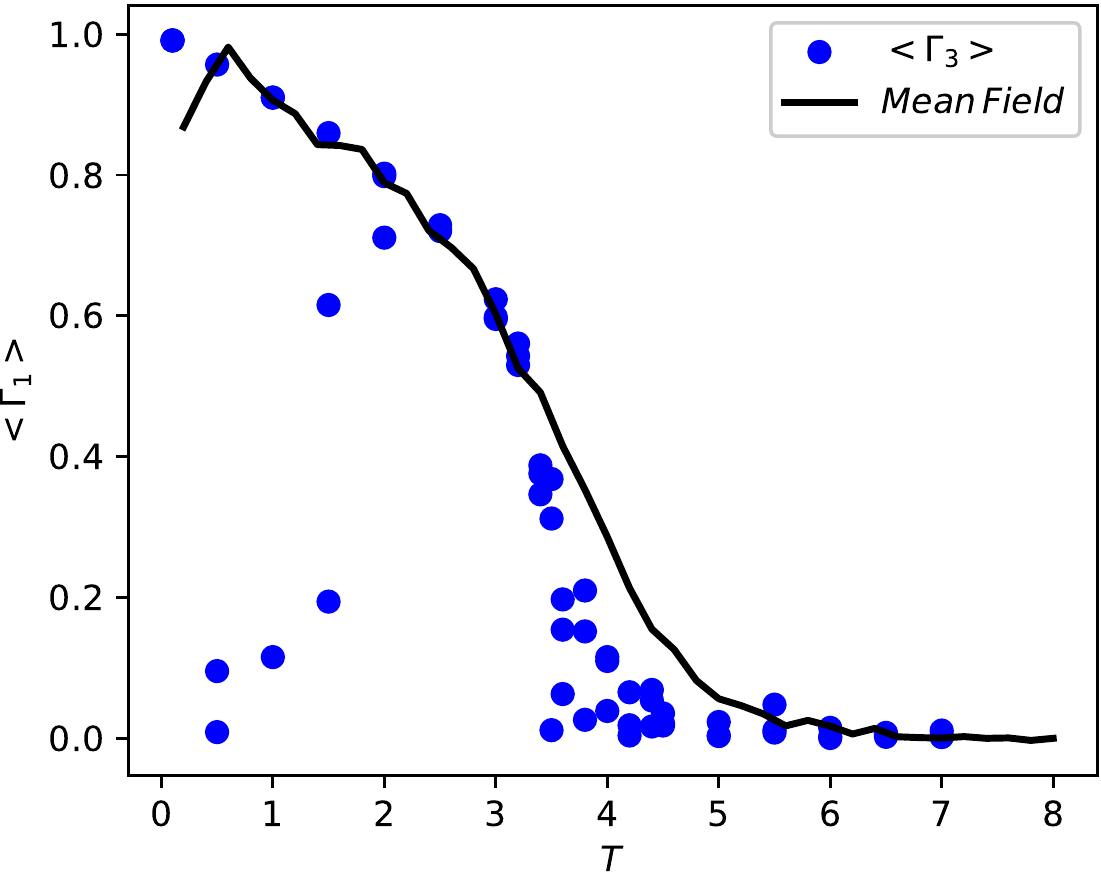}\label{PT-HT1}}
  \subfigure[$\mathbb I=\mathrm{diag}(1,2,3)$, $\mathbb J=\mathrm{diag}(1,1,1)$, $\frac{c_1}{c_2}=1$]{\includegraphics[scale=0.67]{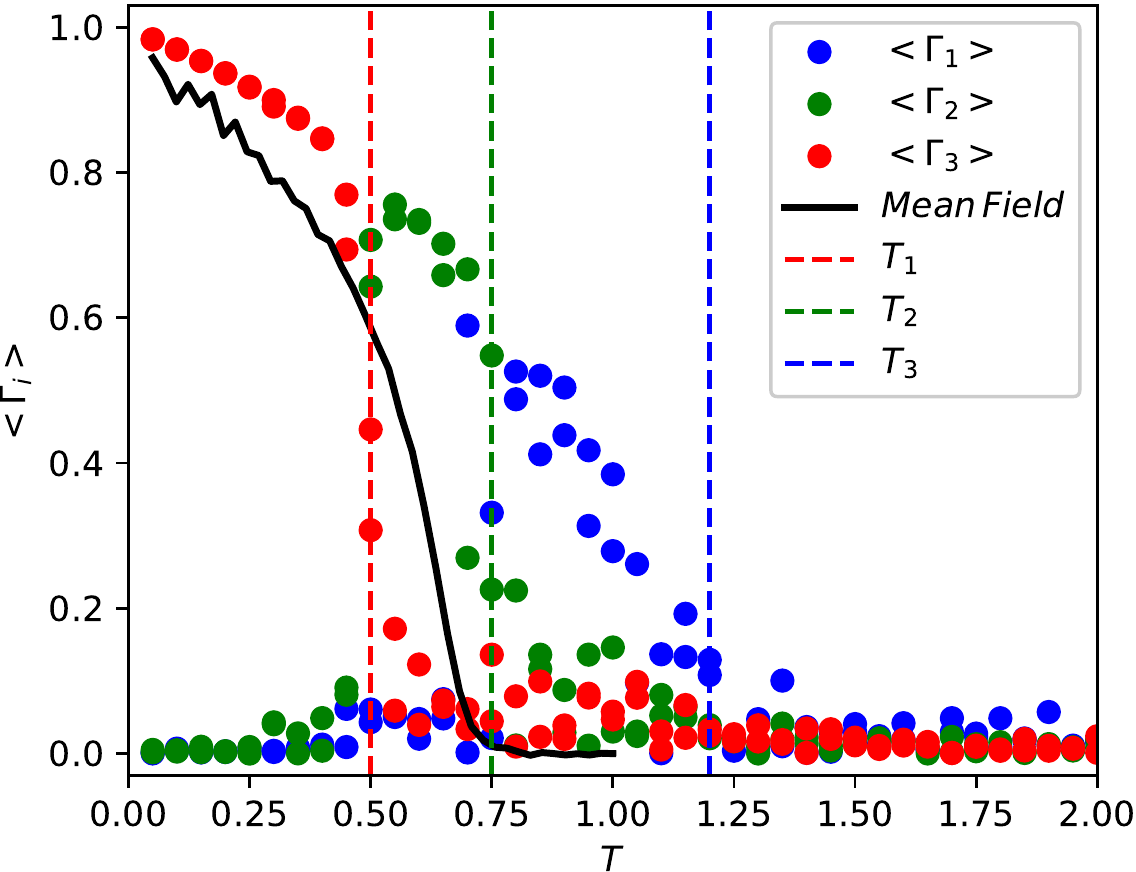}\label{PT-HT2}}
  \subfigure[$\mathbb I=\mathrm{diag}(1,2,3)$, $\mathbb J=\mathrm{diag}(1,1,1)$,$\frac{c_1}{c_2}=0.8$]{\includegraphics[scale=0.67]{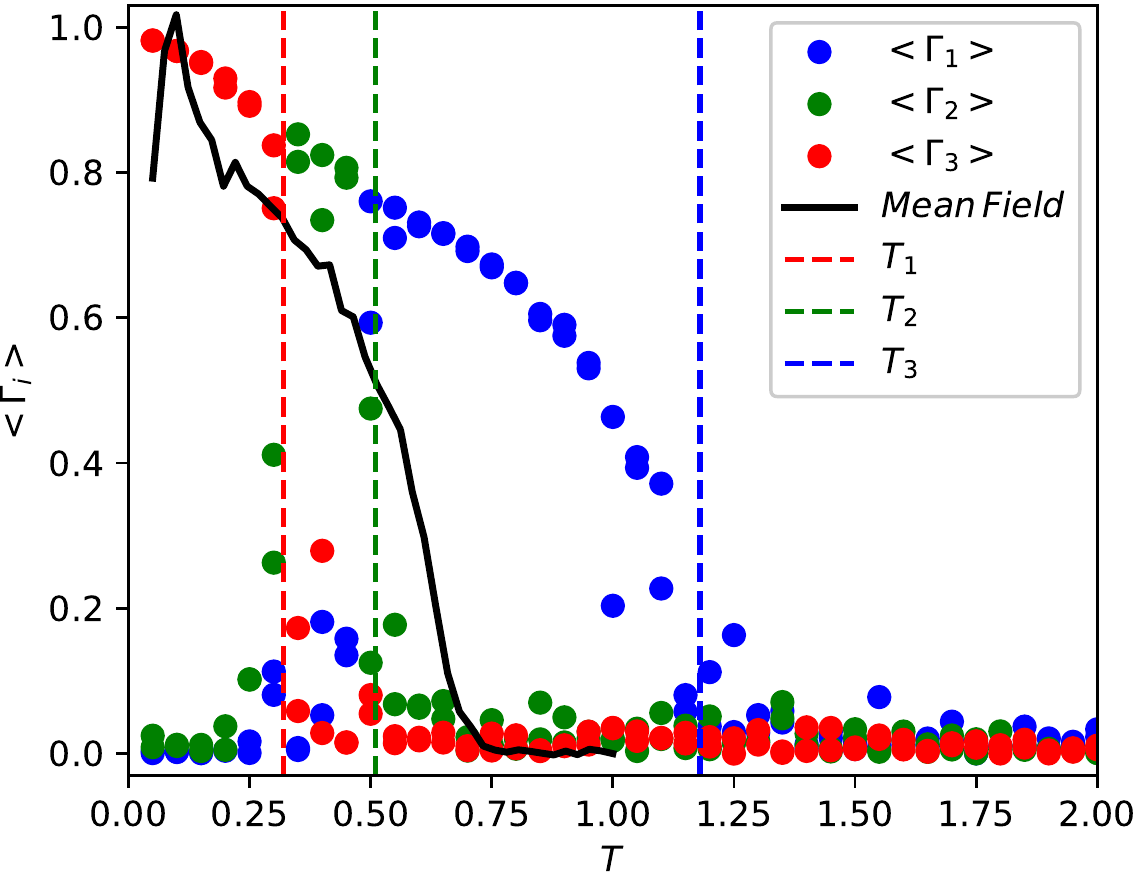}\label{PT-HT3}}
  \subfigure[$\mathbb I=\mathrm{diag}(1,2,3)$, $\mathbb J=\mathrm{diag}(1,1,1)$,$\frac{c_1}{c_2}= 1.15$]{\includegraphics[scale=0.67]{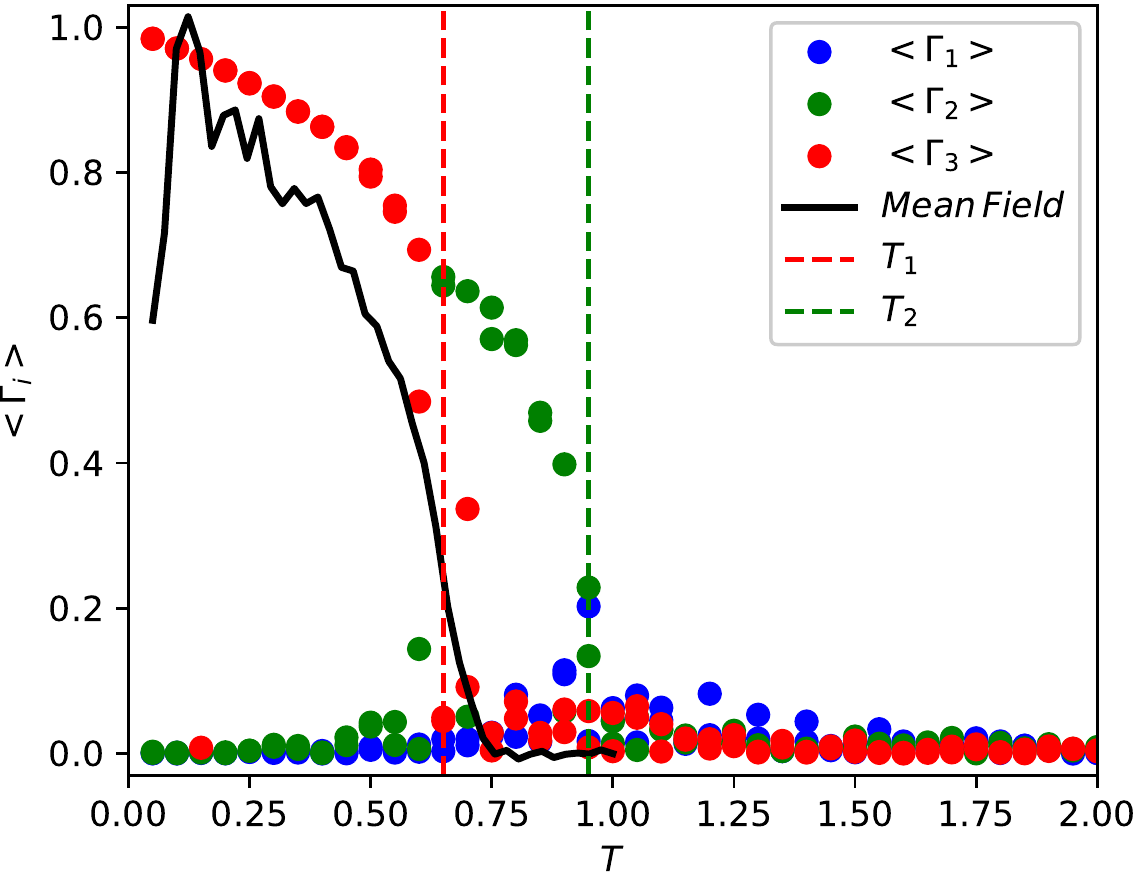}\label{PT-HT4}}
  \caption{The four panels in this figure show phase transitions in the heavy top lattice. In figure \ref{PT-HT1}, the blue dots represent the $\Gamma_3$-components of the averaged position $\Braket{\boldsymbol \Gamma}$ and in \ref{PT-HT2}-\ref{PT-HT4}, the blue, green and red dots represent the $\Gamma_1$, $\Gamma_2$ and $\Gamma_3$-components of the averaged position respectively. The black line in all four figures correspond to the mean field approximation \eqref{mean-position}.
The top left panel \ref{PT-HT1} shows a phase transition similar to the rigid body phase transition, and the other panels \ref{PT-HT2}-\ref{PT-HT4} display a `triple-humped' phase transition. The corresponding critical temperatures are indicated by the vertical dashed lines.
  In \ref{PT-HT2}-\ref{PT-HT4}, we also varied the Casimirs (where $c_1 = \boldsymbol \Pi \cdot \boldsymbol \Gamma$ and $c_1 = \|\boldsymbol \Gamma\|^2$) to show that the three critical temperatures depend on their values, and may even disappear for large enough $c_1/c_2$.  
  Due to the high dimension and non-compactness of the phase space, the mean field approximation requires intensive computation to be precise, explaining the visible errors, especially at low temperatures.
}
  \label{fig:HT-PT}
\end{figure}

\section{Conclusion and outlook} \label{conclusion}

In this paper, we established a link between geometric mechanics and statistical mechanics by constructing a network of interacting Lie-Poisson system on $\mathfrak g^*$ with noise and dissipation that preserves the coadjoint orbits, which gave us a canonical ensemble for the system in statistical equilibrium. For the construction of the system, we considered two types of coupling, one where the neighbours are coupled in the reduced space and the other where the neighbours are coupled directly on the configuration group by considering a representation of the group on a given vector space. The first approach yielded a direct generalisation of the classical Heisenberg model to include general symmetry groups with an additional kinetic energy term and the second approach gave a system that is possibly new. In the special case where $\mathfrak g$ is compact and semi-simple with Hamiltonian of the form kinetic + potential energy, we were able to find the equilibrium solutions of the purely deterministic system as the eigenvectors of the underlying extended graph Laplacian and found that (1) for the momentum-coupled case, the equilibrium solution corresponding to the lowest and highest eigenvalue of the graph Laplacian are nonlinear stable and (2) for the position-coupled case, the equilibrium solution corresponding to the lowest eigenvalue of the graph Laplacian is nonlinearly stable. Furthermore, we showed that in both cases, these equilibrium solutions can be classified into ferromagnetic and anti-ferromagnetic states. In our numerical simulation of the rigid body lattice and the heavy top lattice, which are the simplest examples of momentum-coupled and position-coupled systems respectively, we observed a second order phase transition, similar to that in the Ising model or in the Heisenberg model. However, in the heavy top network, we also observed a `triple-humped' phase transition, in which the system underwent two intermediate phase transitions before settling down to the lowest energy configuration as we decreased the temperature, which is unusual for simple lattice models.

In future work, we would like to investigate further this new type of phase transition behaviour that we observed for the heavy top network, which we believe is related to the metastability of the intermediate ferromagnetic states. However, even without noise, we numerically observed unusual behaviour of the heavy top network, such as when the spins start close to an arbitrary ferromagnetic state, there is an exchange between the kinetic and potential energy that causes the spins to relax to a state that oscillates closely around the stable ferromagnetic state with lowest potential energy, despite the absence of dissipation. 
So studying the deterministic heavy top network further could also be interesting and may help us understand this phase transition behaviour better. 
This partial synchronisation result deserves a more detailed study, in particular for more general networks, where synchronisation of oscillators are an active subject of research (see for example \cite{barahona2002synchronization} and the many subsequent works). 
Other interesting phenomena can be observed for Heisenberg models on certain types of networks, such as supra-oscillations (see for example\cite{expert2017graph} and references therein).  

Regarding phase transitions, one could also compute the critical exponents and the corresponding universality class of the phase transition seen here, or even doing a more thorough analysis using Landau theory to better understand the dynamics near the critical temperature. We can also investigate different types of phase transitions, for instance by varying the external magnetic field instead of temperature, or even extending our domain from a simple lattice to a general network. 
One may also guess that certain topological phase transitions could even be observed in these systems, such as the Kosterlitz-Thouless transition \cite{kosterlitz1973ordering}. 

As we can see, there are many interesting questions that are open for further investigation which we hope to address in future works.

\subsection*{Acknowledgements}
{\small
AA is grateful to John Gibbon for suggesting him to look further into this research direction. 
AA acknowledges EPSRC funding through award EP/N014529/1 via the EPSRC Centre for Mathematics of Precision Healthcare. ST acknowledges funding through Schr\"odinger scholarship scheme.
}

\bibliographystyle{alpha}
\bibliography{biblio}

\end{document}